\documentclass[table,svgnames,dvipsnames]{article}

\usepackage{soul}
\usepackage{url}
\usepackage[hidelinks]{hyperref} 
\usepackage[small]{caption}
\usepackage{graphicx}
\usepackage{booktabs}
\usepackage{amsmath,multirow}
\usepackage{amsthm,thm-restate,thmtools}
\usepackage{amssymb}
\usepackage{enumerate}
\usepackage{xcolor,xfrac}
\usepackage[english]{babel}
\usepackage{charter}
\usepackage{authblk}
\usepackage{framed}
\usepackage[normalem]{ulem}  
\usepackage{tikz}
\usetikzlibrary{shapes.arrows}
\usetikzlibrary{shapes.geometric} 
\usepackage{amsfonts} 
\usepackage[T1]{fontenc}
\usepackage[utf8]{inputenc}
\usepackage[top=1 in,bottom=1in, left=1 in, right=1 in]{geometry}  
\usepackage{enumitem}
\usepackage{tabularx}
\hypersetup{
     colorlinks   = true,
     linkcolor    = red, 
     urlcolor     = teal, 
	 citecolor    = blue 
}
\usepackage{mathtools}
\usepackage{nicefrac} 
\usepackage{pifont}
\usepackage{bm}
\usepackage{dsfont}
\usepackage{algorithm}
\usepackage{graphicx}
\usepackage[noend]{algorithmic}
\usepackage{subcaption}
\usepackage{gensymb}
\usepackage{pgfplots}
\usepackage{bbm}
\usepackage{comment}
\usepackage{xr}
\usepackage[sort,numbers]{natbib}

\usepackage{tkz-graph}

\usepackage{thm-restate}
\usepackage{todonotes}
\usepackage{enumerate}

\newcommand{\I}{\mathcal{I}}
\newcommand{\A}{\mathbf{a}}

\newtheorem{theorem}{Theorem}[section]
\newtheorem{proposition}[theorem]{Proposition}
\newtheorem{lemma}[theorem]{Lemma}

\theoremstyle{definition}
\newtheorem{definition}[theorem]{Definition}

\theoremstyle{remark}

\newtheorem{example}[theorem]{Example}

\DeclareMathOperator{\OPT}{\mathrm{OPT}}

\allowdisplaybreaks

\usepackage{newfloat}
\usepackage{listings}
\DeclareCaptionStyle{ruled}{labelfont=normalfont,labelsep=colon,strut=off} 
\floatstyle{ruled}
\newfloat{listing}{tb}{lst}{}
\floatname{listing}{Listing}

\usepackage{booktabs}

\title{Online House Allocation with Subsidy}

\author[1]{Nicholas Teh}
\author[1]{Saar Cohen}
\author[2]{Karen Frilya Celine}
\author[3]{Benjamin Yu}
\author[1]{\authorcr Michael J. Wooldridge}

\affil[1]{University of Oxford, United Kingdom}
\affil[2]{National University of Singapore, Singapore}
\affil[3]{Carleton University, Canada}
\date{\empty}

\begin{document}

\maketitle

\begin{abstract}
House allocation is a fundamental problem in which each agent is assigned exactly one house. While the classical model assumes that all houses are available before the allocation is computed, many practical settings require decisions to be made as houses become available over time. We introduce the \emph{online house allocation problem}, where houses arrive sequentially and the algorithm must maintain an allocation without knowledge of future arrivals. Unlike online fair division, the one-house-per-agent constraint makes \emph{recourse} an inherent part of the problem, as accepting a newly arrived house may require reassigning previously allocated houses. We study online house allocation under \emph{subsidy-based fairness}, where monetary subsidies eliminate envy among agents. We show that envy-freeability can always be maintained online using bounded recourse and that reassignment chains of length linear in the number of agents are unavoidable in the worst case. In contrast, minimizing the total subsidy is fundamentally harder: no deterministic online algorithm against an adaptive adversary, and no randomized online algorithm against a non-adaptive adversary, admits a bounded competitive ratio, even for two agents and four houses. We complement these impossibilities by showing that exact online subsidy minimization is possible whenever there is at most one extra house beyond the number of agents, and that this guarantee is best possible with respect to the number of extra houses. Finally, we develop learning-augmented algorithms that recover the offline optimum under accurate predictions while providing explicit robustness guarantees when predictions are inaccurate.
\end{abstract}

\section{Introduction}
In many housing systems, apartments do not become available all at once. New units become available over time as buildings are completed, existing tenants move out, or vacant apartments are refurbished. Meanwhile, households on the waiting list often need housing immediately, making it impractical to postpone allocation decisions until every future apartment is known. Housing authorities must therefore assign apartments as they become available, despite having no knowledge of future vacancies. Moreover, when a particularly desirable apartment later enters the system, it may be preferable to reassign it to a household that has already received housing, triggering a sequence of reallocations among previously assigned apartments. Similar challenges arise in university accommodation, refugee resettlement, and other settings where housing resources become available incrementally over time.

The classical abstraction of these applications is the \emph{house allocation problem}, a fundamental model in the economics of resource allocation~\citep{hylland1979positions,zhou1990house,abdulkadirouglu1998random}. In its traditional \emph{offline} form, a set of $m$ indivisible houses is allocated to $n\le m$ agents with subjective preferences, so that each agent receives exactly one house. A central assumption in this literature is that all houses and the agents' preferences are known before the allocation is computed, allowing the algorithm to optimize over the complete instance. In many practical settings, however, this assumption is unrealistic: houses become available sequentially, and allocation decisions often cannot be postponed until the entire housing stock is known. This motivates studying house allocation in an \emph{online} setting, where decisions must be made as new houses arrive.

In this paper, we therefore introduce the \emph{online house allocation problem}, in which houses arrive one at a time. Upon each arrival, the agents' values for the new house are revealed, and the algorithm must decide how to update the current allocation without knowing what houses will arrive in the future. As in the offline problem, each agent ultimately receives \textit{exactly one} house, so an allocation is naturally viewed as a \textit{matching} between agents and houses. Unlike offline house allocation, however, accepting a newly arrived house may require reassigning previously allocated houses, giving rise to a \textit{chain} of reallocations that ends when a house is either assigned to a previously unassigned agent or discarded permanently. Since discarded houses cannot later be recovered, these decisions must carefully balance the current allocation against an uncertain future. We measure the amount of recourse available to an online algorithm by its \emph{flexibility}, namely the maximum length of such a reassignment chain.

Although related to online fair division~\citep{AleksandrovWalsh2020Survey,AleksandrovEtAl2015}, online house allocation presents fundamentally different challenges. The key distinction is the one-house-per-agent constraint: unlike in online fair division, assigning a newly arrived house may require modifying the current matching rather than simply extending an agent's bundle. Consequently, algorithms must reason not only about the arriving house, but also about how to update the existing allocation.

A natural question is what guarantees an online house allocation should satisfy. A central fairness notion for house allocation is \emph{envy-freeness}: no agent should prefer the house assigned to another agent over her own~\citep{foley1967resource,varian1974equity}. Unfortunately, an envy-free allocation need not exist, even in the offline setting. A natural relaxation is to allow \emph{subsidies}, i.e., monetary transfers that compensate agents for any remaining envy~\citep{halpern2019fair,choo2024housesubsidy}. An allocation that can be made envy-free using subsidies is called \emph{envy-freeable}, and among all envy-freeable allocations, one naturally seeks those requiring the least total subsidy. Besides fairness, we also study standard notions of economic efficiency. A central question in this paper is therefore how much \emph{recourse} is needed to maintain fair and efficient allocations under the matching constraint while making irrevocable decisions without knowledge of future arrivals.

\subsection{Our Contributions}

We study online house allocation through two closely related questions:
how much \emph{recourse} is needed to update an allocation as houses
arrive, and how much \emph{subsidy} is needed to ensure fairness. Our
results characterize what is achievable online, identify fundamental
limitations, and show how predictions can overcome some of them.

\paragraph{Maintaining envy-freeability online.}
We show that envy-freeability can always be maintained despite the
online arrival of houses and the matching constraint. We design an
online algorithm that maintains a welfare-maximizing allocation after
every arrival while updating the allocation through a single
reassignment chain. Thus, fairness can be preserved online using only a
simple and natural form of recourse.

\paragraph{The cost of recourse.}
We identify the recourse required for online house allocation. We
show that maintaining envy-freeability requires reassignment chains
linear in the number of agents, and that this bound is asymptotically
tight. For welfare-maximizing allocations, we obtain a tight characterization:
in the worst case, all $n$ agents may need to change assignment. Thus, long reassignment chains are inherent to the
online problem rather than an artifact of our algorithm.

\paragraph{Limits of online subsidy minimization.}
Although envy-freeability can always be maintained online, minimizing
the required subsidy is considerably harder. We show that
welfare-maximizing allocations may require arbitrarily larger subsidies
than subsidy-minimizing ones, and prove that no online algorithm
achieves a bounded competitive ratio for the minimum subsidy, even with
only two agents and four houses. This impossibility persists even if
the algorithm may freely reoptimize the allocation using all currently
available houses, showing that multiplicative competitive analysis is
inadequate for online subsidy minimization.

\paragraph{Positive subsidy guarantees.}
We complement these impossibility results with positive algorithmic
guarantees. We show that exact online subsidy minimization is possible
whenever there is at most one more house than agents, and that this is best possible with respect to the number of extra houses.
Beyond this threshold, exact minimization becomes impossible, but we
obtain additive guarantees, including the optimal additive loss of
$1/2$ for two agents and an additive bound of $(n-1)^2/n$ for $n$ agents with
identical utilities.

\paragraph{Learning-augmented online house allocation.}
Finally, we investigate whether predictions can overcome the
limitations of purely online algorithms. Rather than predicting welfare
values, our prediction models capture the matching structure of house
allocation by identifying the houses that should remain available for
the final allocation. We design algorithms that recover the offline
optimum when these predictions are correct while providing robustness
guarantees when predictions are inaccurate. We complement these results
with lower bounds showing that our guarantees are
nearly optimal.

\subsection{Related Work}
\label{sec:related work}
Due to space limitations, we focus on the closely related work and defer a broader discussion to Appendix~\ref{supp:related}.

\paragraph{House allocation and subsidy-based fairness.}
House allocation is a fundamental problem in which a set of indivisible houses is assigned to agents. Envy-freeness has been studied from both structural and computational perspectives \citep{gan2019envy,kamiyama2021complexity}, while more recent work considers relaxations such as envy minimization \citep{MadathilMisraSethia2025,HosseiniRoySethia2026}. Our fairness notion is closest to subsidy-based fairness. \citet{halpern2019fair} introduced envy-freeability for indivisible goods, and \citet{choo2024housesubsidy} studied minimum-subsidy envy-free house allocation, showing NP-hardness in general together with tractable special cases. We build on this \textit{offline} framework, but present the \textit{online} setting where houses arrive sequentially and the algorithm must make allocation and recourse decisions without knowing future arrivals.

\paragraph{Online matching and recourse.}
Our model is also related to online bipartite matching. Classical online matching requires irrevocable decisions \citep{KarpVaziraniVazirani1990}, while later work allows limited recourse to maintain high-quality matchings \citep{BernsteinHolmRotenberg2019,ShinKimLeeAn2020}. Unlike these works, we study recourse as a means of preserving fairness instead of maximizing the size or weight of a matching.
Fairness has also been studied in online matching \citep{HosseiniHuangIgarashiShah2024,hajiaghayi2024fairness}. These works consider irrevocable allocation of arriving items under matching constraints and seek approximate fairness across agent classes. In contrast, we study individual subsidy-based fairness under cardinal utilities and characterize the recourse needed to maintain it. 

\paragraph{Dynamic house allocation and envy-free matching.}
Prior work also uses the term \emph{online house allocation} for dynamic exchange settings in which agents arrive and depart over time and may trade houses \citep{JalilzadehPlankenDeWeerdt2010,Lesca2025}. Our model instead has a fixed set of agents, exogenously arriving houses, and a subsidy-based fairness objective.
Our work is also related to envy-free matchings and their reconfiguration \citep{AignerHorevSegalHalevi2022,ItoEtAl2025}. Unlike these offline settings, we study online arrivals and characterize the recourse needed to maintain subsidy-based fairness.

\section{Preliminaries}
\label{sec:prelim}
For any positive integer $z$, let $[z] \coloneq \{1,\dots,z\}$, and let
$[0] \coloneq \varnothing$.
We consider an online variant of the \textit{house allocation problem}, in which houses are allocated to agents as they arrive, one at a time.
Formally, an instance consists of a fixed set $N=[n]$ of
\emph{agents}, a finite set $H=\{h_1,\dots,h_m\}$ of $m$  houses, presented in the arrival order $h_1,\dots,h_m$, where $m\ge n$,
and for each agent $i\in N$ a non-negative utility function $u_i:H\cup\{\varnothing\}\to\mathbb{R}_{\ge 0}$ where $u_i(\varnothing)=0$ for each agent $i \in N$. 
However, the value of~$m$, the future houses, and the agents' utilities for future houses are not known to the online algorithm; they are used only in the ex post description of the instance and in the analysis.
In particular, house $h_t$ arrives at round $t \in [m]$, and the utility $u_i(h_t)$ for any agent $i$ is only known from round $t$ onward.
At any round $t \in [m]$, let $H^{t} \coloneq \{ h_1, \dots, h_t \}$ denote the set of houses that have arrived by round $t$; moreover, $H^{0}=\varnothing$. In particular, $H^{m} = H$ is the full set once all houses have arrived. When all agents have identical valuations, we write $u$ for the common
valuation function.

A \emph{(partial) allocation} is a vector
$\A=(a_1,\dots,a_n)\in(H\cup\{\varnothing\})^N$.
If $a_i=h\in H$, then house $h$ is assigned to agent $i$;
if $a_i=\varnothing$, then agent $i$ receives no house.
An allocation $\A$ is \emph{feasible} if no house is assigned to
two distinct agents, i.e., for all distinct $i,j\in N$, if
$a_i\neq\varnothing$ and $a_j\neq\varnothing$, then $a_i\neq a_j$.
An allocation $\A$ is \emph{complete} if
$a_i\neq\varnothing$ for every $i\in N$.
An allocation $\A$ can equivalently be viewed as a matching between $N$ and $H$ given by the edge set $E(\A) \coloneq \{(i,a_i) : i \in N,\ a_i \neq \varnothing\}$.
By setting the weight of edge $(i, a_i)$ to be the utility $u_i(a_i)$, the weight of the matching $E(\A)$ equals the utilitarian welfare $W(\A)\coloneq \sum_{i \in N, a_i \neq \varnothing} u_i(a_i)$.
The \emph{support} of $\A$ is the set of houses assigned by $\A$, i.e., $\mathrm{supp}(\A) \coloneq \{a_i : i\in N,\ a_i\neq\varnothing\}$. 

\paragraph{Adversarial Models.}
As the number of houses $m$ and their arrival order are unknown, we can
view our problem as a game between an adversary and an online
algorithm. We consider the standard adaptive and non-adaptive
adversary models. An \emph{adaptive} adversary observes the algorithm's
past decisions and, at each round $t$, decides whether to terminate the
sequence (after at least $n$ houses have arrived) or introduce a new
house $h_t$ together with its utilities $\{u_i(h_t)\}_{i\in N}$. In
contrast, a \emph{non-adaptive} adversary commits to the entire input
sequence before any of the algorithm's decisions are made. While
adaptive adversaries capture worst-case behavior and often preclude
meaningful fairness guarantees, the non-adaptive model allows such
guarantees and is widely adopted when studying randomized algorithms.

\subsection{Recourse}
\label{sec:recourse}

We investigate how the power available to the online algorithm affects the resulting allocation. We model this power through restrictions on how the allocation may change from one round to the next. Formally, an online algorithm $A$ begins with an initial partial allocation $\A^{0} \coloneq (\varnothing, \dots, \varnothing)$. At each round $t \in [m]$, the algorithm updates the current allocation $\A^{t-1}\in(H^{t-1}\cup{\varnothing})^N$ in response to the arrival of $h_t$ to obtain a new feasible allocation $\A^{t}=(a^{t}_1,\dots,a^{t}_n)\in(H^{t}\cup{\varnothing})^N$.
We denote the final allocation by $\A^{m}=\A$.
Let $D^{t} \coloneq H^{t} \setminus \mathrm{supp}(\A^{t})$ be the set of houses discarded after round $t$.
Any discarded house cannot be allocated in a later round, i.e., $D^{t-1} \subseteq D^{t}$ for every $t \in [m]$.
An online algorithm is \emph{valid} if, for every input sequence with $m\ge n$, its final allocation $\A^{m}$ is complete.
For a randomized online algorithm, validity is required to hold
with probability one over the algorithm's internal randomness
for every such input sequence.
Since the algorithm does not know the stopping time, the input may
terminate after any round $t\ge n$. Hence, validity forces
$D^t=\varnothing$ for every $t\le n$, and it forces
$\A^t$ to be complete for every $t\ge n$. Equivalently, $|\mathrm{supp}(\A^t)|=\min\{t,n\}$ for every $t \in [m]$.

Furthermore, since discarded houses cannot be recovered, when house $h_t$ arrives the only houses that may participate in an update are $h_t$ itself and the houses currently assigned, namely those in $\mathrm{supp}(\A^{t-1})$. Thus, we now formalize the update operations available to the algorithm.

\begin{definition}[Chain update]\label{def:chain-update}
Fix the allocation $\A^{t-1}$ before the arrival of $h_t$. A non-rejection update from $\A^{t-1}$ to $\A^{t}$ is a chain update of length $\ell$ if there exist distinct agents $i_1,\dots,i_\ell\in N$, where $1\le \ell\le n$, such that $a^{t}_{i_1}=h_t$, $a^{t}_{i_r}=a^{t-1}_{i_{r-1}}$ for every $2 \le r \le \ell$, and $a^{t}_j=a^{t-1}_j$ for every $j\notin \{i_1,\dots,i_\ell\}$.
Moreover, $a^{t-1}_{i_r}\neq\varnothing$ for every $r<\ell$. If $a^{t-1}_{i_\ell}\neq\varnothing$, then $a^{t-1}_{i_\ell}$ is discarded; otherwise, the chain terminates at agent~$i_\ell$, who was previously unassigned.
\end{definition}

We then measure the amount of power available to an online algorithm via the notion of $k$-flexibility.

\begin{definition}[$k$-flexibility]\label{def:k-flexibility}
An online algorithm is $0$-flexible if, at every round $t$, it either rejects $h_t$ (i.e., $\A^{t} = \A^{t-1}$), or assigns $h_t$ to an agent $i$ with $a^{t-1}_i=\varnothing$ without changing any previously assigned house. For $k\ge 1$, an online algorithm is $k$-flexible if, at every round $t \in [m]$, it either rejects $h_t$ or performs a chain update of length at most $k$.
\end{definition}
The case $k=0$ is special: the algorithm may fill an unassigned agent's position but may not change any existing assignment. Once all agents are assigned, a $0$-flexible algorithm can only reject new houses.
Our recourse measure is the largest number of agents appearing
in one chain at a single round. It does not bound the total
number of assignment changes over all rounds or the number of
times a particular agent is reassigned. 

\subsection{Fairness Goal: Envy-Freeability via Subsidies}

Our fairness goal is to return a complete allocation that can be made envy-free using \emph{subsidies}. Exact envy-freeness may be unattainable even offline, so we study envy-freeability and the minimum total subsidy required to achieve it. Formally, a subsidy is represented by a non-negative vector $\mathbf{s}=(s_1,\dots,s_n)$, where $s_i$ denotes the payment given to agent $i$. We call a complete allocation $\A$ together with a subsidy vector $\mathbf{s}$ an \emph{outcome} $(\A,\mathbf{s})$. We then define envy-freeness for such outcomes as follows.

\begin{definition}[Envy-freeness with subsidy]\label{def:envy-freeness-with-subsidy}
Let $\A$ be a complete allocation and let $\mathbf{s}=(s_1,\dots,s_n)\in \mathbb R_{\ge 0}^N$ be a subsidy vector. 
The outcome $(\A,\mathbf{s})$ is envy-free if, for every pair of agents $i,j\in N$, $u_i(a_i)+s_i \ge u_i(a_j)+s_j$. 
We say that $\mathbf{s}$ is envy-eliminating for $\A$ if $(\A, \mathbf{s})$ is envy-free.
\end{definition}

Even in the offline setting, not every complete allocation admits an envy-eliminating subsidy vector \citep[Example 2.2]{choo2024housesubsidy}; for completeness, we provide the example in Appendix~\ref{app:example by choo et al}. This motivates the following notion of envy-freeability, introduced by \citet{halpern2019fair} in the context of fairly allocating indivisible goods.

\begin{definition}[Envy-freeability]
A complete allocation $\A$ is \emph{envy-freeable} if there exists a subsidy vector $\mathbf{s}$ such that the outcome $(\A, \mathbf{s})$ is envy-free.\footnote{The allocations $\A^0,\A^1,\dots,\A^m$ denote the sequence of allocations
maintained by the online algorithm and may be revised as houses arrive.
Since the process may terminate after any round, every complete
allocation $\A^t$ with $t\ge n$ is required to be envy-freeable.
Subsidies, however, are evaluated only for the final complete allocation
$\A^m$; no intermediate payments are made or revised.}
\end{definition}

Proposition \ref{prop:hs-characterization} gives a useful characterization of envy-freeable allocations proven by \citet{halpern2019fair}, which is based on the weighted envy graph that encodes the envy relations induced by an allocation.
\begin{definition}[Weighted envy graph]
Given a complete allocation $\A$, its weighted envy graph $G_\A$ is the complete directed graph on vertex set $N$, with one directed edge $i\to j$ for every ordered pair of distinct agents $i,j\in N$. The weight of edge $i\to j$ is $w_\A(i,j)\coloneq u_i(a_j)-u_i(a_i)$. 
The weight of a directed cycle $C=(i_1,i_2,\dots,i_\ell,i_1)$ is $w_\A(C)\coloneq\sum_{r=1}^{\ell} w_\A(i_r,i_{r+1})$, where $i_{\ell+1}=i_1$.
\end{definition}

\begin{proposition}[\protect{\citet[Theorem 1]{halpern2019fair}}] \label{prop:hs-characterization}
Let $\A$ be a complete allocation. The following are equivalent:
(1)~$\A$ is envy-freeable;
(2)~$G_\A$ has no positive-weight directed cycle; and
(3)~$\A$ maximizes the utilitarian welfare over all permutations of its assigned houses, i.e., $\sum_{i\in N} u_i(a_i) 
\ge \sum_{i\in N} u_i(a_{\pi(i)})$ for every permutation $\pi$ of $N$.
\end{proposition}

For a complete allocation $\A$, we define the minimum total subsidy required to make $\A$ envy-free as $\mathrm{sub}(\A) \coloneq \min\{\sum_{i\in N}s_i : s_i\ge 0\ \text{for all }i\in N,\ (\A,\mathbf{s})\text{ is envy-free}\}$.
If no subsidy vector renders $\A$ envy-free, we set $\mathrm{sub}(\A)=\infty$. For an instance $\I$, we denote the minimum total subsidy required by any complete allocation as $\mathrm{OPT}(\I) \coloneq \min\{\mathrm{sub}(\A): \A \text{ is a complete allocation for } \I\}$.
Following standard conventions for online algorithms (see, e.g., \citep{fiat1998online}), we measure the performance of an online algorithm $A$ in terms of its \textit{competitive ratio}. For a deterministic online algorithm $A$, its competitive ratio
with respect to total subsidy is the smallest $\alpha\ge 1$ such that $\mathrm{sub}(A(\I))
\le \alpha\cdot \mathrm{OPT}(\I)$
for every instance $\I$ for which $\mathrm{OPT}(\I)>0$.
For randomized algorithms, the left-hand side is replaced by
$\mathbb E[\mathrm{sub}(A(\I))]$, where the expectation is
over the algorithm's internal randomness.

For any $n$-element set $S \subseteq H$, define $\mu(S) \coloneq \min\{\mathrm{sub}(\A) : \A \text{ is complete and } \mathrm{supp}(\A)=S\}$.
We will use the following fixed-support characterization, which follows from Proposition~\ref{prop:hs-characterization} of this work as well as Proposition~4.2 of \citet{choo2024housesubsidy}.

\begin{restatable}{lemma}{lemfixedsupport}
\label{lem:fixed-support}
Let $S \subseteq H$ be an $n$-element set of houses, and consider the complete bipartite graph $N \times S$ in which edge $(i,h)$ has weight $u_i(h)$.
An allocation $\A$ with support $S$ is envy-freeable if and
only if $E(\A)$ is a maximum-weight perfect matching between $N$ and $S$. Moreover, every maximum-weight perfect matching between $N$ and $S$, together with a minimum envy-eliminating subsidy vector, attains $\mu(S)$.
\end{restatable}

\subsection{Efficiency Solution Concepts}
Proposition \ref{prop:hs-characterization} reveals a close connection between envy-freeability and economic efficiency: every utilitarian welfare-maximizing allocation is envy-freeable. This suggests that, in addition to minimizing the subsidy required to achieve envy-freeness, it is natural to seek allocations that satisfy standard efficiency properties. We will focus on the following notions.

\begin{definition}[Non-wastefulness]
A complete allocation $\A$ is \emph{non-wasteful} if there is no unassigned house
$h\in H\setminus \mathrm{supp}(\A)$ and agent $i\in N$ such that $u_i(h)>u_i(a_i)$.
\end{definition}

\begin{definition}[Pareto optimality]
A complete allocation $\A$ is \emph{Pareto-optimal} if there is no complete
allocation $\A'$ with $\mathrm{supp}(\A')\subseteq H$ such that $u_i(a'_i)\ge u_i(a_i)$ for every $i \in N$, and $u_j(a'_j)>u_j(a_j)$ for some $j \in N$.
\end{definition}

\section{Envy-Freeability} \label{sec:envy-freeability}
In this section, we show that envy-freeability can always be
maintained online. By Proposition~\ref{prop:hs-characterization}, it is sufficient
to maintain a welfare-maximizing allocation. The main difficulty comes from the matching constraint.
In the standard online fair division model, accepting a newly arriving item never forces an already assigned item to leave an agent's bundle.
In contrast, for house allocation, each agent can only hold at most one house, so accepting $h_t$ for an already assigned agent may displace a currently held house, which may in turn have to be reassigned or discarded. 
Our main observation is that this matching constraint can nevertheless be handled without knowing $m$, future houses, or future utilities: after each arrival, it suffices to maintain a welfare-maximizing allocation among the non-discarded houses, and the required recourse can always be chosen to be a single chain update.

\begin{algorithm}[t!]
\caption{Active max-weight chain algorithm}
\begin{algorithmic}[1]
\STATE Initialize $\A^{0} \gets (\varnothing,\dots,\varnothing)$ and $D^{0} \gets \varnothing$. \label{state:init}
\FOR{each arriving house $h_t$}
        \STATE Let $\mathcal{A}_t$ be the set of feasible partial allocations
    $\mathbf{b} \in (H^{t} \cup \{\varnothing\})^N$ such that $|\mathrm{supp}(\mathbf{b})| = \min\{t,n\}$ and $\mathrm{supp}(\mathbf{b}) \cap D^{t-1} = \varnothing$. \label{state:feasible}
    \STATE Choose $\A^{t} \in \arg\max_{\mathbf{b} \in \mathcal{A}_t} \sum_{i \in N} u_i(b_i)$,
    breaking ties by minimizing $|E(\mathbf{b}) \triangle E( \A^{t-1})|$. Any remaining ties are broken by a fixed deterministic order over feasible allocations. \label{state:select}
    \STATE Update $D^{t} \gets H^{t} \setminus \mathrm{supp}(\A^{t})$. \label{state:update}
\ENDFOR
\end{algorithmic}
\label{alg:active-max-weight-chain}
\end{algorithm}

Consider Algorithm~\ref{alg:active-max-weight-chain}. After each round $t$, it maintains a set $D^{t}$ of \emph{discarded houses}. Initially, $D^{0}=\varnothing$ (line~\ref{state:init}). 
When house $h_t$ arrives, in line \ref{state:feasible}, the algorithm considers the set $\mathcal A_t$ of all feasible partial allocations that (1) assign exactly $\min\{t,n\}$ houses, which is the maximum possible while ensuring validity, and (2) use only active houses, i.e., houses in $H^{t}\setminus D^{t-1}$.
Among all allocations in $\mathcal A_t$, the algorithm selects one maximizing the utilitarian welfare. Ties are broken in favor of an allocation closest to the current allocation (line~\ref{state:select}). Formally, for any partial allocation $\mathbf{b}$, let $E(\mathbf{b})\coloneq\{(i,b_i): i\in N,\ b_i\neq \varnothing\}$ denote its matching-edge set. The tie-breaking rule minimizes the symmetric-difference distance $|E(\mathbf{b}) \triangle E(\A^{t-1})|$. Finally, after the new allocation $\A^{t}$ is chosen, every house in $H^{t}\setminus \mathrm{supp}(\A^{t})$ is declared discarded and added to $D^{t}$ (line~\ref{state:update}).
Note that the set $\mathcal A_t$ is nonempty at every round. If $t\le n$, the previous allocation $\A^{t-1}$ can be extended by assigning $h_t$ to an unassigned agent; otherwise, $\A^{t-1} \in \mathcal A_t$. Hence, both the welfare maximization and the tie-breaking rule are well defined.

The key invariant is that, after the arrival of each house, at least one maximum-welfare allocation avoids all previously discarded houses. Lemma~\ref{lem:discarded-remain-discardable} formalizes this property. 

\begin{restatable}{lemma}{lemdiscardedremaindiscardable}
\label{lem:discarded-remain-discardable}    
For $S\subseteq H$, let $F(S)$ denote the maximum utilitarian welfare of a feasible allocation using houses from~$S$ and assigning exactly $\min\{|S|,n\}$ houses. Suppose $D\subseteq S$ and $F(S)=F(S\setminus D)$.
Then, $F(T)=F(T\setminus D)$ for every superset of houses $T \supseteq S$.
\end{restatable}

With this invariant in hand, Theorem~\ref{thm:active-chain} shows that the matching constraint does not prevent maintaining a welfare-maximizing allocation online. Despite the possibility that accepting a new house displaces an already assigned one, every update can be implemented as a single chain update while preserving welfare optimality and thus envy-freeability.

\begin{restatable}{theorem}{thmvalidflexmaxwelfare}
\label{thm:active-chain}
Algorithm~\ref{alg:active-max-weight-chain} is valid and $n$-flexible. For every input sequence, its final allocation maximizes the utilitarian welfare among all complete allocations; consequently, it is envy-freeable.
\end{restatable}

\subsection{Recourse Lower Bounds}

The next proposition establishes a lower bound on the recourse required to guarantee an envy-freeable final allocation. In particular, for $n\in\{2,3\}$, the lower bound matches the $n$-flexible guarantee of Theorem~\ref{thm:active-chain}.

\begin{restatable}{proposition}{strongeroddlowerbound}
\label{prop:stronger-odd-lower-bound}
For every $n \ge 2$, define $\lambda_n \coloneq \left\lceil \frac{n}{2}\right\rceil + 1$.
Against an adaptive adversary, no deterministic valid
$k$-flexible algorithm with $k < \lambda_n$ can always
guarantee an envy-freeable final allocation, even when
$m=n$.
\end{restatable}

Together with Theorem~\ref{thm:active-chain}, Proposition~\ref{prop:stronger-odd-lower-bound} places the
worst-case chain length for deterministic envy-freeability
between $\left\lceil\frac n2\right\rceil+1$ and $n$ against an adaptive adversary. The two values coincide for
$n\in\{2,3\}$. We do not claim a randomized lower bound
or a bound on the total number of reassignments over all
rounds.

\begin{restatable}{proposition}{welfaremaxnlowerbound}
\label{prop:welfare-max-n-lower-bound}
For every $n\ge 2$ and every $k<n$, against an adaptive
adversary, no deterministic valid $k$-flexible algorithm can
always output a final allocation that maximizes utilitarian welfare
among all complete allocations.
\end{restatable}

For maximum welfare, the value $n$ is exact: Algorithm~\ref{alg:active-max-weight-chain}
uses chains of length at most $n$, and Proposition~\ref{prop:welfare-max-n-lower-bound}
shows that no smaller worst-case chain limit suffices for
deterministic algorithms against an adaptive adversary.

\section{Limits of Online Subsidy Minimization} \label{sec:subsidy-limits}
Finding an allocation with minimum envy-eliminating subsidy is already computationally hard in \textit{offline} settings~\citep{choo2024housesubsidy}. Here, we show that even with unlimited computational power, the lack of knowledge about future arrivals fundamentally limits online subsidy minimization. This obstacle already appears in very simple instances with identical utilities, normalized values, and only a small number of extra houses.
In particular, we show that no online algorithm can achieve a bounded multiplicative ratio against the minimum subsidy $\mathrm{OPT}(\I)$, and that efficiency-based requirements do not avoid this difficulty.
Throughout this section, unless stated otherwise, utilities are normalized so that
$u_i(h)\in[0,1]$ for every $i\in N$ and $h\in H$. This is the standard
item-value normalization used for subsidy bounds. Under this
normalization, subsidy values are measured on a common scale. If every utility in an instance is multiplied by a common
factor $U>0$, then both $\mathrm{sub}(\A)$ and
$\mathrm{OPT}(\I)$ are multiplied by $U$. Thus, the
additive bounds stated for utilities in $[0,1]$ scale
linearly under uniform rescaling, while the multiplicative
lower bounds are unchanged.

\subsection{A Universal Subsidy Bound}
\label{subsec:universal-benchmark}

We begin with a universal subsidy bound for envy-freeable
allocations. The bound will be used repeatedly throughout
the remainder of the paper. It follows from the
longest-path subsidy construction of
\citet{halpern2019fair}, specialized to complete house
allocations with normalized utilities.

\begin{restatable}{lemma}{absolutesubsidybound}
\label{lem:absolute-subsidy-bound}
    Let $\A$ be a complete envy-freeable allocation. If $u_i(h)\in[0,1]$ for every $i\in N$ and $h\in H$,  $\mathrm{sub}(\A)\le n-1$.
\end{restatable}

By Theorem~\ref{thm:active-chain}, Algorithm~\ref{alg:active-max-weight-chain} returns an envy-freeable allocation, so Lemma~\ref{lem:absolute-subsidy-bound} implies a final total subsidy of at most $n-1$. 

\subsection{Incompatibility of High Welfare and Minimum Subsidy}
The universal bound of Lemma~\ref{lem:absolute-subsidy-bound}
does not compare an online algorithm with the instance-wise
optimum $\OPT(\I)$. We therefore ask whether stronger
guarantees are possible. Since
Algorithm~\ref{alg:active-max-weight-chain} maximizes
utilitarian welfare, we study the
relationship between welfare maximization and subsidy
minimization.
Recall that $W(\A)$ denotes the utilitarian welfare of $\A$.
For an instance $\I$, let $W^*(\I)\coloneq
\max\{W(\A):\A\text{ is a complete allocation for }\I\}$.
When $m=n$, any complete allocation uses all houses.
By Lemma~\ref{lem:fixed-support}, any maximum-weight perfect matching, together with a minimum envy-eliminating subsidy vector, is subsidy-optimal. Thus, Algorithm~\ref{alg:active-max-weight-chain} is subsidy-optimal when $m=n$. 

The next result shows that one extra house already separates welfare maximization from subsidy minimization, even when all agents have the same values.

\begin{restatable}{theorem}{efficiencysubsidyincompatible}
\label{thm:efficiency-subsidy-incompatible}
Fix $n \ge 2$, $\rho \in (0,1]$, and $R>0$. There
is an identical-utilities instance $\I$ with $m=n+1$, values in $[0,1]$, and $\OPT(\I)>0$, such that any feasible
complete allocation $\A$ with $W(\A) \ge \rho W^*(\I)$ satisfies $\mathrm{sub}(\A) > R \cdot \mathrm{OPT}(\I)$.
\end{restatable}

Taking $\rho=1$, Theorem~\ref{thm:efficiency-subsidy-incompatible}
applies to Algorithm~\ref{alg:active-max-weight-chain}. The hard
instances of Theorem~\ref{thm:efficiency-subsidy-incompatible} further
show that non-wastefulness and Pareto optimality fail to guarantee a
bounded subsidy ratio: any allocation satisfying either property has an arbitrarily large subsidy ratio.

\subsection{Unbounded Competitive Ratio}
\label{subsec:no-competitive-ratio}

Theorem~\ref{thm:efficiency-subsidy-incompatible} identifies an incompatibility between welfare-oriented objectives and subsidy minimization: high welfare can already force arbitrarily high subsidy relative to $\OPT(\I)$. This still leaves open the possibility that an online algorithm designed solely to minimize subsidy could achieve a bounded multiplicative guarantee. Our next theorem rules out this possibility, even for two agents with identical utilities.
For this theorem only, we prove the lower bound even under the following
permissive update rule. After $h_t$ arrives, the algorithm may choose
any feasible allocation $\A^t$ satisfying $\mathrm{supp}(\A^t) \subseteq \mathrm{supp}(\A^{t-1}) \cup \{h_t\}$.
As in Section~\ref{sec:recourse}, every house in
$H^t \setminus \mathrm{supp}(\A^t)$ is then permanently discarded.

\begin{restatable}{theorem}{noonlinecompetitiveratio}
\label{thm:no-online-competitive-ratio}
Even for $n=2$, $m=4$, and identical utilities in $[0,1]$,
no online algorithm admits a bounded competitive ratio for
total subsidy. This holds for deterministic algorithms
against adaptive adversaries and randomized algorithms
against non-adaptive adversaries, even under the more
permissive update rule above.
\end{restatable}

Thus, multiplicative approximation guarantees are
fundamentally unsuitable for online subsidy minimization.
We therefore turn to additive loss and to
settings in which exact minimization remains possible.

\section{Positive Subsidy Guarantees}
\label{sec:positive-subsidy}
We now give positive results for online subsidy minimization. We first
show that, for arbitrary non-negative utilities, exact minimization is
possible whenever \(m\in\{n,n+1\}\). For identical utilities in
\([0,1]\) and any number of arrivals, we then obtain additive
guarantees: we exactly solve the two-agent case with optimal
deterministic additive loss \(1/2\), and we give a deterministic valid
\(1\)-flexible algorithm with final subsidy at most \((n-1)^2/n\) for
arbitrary \(n\).

\paragraph{Exact Minimization with at Most One Extra House.} The next result gives a deterministic online algorithm that is valid for every input sequence. 
Its exact optimality guarantee applies only when the sequence has length $n$ or $n+1$, without knowing in advance which case holds.

\begin{restatable}{theorem}{oneextrahouse}
\label{thm:one-extra-house}
    There is a deterministic polynomial-time $n$-flexible
    online algorithm that is valid on every input sequence and is
    subsidy-optimal whenever $m\in\{n,n+1\}$. In particular,
    for every such instance, $\mathrm{sub}(\A^m)=\mathrm{OPT}(\I)$.
    The algorithm does not need to know in advance whether the
    sequence ends after $n$ or $n+1$ houses.
\end{restatable}
This guarantee is the best possible with respect to the number of extra
houses. By Theorem~\ref{thm:no-online-competitive-ratio}, once two
extra houses may arrive (already for $n=2$ and $m=4$), no
deterministic valid online algorithm admits a bounded competitive ratio
against an adaptive adversary, even for identical normalized utilities.
The randomized lower bound gives the corresponding impossibility
against a non-adaptive adversary.

\paragraph{Identical Utilities.}
Once the model allows two extra houses, exact minimization cannot be
guaranteed uniformly over all numbers of agents, so we turn to additive
guarantees. Appendix~\ref{app:identical} completely resolves the
two-agent case with identical normalized utilities: the optimal
deterministic additive loss against an adaptive adversary is exactly
\(1/2\), attained by a valid \(1\)-flexible algorithm
(Theorem~\ref{thm:identical-two}).

For arbitrary \(n\), Theorem~\ref{thm:identical-general} gives a deterministic valid \(1\)-flexible algorithm that is subsidy-optimal when \(m=n\) and has final subsidy at most \((n-1)^2/n\) whenever \(m\ge n+1\). Consequently, its additive loss is at most
\((n-1)^2/n\). Proposition~\ref{prop:selection-bound-lower} shows that the absolute subsidy constant \((n-1)^2/n\) in the underlying one-step selection bound is best possible. This does not determine the optimal additive loss for \(n\ge3\), which remains open.

\section{Algorithms with Predictions} \label{sec:learning-augmented}

Sections~\ref{sec:subsidy-limits} and~\ref{sec:positive-subsidy}
show that purely online subsidy minimization has a sharp limitation:
exact minimization is possible with at most one extra house, but no
online algorithm admits a bounded competitive ratio once two or more
extra houses may arrive. We therefore ask whether predictions can
recover exact optimality when correct, while remaining robust to
prediction errors.
Unlike prior learning-augmented allocation work, which often
predicts welfare values or aggregate statistics
\citep{banerjee2022online,cohen2023general,neoh2026online,choo2026approximate},
our predictions are tailored to the matching structure of house
allocation: they identify the houses that should remain available for
the final allocation; see Appendix~\ref{sec:Semi-Online Fair Division}
for more details.

We study two prediction models: (1) score predictions, which rank
houses for arbitrary utilities, and (2) cutoff predictions, which
predict a single house-quality threshold for common-quality utilities.
In both models, correct predictions recover \(\mathrm{OPT}(\I)\), while
the bounds given below quantify the additional subsidy incurred from
prediction errors.

\subsection{Score Predictions}

A score prediction assigns a real number $p(h)$ to each house $h$ when it
arrives. We use a fixed deterministic tie-breaking rule that does not depend
on the round. Thus, for every prefix $H^t$, the scores together with the
tie-breaking rule induce a total order on the houses in $H^t$, where houses
with larger scores are ranked higher. Let $B_0\coloneq\varnothing$, and for every
$t\in[m]$, let $B_t$ denote the set consisting of the
$\min\{t,n\}$ highest-ranked houses in $H^t$ under this order. 
Intuitively, $B_t$ contains the houses that the prediction currently
identifies as the most promising ones to remain available for the
final matching.

At the final round, $B_m$ contains exactly $n$ houses.
We measure the quality of the final score-selected house set
by $\eta_{\mathrm{sc}}(p;\I) := \mu(B_m)-\mathrm{OPT}(\I)$,
i.e., the additional subsidy incurred by $B_m$ instead of a minimum-subsidy $n$-house set.
We call the scores \textit{correct} when $\eta_{\mathrm{sc}}(p;\I)=0$, i.e., when $B_m$ is a minimum-subsidy $n$-house set.

Given the predicted support $B_t$ at each round $t$, the natural online algorithm is the \textit{score-keeping algorithm}, which maintains, after each arrival, a maximum-welfare allocation among all allocations whose support is exactly $B_t$; it uses the same tie-breaking
rule as Algorithm~\ref{alg:active-max-weight-chain}. That is, the
algorithm behaves like Algorithm~\ref{alg:active-max-weight-chain}, but with
the support at round $t$ prescribed by the scores rather than determined
by the algorithm. The next theorem shows that if the prediction correctly
identifies the houses that should remain available for the final matching,
then the score-keeping algorithm recovers the offline optimum.

\begin{restatable}{theorem}{scoreprediction}
\label{thm:score-prediction}
For every score prediction, the score-keeping algorithm is
deterministic, valid, and $n$-flexible.  Its final allocation
$\A^m$ attains the minimum subsidy possible for the final
score-selected house set: $\mathrm{sub}(\A^m)=\mu(B_m)$.
Consequently, $\mathrm{sub}(\A^m)
=
\mathrm{OPT}(\I)+\eta_{\mathrm{sc}}(p;\I)$.
If $u_i(h)\in[0,1]$ for every $i\in N$ and $h\in H$, then $\mathrm{sub}(\A^m)
\le
\min\left\{
\mathrm{OPT}(\I)+\eta_{\mathrm{sc}}(p;\I),
\,n-1
\right\}$.
\end{restatable}

Theorem~\ref{thm:score-prediction} gives the learning-augmented guarantee for arbitrary utilities. It shows that correct scores give exact subsidy minimization, while incorrect scores incur additive loss exactly $\eta_{\mathrm{sc}}(p;\I)$, and under normalized utilities the subsidy is never above $n-1$.
The exact guarantee is also unchanged under sufficiently small perturbations of the score values.
Let $p^*$ be a correct score prediction, and let
$B_m^*$ be its final top $n$ houses. Define $\Gamma(p^*) := \min_{{h\in B_m^*,g\in H\setminus B_m^*}} \left( p^*(h)-p^*(g) \right)$,
with the convention that $\Gamma(p^*)=+\infty$ when $m=n$. If
$\Gamma(p^*)>0$, the final ranking is unchanged under sufficiently
small perturbations. In particular, if $\max_{h\in H}|p(h)-p^*(h)|<\frac{\Gamma(p^*)}{2}$,
then $p$ and $p^*$ induce the same final top $n$ houses, and
hence $\mathrm{sub}(\A^m)=\mathrm{OPT}(\I)$. No positive
perturbation guarantee is implied when $\Gamma(p^*)=0$.

\paragraph{Cutoff Predictions for Common-Quality Utilities.}
Score predictions are general, but they require a full ranking of the houses.
Suppose
$u_i(h)=\beta_i+\alpha_iq_h$, where the quality $q_h$ of an arriving house is observed when $h$
arrives and
$0\le\alpha_1\le\cdots\le\alpha_n$.
For this class, the offline support-selection problem is one-dimensional:
Appendix~\ref{app:cutoff} shows that a minimum-subsidy choice of
$n$ houses can always be taken as $n$ consecutive houses in quality
order. In particular, an optimal support is determined by its largest
quality, so a single predicted cutoff suffices. Given a predicted
cutoff $\tau$, the cutoff algorithm assigns each house the score
$p_\tau(h)=q_h$ if $q_h\le\tau$ and $-1$ otherwise, and then runs the
score-keeping algorithm with these scores. Let $\tau^*$ be a correct
cutoff, and let $\tau\ge\tau^*$ be the cutoff used by the algorithm.
Theorem~\ref{thm:predicted-cutoff} gives
$\mathrm{sub}(\A^m) \le \min\{ \mathrm{OPT}(\I) +(n-1)\alpha_{n-1}(\tau-\tau^*), \, n-1 \}$.
Thus, a correct cutoff gives $\OPT(\I)$. If
$\lvert\widehat{\tau}-\tau^*\rvert\le\varepsilon$, using
$\widehat{\tau}+\varepsilon$ gives an added term of at most
$2(n-1)\alpha_{n-1}\varepsilon$.
Appendix~\ref{app:cutoff} also shows that the linear term cannot be
reduced in general and that a small underestimate can cause additive
loss arbitrarily close to $n-1$.

\subsection{Limits of Prediction-Based Guarantees}
\label{subsec:prediction-limits}

The guarantees above are nearly best possible for deterministic algorithms. We show that achieving the score-prediction guarantee may require a full $n$-step update, and that exact optimality under correct predictions does not improve the worst-case loss under incorrect predictions. Appendix~\ref{sec:underestimate-costly} further shows that, for cutoff predictions, even a small underestimate can incur loss arbitrarily close to $n-1$.

\begin{restatable}{theorem}{predictedsupportrecoursetight}
\label{thm:predicted-support-recourse-tight}
Against an adaptive adversary, for every $n \ge 2$ and every $k<n$, no deterministic valid
$k$-flexible online algorithm, even when it is given the score
$p(h_t)$ at the arrival of each house $h_t$, can ensure $\mathrm{sub}(\A^m)=\mu(B_m)$ for every input sequence and every score prediction. This remains true when $m=n+1$.
\end{restatable}

The second limitation concerns incorrect predictions. For cutoff predictions, we say that a deterministic algorithm is exact on correct cutoffs if, whenever the predicted cutoff is correct and $C(\tau)$ is the unique $n$-house set minimizing $\mu$, the algorithm returns a final allocation $\A^m$ with $\mathrm{sub}(\A^m)=\mathrm{OPT}(\I)$.
For score predictions, define exactness on correct scores analogously: when the final top $n$ houses form the unique $n$-house set minimizing $\mu$, the algorithm must return an allocation with subsidy $\mathrm{OPT}(\I)$. Now, we turn to the cost of incorrect predictions.

\begin{restatable}{theorem}{wrongpredictionloss}
\label{thm:wrong-prediction-loss}
Against an adaptive adversary, fix $n\ge 2$. Let $A$ be a deterministic valid prediction-augmented algorithm that is exact on correct predictions, either in the cutoff model or in the score model. For every $\gamma>0$, there is an identical-utilities instance $\I$ with values in $[0,1]$ and an incorrect prediction $\pi$ in the same model such that $\mathrm{sub}(A(\I,\pi))-\mathrm{OPT}(\I)\ge n-1-\gamma$.
\end{restatable}

Thus, Theorems~\ref{thm:score-prediction}--\ref{thm:wrong-prediction-loss}
characterize the power and limitations of predictions for online subsidy minimization.
For arbitrary utilities, score predictions recover the offline optimum under accurate predictions and incur additive loss proportional to the prediction error. For common-quality utilities, a single cutoff prediction suffices, with guarantees degrading gracefully as the prediction error increases.

The lower bounds show that these guarantees are near-best possible for deterministic algorithms. Matching the score prediction exactly may require $n$-flexibility, and exactness on correct predictions cannot improve the worst-case $n-1$ additive bound when predictions are arbitrary. Thus, predictions are most useful when they identify the right houses to keep, while the normalized $n-1$ subsidy bound remains the right guarantee when they do not.

\section{Conclusion and Future Work}

We introduced \emph{online house allocation}, where houses arrive sequentially and the algorithm must maintain an allocation without knowledge of future arrivals. Our results show that the matching constraint makes recourse an intrinsic part of the problem: envy-freeability can always be maintained, but long reassignment chains are unavoidable in the worst case. At the same time, exact subsidy minimization is possible with at most one extra house, yet no online algorithm admits a bounded competitive ratio once two or more extra houses may arrive. Beyond this threshold, we showed that additive guarantees remain possible for identical utilities, and that predictions tailored to the matching structure of house allocation recover the offline optimum when correct while remaining robust to inaccurate predictions.

Several directions remain open. The most immediate is to determine the optimal deterministic additive guarantee for identical utilities when \(n\ge3\). It would also be interesting to sharpen the randomized lower bounds for additive subsidy minimization and to explore alternative prediction models. Finally, it would be valuable to study other fairness notions and alternative recourse measures, such as bounds on total reassignment or on the number of times an individual agent may be moved.

\bibliographystyle{plainnat} 
\bibliography{abb,bib}

\clearpage
\appendix

\section{Additional Related Work}
\label{supp:related}

\subsection{Online Fair Division}

Online fair division studies allocation problems in which goods or
agents are not all available at once; see the survey of
\citet{AleksandrovWalsh2020Survey}. A canonical online model,
motivated by applications such as food banks, was introduced by
\citet{AleksandrovEtAl2015}, where items arrive sequentially and must
be allocated without knowledge of future arrivals.

Most closely related to our work is the recent paper of
\citet{KulkarniMehtaNarayanPonitka2025}, who study online fair division
with subsidy. In their model, indivisible items arrive online and are
immediately and irrevocably allocated to offline agents; the goal is
to maintain envy-freeability while minimizing the required subsidy.
Their results show that envy-freeability can be maintained for additive
and several more general valuation classes, although the required
subsidy may be much larger online than offline. Our setting differs
fundamentally: each agent ultimately receives a single house rather
than a bundle. Consequently, accepting a newly arriving house may
require replacing or reassigning previously allocated houses, making
recourse an intrinsic aspect of the problem.

Earlier work on online fair division has studied binary additive
valuations and applications to food banks
\citep{AleksandrovEtAl2015}, vanishing envy under general additive
valuations \citep{benade2018make}, the compatibility of fairness and
Pareto efficiency under increasingly powerful adversaries
\citep{zeng2020fairness,benade2025dynamic}, and egalitarian welfare
maximization \citep{springer2022online}. More recent work obtains stronger guarantees by imposing structure on valuations or by relaxing immediate allocation. Amanatidis et al.~\cite{amanatidis2025onlineFD-2valued} study personalized two-value instances under immediate and irrevocable allocation, and also show how limited lookahead enables matching-based fairness guarantees. In a different direction, Amanatidis et al.~\cite{amanatidis2026buffers} allow a bounded reordering buffer, which delays the allocation of a small number of items and gives us strong envy-freeness guarantees for personalized $k$-value instances.
These models all allocate arriving items irrevocably into bundles and study fairness or
efficiency of the resulting bundle allocation. Their additional power comes from restrictions on valuations, advance information, or delayed allocation. In contrast, our model is a one-house-per-agent matching problem with arbitrary nonnegative cardinal utilities. The matching constraint fundamentally changes the online problem by making reassignment unavoidable and motivating our study of recourse alongside subsidy-based fairness.

\subsection{Semi-Online Fair Division} 
\label{sec:Semi-Online Fair Division}
There is also a growing body of work on semi-online fair division, also known as online fair division with predictions, where algorithms have access to a priori (side) information about future items that is, e.g., learned from historical data. \citet{banerjee2022online} study online Nash social welfare for
divisible goods.  Without additional information, their
worst-case guarantee scales linearly with the number of
agents.  They then assume predictions of each agent's total
value for all arriving goods.  With correct predictions,
their algorithm obtains logarithmic guarantees in the number
of agents and in the number of arrivals, and they prove
nearly matching logarithmic lower bounds.  Their predictions
concern aggregate future value, whereas our predictions
identify houses that should remain available for a final
one-house-per-agent matching.

Building on their work, other studies consider \textit{normalized} valuations (i.e., each agent's valuation of the entire set of goods is equal to one); specifically, they note that considering normalized valuations is equivalent to assuming access to the sum of agents’ values over all future items. \citet{barman2022universal} explore online resource allocations with the goal of maximizing the generalized means, covering a spectrum of welfare functions, like average utilitarian welfare, egalitarian social welfare and Nash social welfare. They provide tight (up to poly-logarithmic factors) approximation guarantees, which have been recently improved by \citet{huang2025long}. Online fair division with predictions of \textit{divisible} goods has also been studied in other works (e.g., \citep{cohen2023general,an2024best}). 

In contrast, \citet{zhou2023multi} study the allocation of \textit{indivisible} goods and chores among agents with {normalized} valuations, which induce positive and negative utilities, respectively. Given $n$ agents, they analyze the fairness notion of maximin fair share (MMS), which requires that each agent receives a bundle of items she values at least as much as she would have obtained if she were allowed to partition all items into $n$ bundles and then get the least valuable bundle (according to her own valuation). \citet{neoh2026online} build on the work by \citet{zhou2023multi}, investigating how partial knowledge of future items can facilitate the design of online algorithms that satisfy certain fairness notions, including MMS and EF1. They examine two information models: (1) algorithms know the total sum of valuations (effectively assuming normalized utilities); and (2) algorithms have access to frequency predictions, i.e., for each agent and each value, a predictor specifies the frequency of this value among the agents’ valuations for arriving items. Several other works consider the model with indivisible items \citep{spaeh2023online,balkanski2023strategyproof,cohen2024plant}. 

In a related work, \citet{choo2026approximate} study approximate proportionality up to one good (PROP1) in online fair division of indivisible goods. They show that three natural greedy algorithms fail to guarantee any positive approximation to PROP1 against adaptive adversaries. Given this hardness result, they study non-adaptive adversaries and the use of side-information, in the spirit of learning-augmented algorithms. In particular, against non-adaptive adversaries, \citet{choo2026approximate} prove that the uniformly random allocation rule achieves an $\Omega(1/\log(n/\delta))$-approximation to PROP1 with probability at least $1-\delta$, for any number of agents $n\ge 2$ and failure probability $\delta>0$. Afterwards, rather than assuming access to normalization information as done by \citet{barman2022universal}, \citet{huang2025long} and \citet{neoh2026online}, \citet{choo2026approximate} assume only access to maximum item value (MIV) predictions, i.e., for each agent, we are given their maximum valuation over any single item. Given such predictions, they obtain a $1/n$-PROP1 guarantee against adaptive adversaries. \citet{choo2026approximate} also show that stronger fairness notions such as EF1, MMS, and PROPX remain inapproximable even with perfect MIV predictions.
In contrast, our model is not a bundle allocation problem. Each agent can hold at most one house, and a newly arrived house may force the algorithm to replace, discard, or reassign previously allocated houses. Thus, our guarantees concern envy-freeability via subsidies under an online matching constraint, rather than PROP1, EF1, MMS, or generalized-mean welfare guarantees for bundles.

Across these approaches, side information typically summarizes future values through totals, normalization, frequencies, or maxima. Our predictions are structural instead: they identify the houses that should remain available for the final matching, rather than estimating the minimum subsidy $\mathrm{OPT}(\I)$. This distinction is specific to unit demand. Since an agent cannot retain an old house when receiving a new one, the prediction guides which houses remain active, not merely how an expanding collection of items should be distributed among bundles.

\citet{wang2026online} study semi-online allocation of
indivisible goods and chores under a different information
model.  Some of their guarantees for two agents assume that
the complete collection of possible utility values is known
before arrivals begin.  In our model, only the values of the
current house are revealed, and our deterministic guarantees
are stated against an adaptive adversary.  The two models
therefore provide different forms of future information and
are not directly comparable.

\subsection{Temporal Fair Division} 
\citet{elkind2025temporal}, \citet{cookson2025temporal}, and \citet{choi2026tfdmm} study temporal fair division, where allocations are evaluated across multiple prefixes or time periods and information about future items is available in advance. Their focus is on fairness guarantees that hold repeatedly over time. Goldberg et al.~\cite{Goldberg2026} study the complementary objective of minimizing cumulative envy over a sequence of allocations, thereby measuring the total fairness loss accrued across time.

Our model also produces a sequence of allocations, but the role of time is different. Future houses and utilities are unknown, and each update may revise the current matching through a single reassignment chain while permanently discarding houses that leave the active set. Theorem~\ref{thm:active-chain} gives a welfare guarantee after every arrival, whereas subsidies are assessed only for the final complete allocation. Thus, our information model, recourse constraint, and terminal subsidy objective differ from the prefix-wise or cumulative fairness objectives studied in temporal fair division.

\subsection{Online Matching} 
In the classical online bipartite matching problem, one side of the graph is fixed while vertices on the other side arrive online, and the algorithm must irrevocably match or reject arriving vertices. \citet{KarpVaziraniVazirani1990} gave the celebrated Ranking algorithm and its optimal competitive guarantee for adversarial arrivals.  Fully online matching instead allows all
vertices to arrive over time.  \citet{huang2018match} show that
Ranking is $0.5211$-competitive for general graphs in that
model and that the classical $1-1/e$ value cannot be
attained.

Our setting keeps the agents fixed and lets houses arrive,
but differs from standard online matching in two ways.
First, an accepted house may trigger a bounded reassignment
chain.  Second, our objectives are envy-freeability and
subsidy rather than matching cardinality.  Work on matching
with replacements, including \citet{BernsteinHolmRotenberg2019} and \citet{ShinKimLeeAn2020}, is therefore the
closest algorithmic comparison.

\subsection{Fairness in Online Matching}
Fairness has also been studied directly in online matching. \citet{HosseiniHuangIgarashiShah2024} introduce class fairness in online matching, where online-arriving items are assigned irrevocably to agents partitioned into classes, and fairness is required across classes. Their model shares the matching constraint that each agent can receive at most one item, but fairness is evaluated across classes rather than individuals. To this end, they adapt several standard fairness notions to the class level, including class envy-freeness (CEF), class EF1 (CEF1), class proportionality (CPROP), and class maximin share fairness (CMMS). They design online algorithms that achieve approximate fairness and efficiency guarantees, and also provide upper bounds on the approximations that can be achieved by any online algorithm. Their algorithms satisfy an efficiency notion called non-wastefulness, which implies $1/2$-approximation of the optimal utilitarian social welfare (USW), i.e., the sum of agent utilities, which is effectively the size of the matching. Subsequently, \citet{hajiaghayi2024fairness} resolve an open problem posed by \citet{HosseiniHuangIgarashiShah2024} by providing the first non-wasteful randomized algorithm that simultaneously achieves constant-factor guarantees for both CEF and CPROP in expectation. They further complement this positive result with impossibility bounds on the achievable approximation of CEF, for both indivisible and divisible matching, and quantify the trade-off between fairness and utilitarian social welfare through a price-of-fairness analysis.
More broadly, neighboring work in online matching considers weighted objectives~\cite{feldman2009online}, fully online arrival where both sides of the graph evolve~\cite{wang2015two,huang2018match}, repeated matching~\cite{caragiannis2023repeatedmatching,micheel2024repeated,Lim2026repeated,gollapudi2020almost}, and temporal slot assignment~\cite{elkind2022temporalslot}. These strands provide useful algorithmic comparisons, but the cited models do not combine individual subsidy-based fairness with bounded reassignment chains.

\subsection{Additional Related Problems}
Our work is also adjacent to fair division under feasibility constraints, including cardinality, matroid, and budget constraints \cite{suksompong2021constraints,biswas2018cardinality,dror2023matroid,barman2023budget,Elkind2024}. Most closely, Cohen et al.~\cite{cohen2026house} consider an online goods model in which each arriving item must be assigned irrevocably to a budget-feasible agent or to charity. In contrast, our feasibility constraint is unit demand: each agent receives exactly one house, and retaining a newly arriving house may require a chain of reassignments followed by an irreversible discard. Moreover, we seek envy-freeability through subsidies rather than bundle-based guarantees such as EF1 \cite{lipton2004approximately} or EFX \cite{caragiannis2019unreasonable,neoh2025efx}, so feasibility constrains both the final matching and the recourse required to maintain fairness online.

Other sequential partitioning problems share our emphasis on decisions made under uncertainty. Work on online clustering and coalition formation~\cite{cohen2024onlinefriends,cohen2024online,cohen2025online,cohen2025decentralized,cohen2026delayed,cohen2023online,cohen2025fair} forms or learns partitions of agents as information is revealed. The evolving state in those models is a coalition or cluster structure, whereas our state is a matching between a fixed set of agents and exogenously arriving houses. Consequently, their objectives and update operations are defined over coalition utilities, stability, or learning performance rather than subsidy-based envy-freeness and reassignment chains.

Temporal voting ~\cite{alouf2022better,ElkindObraztsovaTeh2024TemporalFairness,teh2026price,phillips2026strengthening,elkind2025not,elkind2025verifying,elkind2024temporal,zech2024multiwinner} likewise studies repeated decisions and notions of representation or fairness across time. There, each period selects a collective alternative or committee, and performance is evaluated over the resulting sequence of collective outcomes. In our model, each arrival changes the resource set and may force a revision of the current matching. The one house per agent constraint, the permanent loss of discarded houses, and the final subsidy objective therefore have no direct counterpart in temporal voting.

\section{Example 2.2 by \citet{choo2024housesubsidy}}
\label{app:example by choo et al}

\begin{example}[\citet{choo2024housesubsidy}]
\label{ex:envy-freeable}
    Suppose there are $m = 2$ houses and $n = 2$ agents whose utilities are given by $u_1(h_1)= u_1(h_2)=u_2(h_1)=200$ and $u_2(h_2)=100$.
    We claim that there does not exist an envy-eliminating subsidy vector for the allocation $\A = (h_1,h_2)$.
    To see this, let $\mathbf{s} = (s_1,s_2)$ be any subsidy vector. In order for the outcome $(\A,\mathbf{s})$ to be envy-free, the following must hold:
    \begin{align*}
        200 + s_1 = u_1(h_1) + s_1 & \geq u_1(h_2) + s_2 = 200 + s_2, \text{ and}\\
        100 + s_2 = u_2(h_2) + s_2 & \geq u_2(h_1) + s_1 = 200 + s_1.
    \end{align*}
    This implies that $s_1 \geq s_2 \geq 100+ s_1$, which is impossible.
\end{example}

\section{Omitted Proofs from Section~\ref{sec:prelim}}

\subsection{Proof of Lemma~\ref{lem:fixed-support}}

\lemfixedsupport*

\begin{proof}
The first statement follows directly from Proposition~\ref{prop:hs-characterization}. Indeed, an allocation with support $S$ is envy-freeable if and only if it maximizes utilitarian welfare over all permutations of the houses in $S$, which is equivalent to saying that its edge set is a maximum-weight perfect matching between $N$ and $S$.

For the second statement, we use Proposition 4.2 of \citet{choo2024housesubsidy}: when the number of houses equals the number of agents, a minimum-subsidy envy-free outcome is obtained by taking a maximum-weight perfect matching and computing minimum envy-eliminating subsidies for that matching; the value obtained is independent of the chosen maximum-weight perfect matching. Applying this result to the house set $S$ gives that every maximum-weight perfect matching between $N$ and $S$ attains the same subsidy value, namely $\mu(S)$.
\end{proof}

\section{Omitted Proofs from Section~\ref{sec:envy-freeability}}

\subsection{Proof of Lemma~\ref{lem:discarded-remain-discardable}}

\lemdiscardedremaindiscardable*

\begin{proof}
The proof proceeds in three steps. First, we reformulate $F$ as the maximum weight of a matching in an appropriate bipartite graph. This reformulation then gives us that $F$ is monotone. Second, we show that $F$ is submodular\footnote{The function $F$ is the standard assignment valuation (also called an OXS valuation): for each set of houses, it returns the maximum weight of a matching between those houses and the agents. Assignment valuations are a classical subclass of submodular set functions \citep{shapley1974cores,lehmann2006combinatorial}. We nevertheless give a self-contained proof of the properties we need.}. The key idea is to compare two optimal matchings through their symmetric difference, which decomposes into alternating paths and cycles. By exchanging edges component-by-component, we construct two new matchings whose total weight matches that of the original optimal matchings, yielding the required diminishing-returns inequality. Finally, we combine monotonicity and submodularity to show that if removing a set of houses does not decrease the value of $F$ for a set $S$, then removing the same houses from any superset of $S$ also leaves the value of $F$ unchanged.

\paragraph{Matching formulation of $F$.} Recall that each allocation $\A$ can be represented by its matching $E(\A)$ in the complete bipartite graph with vertex sets $N$ and $H$, where each edge $(i,h)\in N\times H$ has weight $u_i(h)$. For every $S\subseteq H$, we now show that
\begin{equation}
\label{eqn:discard_matching_formulation}
F(S)=\max\left\{w(M) : M \text{ is a matching in } N\times S\right\},    
\end{equation}
where $w(M) \coloneq \sum_{(i,h)\in M} u_i(h)$ is the weight of matching~$M$.
Indeed, every feasible allocation using houses from $S$ and assigning exactly $\min\{|S|,n\}$ houses corresponds to a matching of cardinality $\min\{|S|,n\}$ in $N\times S$, so the right-hand side of \eqref{eqn:discard_matching_formulation} is at least $F(S)$.

Conversely, let $M$ be any matching in $N\times S$ with $|M|<\min\{|S|,n\}$. Then, there exists at least one unmatched agent and at least one unmatched house in $S$. Since all edge weights are non-negative and the graph is complete bipartite, we can add an edge between an unmatched agent and an unmatched house without decreasing the weight. Repeating this until the matching has cardinality $\min\{|S|,n\}$ gives us a matching $M'$ of weight $w(M') \geq w(M)$.
Moreover, the matching $M'$ corresponds to a feasible allocation counted by $F(S)$.
Hence, we have $F(S) \geq w(M') \geq w(M)$.
Since this holds for every matching $M$ in $N\times S$, the
right-hand side of \eqref{eqn:discard_matching_formulation}
is at most $F(S)$, proving the equality.

\paragraph{Monotonicity of $F$.} We next show that $F$ is monotone, i.e., for any subsets $R \subseteq S \subseteq H$, we have $F(R) \leq F(S)$. Indeed, let $M$ be a maximum-weight matching between $N$ and $R$ with cardinality $\min\{|R|,n\}$. Since $R \subseteq S$, the matching $M$ is also a matching between $N$ and $S$.
By the previous paragraph, $F(S)$ is the maximum-weight among all matchings between $N$ and $S$.
Hence, $F(S) \geq w(M) = F(R)$, and therefore $F$ is monotone, i.e., $F(R)\le F(S)$.

\paragraph{Submodularity of $F$.} Subsequently, we prove that $F$ is also submodular, i.e., for any $h \in H$ and $X \subseteq Y \subseteq H \setminus \{h\}$, we have 
\begin{equation}
\label{eqn:discard_submodular}
    F(X \cup \{h\}) - F(X) \ge F(Y \cup \{h\}) - F(Y),
\end{equation}
or equivalently,
\begin{equation}
\label{eqn:discard_submodular_alternate}
    F(Y)+F(X\cup\{h\})
    \ge F(Y\cup\{h\})+F(X).
\end{equation}
To prove \eqref{eqn:discard_submodular_alternate}, let $M^+$ (resp. $M^0$) be a maximum-weight matching corresponding to $F(Y \cup \{h\})$ (resp. $F(X)$).
We next construct a matching $M_Y$ (resp. $M_{X, h}$) using only the houses in $Y$ (resp. $X \cup \{h\}$), such that the sum of the weights of $M_Y$ and $M_{X, h}$ is equal to the sum of the weights of $M^+$ and $M^0$.

Consider any connected component $C$ in $M^+\triangle M^0$.
We now show that $C$ cannot contain both $h$ and a house in $Y \setminus X$.
First, observe that since there are no incident edges in a matching, $C$ is either a path or a cycle, and the edges in $C$ must alternate between $M^+$ and $M^0$; we say that $C$ is an alternating path or an alternating cycle.
Furthermore, since $M^0$ uses only houses in $X$, any house in $(Y \cup \{h\}) \setminus X$ is incident to no edge of $M^0$.
Hence, whenever such a house appears in a component of
$M^+\triangle M^0$, it is a degree-one endpoint whose unique
incident edge belongs to $M^+$.
Therefore, such a house cannot lie on an alternating cycle.
Moreover, an alternating path cannot have two such endpoints: a path whose
endpoints are both house vertices has even length, so its first and last
edges belong to different matchings.
On the other hand, a path containing two houses in $(Y \setminus X) \cup \{h\}$ must have two endpoint edges belonging to $M^+$.
Therefore, it follows that $C$ cannot contain both $h$ and a house in $Y \setminus X$.

We now continue with the construction of $M_Y$ and $M_{X, h}$.
From the previous paragraph, we have established that there are three disjoint cases:
\begin{description}
    \item[Case 1:] $C$ contains the house~$h$.
    So, $C$ does not contain any house in $Y \setminus X$, and we can put $C \cap M^+$ into $M_{X, h}$.
    Furthermore, note that $M^0$ does not contain house $h$, so we can put $C \cap M^0$ into $M_Y$.
    \item[Case 2:] $C$ contains some house in $Y \setminus X$. 
    Then, $C$ does not contain house $h$, so we can put $C \cap M^+$ into $M_Y$.
    Moreover, since $M^0$ only uses the houses in $X$, we can put $C \cap M^0$ into $M_{X, h}$.
    \item[Case 3:] $C$ contains neither $h$ nor a house in $Y \setminus X$. That is, all houses in $C$ belong to $X$. 
    In this case, we can assign $C \cap M^+$ to one of $M_Y$ and $M_{X,h}$ and $C \cap M^0$ to the other, arbitrarily.
\end{description}

By our construction, we have established that $M_Y$ uses only houses in $Y$, while $M_{X,h}$ uses only houses in $X\cup\{h\}$.

We next show that $M_Y$ and $M_{X,h}$ are matchings, i.e., they do not contain any incident edges. Observe that
$M^+\cap M^0$ is vertex-disjoint from $M^+\triangle M^0$, and that
distinct connected components of $M^+\triangle M^0$ are vertex-disjoint.
Moreover, the sets $M^+\cap M^0$, $C\cap M^+$, and $C\cap M^0$ contain
no incident edges for every connected component $C$ of
$M^+\triangle M^0$. By construction, each of $M_Y$ and $M_{X,h}$
contains all edges of $M^+\cap M^0$, and for every connected component
$C$ of $M^+\triangle M^0$, exactly one of the sets
$C\cap M^+$ and $C\cap M^0$ is included in $M_Y$, while the other is included in $M_{X,h}$. Therefore, $M_Y$ and $M_{X,h}$ are both
matchings.

Therefore, it follows that
$M_Y$ is a matching in $N\times Y$ and $M_{X,h}$ is a matching in
$N\times (X\cup\{h\})$. Thus, $F(Y)$ is at least the weight of $M_Y$, i.e., $F(Y)\ge \sum_{(i,g)\in M_Y}u_i(g)$, and $F(X \cup \{h\})$ is at least the weight of $M_{X, h}$, i.e.,
\[
F(X\cup\{h\})
\ge
\sum_{(i,g)\in M_{X,h}}u_i(g).
\]

Lastly, note that $M_Y \cup M_{X, h} = M^+ \cup M^0$, since $M_Y$ and $M_{X,h}$ use only edges from $M^+\cup M^0$, and every
edge of $M^+\cup M^0$ is placed in at least one of $M_Y$ and $M_{X,h}$.
Furthermore, $M_Y \cap M_{X, h} = M^+ \cap M^0$, since every edge of $M^+\cap M^0$ is placed in both $M_Y$ and $M_{X,h}$,
whereas every edge of $M^+\triangle M^0$ is placed in exactly one of
them. Therefore, we have
\begin{align*}
F(Y) + F(X \cup \{h\}) 
&\quad\geq \sum_{(i,g)\in M_Y}u_i(g)+\sum_{(i,g)\in M_{X,h}}u_i(g) \\
&\quad= \sum_{(i,g)\in M_Y \cup M_{X, h}}u_i(g) + \sum_{(i,g)\in M_Y \cap M_{X, h}}u_i(g) \\
&\quad= \sum_{(i,g)\in M^+ \cup M^0}u_i(g) + \sum_{(i,g)\in M^+ \cap M^0}u_i(g) \\
&\quad= \sum_{(i,g)\in M^+}u_i(g)+\sum_{(i,g)\in M^0}u_i(g) \\
&\quad= F(Y \cup \{h\}) + F(X).
\end{align*}

This establishes \eqref{eqn:discard_submodular_alternate} and hence
\eqref{eqn:discard_submodular}, giving us that $F$ is submodular. 

\paragraph{Completing the proof of Lemma~\ref{lem:discarded-remain-discardable}.} 

If $D$ is empty, then the lemma holds trivially.
Hence, we can assume that $D$ is nonempty and
let $D=\{h_1,\dots,h_q\}$ for some $q \geq 1$.
For each $r\in\{0,\dots,q\}$, define
\[
S_r \coloneq S\setminus\{h_1,\dots,h_r\}
\quad\text{and}\quad
T_r \coloneq T\setminus\{h_1,\dots,h_r\}.
\]
Since $S_q\subseteq S_{q-1}\subseteq\cdots\subseteq S_0$ and $F$ is monotone, we have $F(S_q)\le F(S_{q-1})\le\cdots\le F(S_0)$.
Together with $F(S_0)=F(S_q)$, it follows that
\[
F(S_{r-1})=F(S_r)
\quad\text{for every }r\in[q].
\]
Now, fix any $r\in[q]$.
We next prove that $F(T_{r-1}) = F(T_r)$.
Since $S_r = S_{r-1} \setminus \{h_r\}$, $T_r = T_{r-1} \setminus \{h_r\}$, and $S_r\subseteq T_r$, applying \eqref{eqn:discard_submodular} with
$X=S_r$, $Y=T_r$ and $h=h_r$ gives us
\[
F(T_{r-1})-F(T_r)
\le
F(S_{r-1})-F(S_r)
=
0.
\]
On the other hand, monotonicity of $F$ implies $F(T_{r-1})-F(T_r) \ge 0$.
Therefore, $F(T_{r-1})=F(T_r)$.

Since this holds for every $r\in[q]$, repeated application gives us
\[
F(T)=F(T_0)=F(T_q)=F(T\setminus D). \qedhere
\]
\end{proof}

\subsection{Proof of Theorem~\ref{thm:active-chain}}

\thmvalidflexmaxwelfare*
\begin{proof}

The proof consists of three parts. We first show that the allocation maintained after each round maximizes the utilitarian welfare among all feasible allocations using the currently active houses and assigning the maximum possible number of houses. This establishes the welfare guarantee. We then verify that the algorithm is valid. Finally, we prove that Algorithm~\ref{alg:active-max-weight-chain} is $n$-flexible. To do so, we analyze the symmetric difference between the matchings corresponding to two consecutive allocations. Using the welfare-optimality of both allocations together with the tie-breaking rule, we show that this symmetric difference contains at most one nonempty connected component. This component is an alternating path containing the newly arrived house $h_t$, which directly gives us a chain update of length at most $n$.

Throughout the proof, let $F$ be the set function defined in Lemma~\ref{lem:discarded-remain-discardable}. In particular, the proof of Lemma~\ref{lem:discarded-remain-discardable} implies that $F$ is monotone and that if $D\subseteq S$ and $F(S)=F(S\setminus D)$, then
\[
F(T)=F(T\setminus D)
\quad\text{for every }T\supseteq S.
\]

\paragraph{Welfare Maximization.} 
Although the theorem only concerns the final allocation $\A^m$, we prove the stronger statement that $\A^t$ is welfare-maximizing after every round $t$. This invariant will be used later in the proof of $n$-flexibility.
To this end, we prove by induction on $t$ that $\A^{t}$ is a maximum-welfare feasible allocation using houses from $H^{t}$ and assigning exactly $\min\{t,n\}$ houses.
The claim is immediate for $t=0$. Suppose it holds after round $t-1$. 
We first show that $\A^{t-1}$ is among the allocations considered in the definition of
$F(H^{t-1}\setminus D^{t-1})$. Since $D^{t-1}=H^{t-1}\setminus \mathrm{supp}(\A^{t-1})$,
we have $\mathrm{supp}(\A^{t-1})
=
H^{t-1}\setminus D^{t-1}$.
Moreover,
\[
|\mathrm{supp}(\A^{t-1})|
=
\min\{t-1,n\}
=
\min\{|H^{t-1}\setminus D^{t-1}|,n\}
\]
Thus, $\A^{t-1}$ uses only houses from
$H^{t-1}\setminus D^{t-1}$ and assigns exactly the number of houses required in the definition of
$F(H^{t-1}\setminus D^{t-1})$. Since $\A^{t-1}$ has utilitarian welfare $F(H^{t-1})$ by the induction hypothesis, it follows that $F(H^{t-1}\setminus D^{t-1}) \ge F(H^{t-1})$.
The reverse inequality follows from monotonicity, since
$H^{t-1}\setminus D^{t-1} \subseteq H^{t-1}$. Therefore, $F(H^{t-1}) = F(H^{t-1}\setminus D^{t-1})$.
Applying Lemma~\ref{lem:discarded-remain-discardable} with $S=H^{t-1}$, $D=D^{t-1}$, $ T=H^{t}$, we obtain
\[
F(H^{t})=F(H^{t}\setminus D^{t-1}).
\]
This shows that discarding the houses in $D^{t-1}$ does not reduce the maximum welfare achievable after round $t$.

To complete the proof of the inductive step, we show that $\A^t$ attains utilitarian welfare $F(H^t \setminus D^{t-1}) = F(H^t)$.
Since $h_t \notin H^{t-1}$ and
$H^{t-1}\setminus D^{t-1}
=\mathrm{supp}(\A^{t-1})$, we have the disjoint union $H^{t}\setminus D^{t-1}
= \mathrm{supp}(\A^{t-1}) \uplus \{h_t\}$.
Hence, $|H^{t}\setminus D^{t-1}|
=
|\mathrm{supp}(\A^{t-1})|+1 =
\min\{t-1,n\}+1$
and so,
\begin{equation}
    \label{eq:h minus d}
    \min\{|H^{t}\setminus D^{t-1}|,n\}
=
\min\{t,n\}.
\end{equation}
Therefore, by the definition of $\mathcal{A}_t$ in line~\ref{state:feasible} of Algorithm~\ref{alg:active-max-weight-chain}, $\mathcal{A}_t$ is exactly the set of feasible allocations using houses from $H^{t}\setminus D^{t-1}$ and assigning exactly $\min\{|H^{t}\setminus D^{t-1}|,n\}$ houses. By the equality in \eqref{eq:h minus d}, this is precisely the class of allocations considered in the definition of $F(H^{t}\setminus D^{t-1})$.
Since Algorithm~\ref{alg:active-max-weight-chain} chooses a maximum-welfare allocation in $\mathcal A_t$, it follows that $\A^t$ attains utilitarian welfare $F(H^t\setminus D^{t-1})=F(H^t)$.
Therefore, $\A^t$ is a maximum-welfare feasible allocation using houses from $H^t$ and assigning exactly $\min\{t,n\}$ houses. This completes the induction.

\paragraph{Validity.} By construction, $|\mathrm{supp}(\A^{t})|=\min\{t,n\}$ for every $t$. Since $m\geq n$, the final allocation $\A^{m}$ is complete. Thus, the algorithm is valid.

\paragraph{$n$-Flexibility.} It remains to show that each transition can be implemented as a chain update of length at most $n$. Let $M_{t-1}\coloneq E(\A^{t-1})$ and $M_t\coloneq E(\A^{t})$.
View $M_{t-1}$ and $M_t$ as matchings in the bipartite graph with agent side $N$ and house side $H^{t}$. Since $\A^{t}$ avoids $D^{t-1}$, every house incident to an edge of $M_t$ is either $h_t$ or belongs to $\mathrm{supp}(\A^{t-1})$. Since $M_{t-1}$ and $M_t$ are matchings, every connected component of $M_{t-1}\triangle M_t$ is an alternating path or an alternating cycle. Moreover, at most one connected component of $M_{t-1}\triangle M_t$ contains $h_t$.

We claim that every nonempty connected component of $M_{t-1}\triangle M_t$ contains the newly arrived house $h_t$.
Suppose, for contradiction, that $C$ is a nonempty component of $M_{t-1}\triangle M_t$ that avoids $h_t$. We first show that $C$ is balanced, i.e., $|C\cap M_{t-1}| = |C\cap M_t|$. This property will allow us to switch along $C$ without changing the number of assigned houses.
We distinguish between two cases:
\begin{description}
    \item[Case 1:] If $t\le n$, then $M_{t-1}$ matches every house in $H^{t-1}$, while
    $M_t$ matches every house in $H^{t}$. Since $C$ avoids $h_t$, every
    house vertex in $C$ has degree two in $M_{t-1}\triangle M_t$. Thus, $C$ is
    either an alternating cycle, or an alternating path whose two endpoints are
    agents. In the latter case, the path has endpoints on the same side of the
    bipartition, and hence contains equally many $M_{t-1}$-edges and $M_t$-edges.
    Thus, $C$ is balanced.

    \item[Case 2:] Now suppose $t>n$. Thus, both $M_{t-1}$ and $M_t$ match every agent. If
    $h_t\notin \mathrm{supp}(\A^{t})$, then $M_{t-1}$ and $M_t$ use the
    same set of house vertices, namely $\mathrm{supp}(\A^{t-1})$. 
    Hence, each vertex that appears in $M_{t-1}\triangle M_t$ has degree two, and every nonempty component is an alternating cycle.
    
    Otherwise, if $h_t\in \mathrm{supp}(\A^{t})$, then $M_t$ uses $h_t$ and has the same cardinality as $M_{t-1}$. Since every other house used by $M_t$ belongs to $\mathrm{supp}(\A^{t-1})$, exactly one house $g\in \mathrm{supp}(\A^{t-1})$ is not used by $M_t$. 
    In $M_{t-1}\triangle M_t$, the only degree-one house vertices are $h_t$ and
    $g$, while every agent has degree zero or two. 
    Therefore, $h_t$ and $g$ are the endpoints of the unique alternating path component, and every component avoiding $h_t$ is an alternating cycle. 
\end{description}

Therefore, $C$ is balanced in all cases. 

We now compare the welfare contribution of the two matchings on the component $C$. Specifically, let
\[
    \Delta(C)
    \coloneq \sum_{(i,h)\in C\cap M_t} u_i(h)
    - \sum_{(i,h)\in C\cap M_{t-1}} u_i(h).
\]
Since $C$ is a connected component of $M_{t-1}\triangle M_t$,
switching along $C$ preserves the matching property: vertices
outside $C$ are unchanged, and the alternating structure of $C$
ensures that no vertex of $C$ is incident to more than one edge
after the switch. Since $C$ is balanced, this switch also preserves
the number of assigned houses.

Suppose, for contradiction, that $\Delta(C)\ne 0$. We distinguish between two cases:
\begin{description}
    \item[Case 1:] If $\Delta(C)>0$, then replacing the $M_{t-1}$-edges of $C$ by
    the $M_t$-edges of $C$ in $M_{t-1}$ gives a feasible allocation
    using only houses from $H^{t-1}$: indeed, $C$ avoids $h_t$,
    and every other house in $C$ belongs to $H^{t-1}$. The switched allocation assigns exactly $\min\{t-1,n\}$ houses and has larger welfare than $\A^{t-1}$, contradicting the fact that $\A^{t-1}$ is a maximum-welfare feasible allocation using houses from $H^{t-1}$.
    
    \item[Case 2:] If $\Delta(C)<0$, then replacing the $M_t$-edges of $C$ by the
    $M_{t-1}$-edges of $C$ in $M_t$ gives a feasible allocation in
    $\mathcal{A}_t$: it still uses only houses in $H^{t}\setminus D^{t-1}$,
    because the $M_{t-1}$-edges use houses in
    $\mathrm{supp}(\A^{t-1})$. It assigns the same number of houses as $\A^t$ and has strictly larger
    welfare, contradicting the fact that $\A^t$ is a maximum-welfare
    allocation in $\mathcal A_t$.
\end{description}

Therefore, $\Delta(C)=0$. We now show that the existence of such a component $C$ contradicts the tie-breaking rule. Let $M' \coloneq (M_t\setminus (C\cap M_t))\cup (C\cap M_{t-1})$.
By the same switching argument, $M'$ corresponds to an allocation
$\mathbf{b} \in \mathcal{A}_t$. Since $\Delta(C)=0$, this allocation has the same welfare as $\A^{t}$, and hence is also a maximum-welfare allocation in $\mathcal A_t$. 
Moreover,
\[
M'\triangle M_{t-1}
=
(M_t\triangle M_{t-1})
\setminus ((C\cap M_t)\cup(C\cap M_{t-1})),
\]
because the switch replaces the $M_t$-edges of $C$ by the
corresponding $M_{t-1}$-edges. Thus, every edge of $C$ disappears
from the symmetric difference, while all edges outside $C$ remain
unchanged. Since $C$ is nonempty, the set $(C\cap M_t)\cup(C\cap M_{t-1})$
is nonempty. Therefore, $|E(\mathbf b)\triangle E(\A^{t-1})| < |E(\A^t)\triangle E(\A^{t-1})|$.
This contradicts the tie-breaking rule in Algorithm~\ref{alg:active-max-weight-chain}.
Hence, no nonempty component of $M_{t-1}\triangle M_t$ avoids $h_t$.

It follows that $M_{t-1}\triangle M_t$ has at most one nonempty component.
If $M_{t-1}\triangle M_t=\varnothing$, then $E(\A^{t})=E(\A^{t-1})$, and hence
$\A^{t}=\A^{t-1}$. Since the newly arrived house $h_t$ does not belong to $ \mathrm{supp}(\A^{t-1})$,
the algorithm rejects $h_t$.

Otherwise, there is a unique nonempty component, and it contains $h_t$. Since $h_t$ is
incident to no edge of $M_{t-1}$, this component cannot be an alternating
cycle; it is an alternating path with $h_t$ as an endpoint. The transition from $\A^{t-1}$ to $\A^t$ is completely determined by this path: every assignment change occurs along the path, while all agents outside the path keep their previous houses. Read this path
starting from $h_t$, and let $i_1,\dots,i_\ell$ be the agents encountered
along it. The agents $i_1,\dots,i_\ell$ are distinct, and $\ell\le n$.
The first edge of the path is $(i_1,h_t)\in M_t$, so $a^{t}_{i_1}=h_t$.

For every $r=2,\dots,\ell$, the path enters $i_r$ through the house
previously held by $i_{r-1}$, so $a^{t}_{i_r}=a^{t-1}_{i_{r-1}}$.

Every agent outside $\{i_1,\dots,i_\ell\}$ is incident to no edge of
$M_{t-1}\triangle M_t$, and therefore $a^{t}_j=a^{t-1}_j$ for every $j\notin\{i_1,\dots,i_\ell\}$.
Furthermore, for every $r<\ell$, the path continues from $i_r$ through its
old house, so $a^{t-1}_{i_r}\neq \varnothing$. If the path ends at the agent $i_\ell$, then $i_\ell$ is not incident
to an $M_{t-1}$-edge, so $a^{t-1}_{i_\ell}=\varnothing$. Otherwise,
the path ends at the house $a^{t-1}_{i_\ell}$; this house is not
incident to any edge of $M_t$, so it belongs to $H^{t}\setminus \mathrm{supp}(\A^{t}) = D^{t}$
and is discarded. 
Thus, the transition from $\A^{t-1}$ to $\A^{t}$ is a chain update of length $\ell\le n$. Therefore, the algorithm is $n$-flexible.

\paragraph{Completing the proof of Theorem~\ref{thm:active-chain}.} 
We have therefore proved the welfare statement for every round
\(t\in\{0\}\cup[m]\), as well as validity and \(n\)-flexibility.
For every \(t\ge n\), the allocation \(\A^t\) is complete and
welfare-maximizing over the prefix \(H^t\); hence it is envy-freeable by
Proposition~\ref{prop:hs-characterization}. In particular, this holds
for the final allocation \(\A^m\).
\end{proof}

\subsection{Proof of Proposition~\ref{prop:stronger-odd-lower-bound}}

\strongeroddlowerbound*

\begin{proof}
Fix $n \ge 2$, and let $A$ be any deterministic valid
$k$-flexible algorithm. Choose a constant $B>n^2$.

The adversary first releases a house $h_1$ with $u_i(h_1)=B$ for every $i \in N$.
Now, since discarded houses cannot be recovered, validity implies
\[
|\mathrm{supp}(\A^{(t)})|=\min\{t,n\}
\quad
\text{for every } t\in\{0,1,\dots,m\}.
\]
Indeed, suppose that for some $t<n$, $|\mathrm{supp}(\A^{(t)})|<t$.
Then at least one of the first $t$ houses has already been
discarded. Since discarded houses cannot be recovered, even
if every future arrival is assigned, after the first $n$
arrivals the algorithm can assign at most
\[
|\mathrm{supp}(\A^{(t)})|+(n-t)
<
t+(n-t)
=
n
\]
houses. Hence, on the instance obtained by continuing the
sequence until $m=n$, the final allocation cannot be
complete, contradicting validity.

In particular, $|\mathrm{supp}(\A^{(1)})|=1$,
so $A$ assigns $h_1$ to some agent. Let this agent be
$p$.

Moreover, for every $t<n$, the update at round $t$ increases the
number of assigned houses from $t-1$ to $t$. Otherwise, after round $t$ fewer
than $t$ houses would be assigned. Even if each of the remaining
$n-t$ rounds increased the number of assigned houses by one, the final
allocation would contain fewer than $n$ houses, contradicting validity.

We next show that any update that increases the number of assigned houses must preserve the assignment status of every previously assigned agent. If $k=0$, then by definition, $h_t$ must be assigned to an unassigned agent, with no previously assigned house changed. If $k\ge1$, then the chain must terminate at an agent who was
unassigned before round $t$; otherwise every agent appearing
in the chain was already assigned before the update, so the
number of assigned houses would remain unchanged. Therefore, every agent who is assigned before round $t<n$ remains assigned after that round. In particular, $p$ remains assigned after every round $t<n$.

Now fix an ordering of the agents $\sigma(1),\sigma(2),\dots,\sigma(n)$
such that $p=\sigma(c)$, where $c\coloneq\left\lceil \frac{n}{2}\right\rceil$.
For an agent $i$, denote $r(i)$ as its rank in this ordering;
that is, $r(\sigma(j))=j$.

For each $t=2,\dots,n-1$, the adversary releases $h_t$
with $u_i(h_t)=B+r(i)(t-1)$ for every $i \in N$.
Thus, the old houses $h_1,\dots,h_{n-1}$ have scores $0,1,\dots,n-2$, respectively, up to the common additive constant $B$. The
constant $B$ does not affect welfare comparisons between
assignments using the same number of houses.

After round $n-1$, validity implies that all houses
$h_1,\dots,h_{n-1}$ are assigned. Let $\mathbf{b} \coloneq \A^{(n-1)}$.
Let $q$ be the unique rank such that agent $\sigma(q)$ is
unassigned in $\mathbf{b}$. Since $p=\sigma(c)$ remains assigned,
we have $q \neq c$.

We use the following observation. Consider houses whose utilities
have the form $u_i(h)=B+r(i) \cdot x_h$,
where the numbers $x_h$ are pairwise distinct. For any fixed
set of agents and any equally sized set of houses whose $x_h$-values are distinct, the unique
welfare-maximizing assignment matches smaller $x_h$'s to
smaller ranks. Indeed, if $r(i)<r(j)$, agent $i$ receives
a house $h$, agent $j$ receives a house $g$, and
$x_h>x_g$, then swapping $h$ and $g$ changes welfare by
\begin{equation*}
    (B+r(i) \cdot x_g+B+r(j) \cdot x_h) - (B+r(i) \cdot x_h+B+r(j) \cdot x_g)  = (r(j)-r(i))(x_h-x_g)>0.
\end{equation*}
Thus, any inversion is suboptimal, and uniqueness follows from
the fact that the $x_h$'s are distinct.

For each $q\in[n]$, define the partial allocation
$\widehat{\A}^q$ of the first $n-1$ houses by
\[
    \widehat a^q_{\sigma(j)} =
    \begin{cases}
        h_j, & \text{if } j<q,\\
        \varnothing, & \text{if } j=q,\\
        h_{j-1}, & \text{if } j>q.
    \end{cases}
\]
By the observation above, $\widehat{\A}^q$ is the unique
welfare-maximizing assignment of $h_1,\dots,h_{n-1}$ to
the agents $N\setminus\{\sigma(q)\}$.

We distinguish between two cases.

\paragraph{Case 1: $\mathbf{b}\neq \widehat \A^q$.}
The adversary releases one final house $h_n$. Choose a number
$x_n$ so that $h_n$ is inserted at rank $q$ in the order of
house scores:
\[
    x_n \coloneq
    \begin{cases}
        -1, & \text{if } q=1,\\
        q-\frac{3}{2}, & \text{if } 2\le q\le n-1,\\
        n-1, & \text{if } q=n.
    \end{cases}
\]
Then set $u_i(h_n)=B+r(i)x_n$ for every $i\in N$.
Since $B>n^2$, all utilities are non-negative.

By the observation above, the unique welfare-maximizing complete
allocation is $\mathbf{c}^q$, where
\[
    c^q_{\sigma(j)} =
    \begin{cases}
        h_j, & \text{if } j<q,\\
        h_n, & \text{if } j=q,\\
        h_{j-1}, & \text{if } j>q.
    \end{cases}
\]
Since $m=n$, every complete allocation assigns all houses in $H$.
Thus, for any complete allocation $\A$, the allocations obtained by permuting the houses in $\mathrm{supp}(\A)$ are exactly the complete allocations of this instance.
By Proposition~\ref{prop:hs-characterization}, any envy-freeable complete allocation must therefore maximize
utilitarian welfare among all complete allocations. Since $\mathbf{c}^q$ is the
unique welfare-maximizing complete allocation, any envy-freeable final
allocation must be $\mathbf{c}^q$.

However, $c^q_{\sigma(q)}=h_n$, while $\sigma(q)$ is unassigned in $\mathbf b$. Thus, any update reaching $\mathbf c^q$ must assign $h_n$ to $\sigma(q)$.
If $k=0$, then by definition, there must be no change to any previously assigned house in such an update. If $k\ge1$, let $(i_1,\dots,i_\ell)$ be the chain update reaching $\mathbf c^q$. Since $h_n$ is assigned to $\sigma(q)$, we must have $i_1=\sigma(q)$. Since $a^{(n-1)}_{\sigma(q)}=\varnothing$, the condition $a^{(n-1)}_{i_1}\neq\varnothing$ required for chains of length at least two fails. Thus, the chain has length $\ell=1$.
Hence, the assignments of $h_1,\dots,h_{n-1}$ remain as in $\mathbf b$. Since
$\mathbf b\neq \widehat{\A}^q$, the resulting allocation cannot be
$\mathbf c^q$, contradicting the fact that every envy-freeable final
allocation must be $\mathbf c^q$.
Thus, $A$ cannot reach an envy-freeable final allocation. 

\paragraph{Case 2: $\mathbf{b}=\widehat{\A}^q$.}
The adversary chooses the final house to be either smaller than all
old houses or larger than all old houses.

First, suppose the adversary sets $u_i(h_n)=B-r(i)$ for every $i \in N$.
This corresponds to $x_n=-1$, so $h_n$ is the smallest house.
The unique welfare-maximizing complete allocation is
$\mathbf{c}^{\min}$, where
\[
    c^{\min}_{\sigma(1)}=h_n
    \quad\text{and}\quad
    c^{\min}_{\sigma(j)}=h_{j-1}
    \quad\text{for } j=2,\dots,n.
\]
Starting from $\widehat{\A}^q$, reaching $\mathbf{c}^{\min}$ requires the update to pass through the agents $\sigma(1),\dots,\sigma(q)$. 
Indeed, since $c^{\min}_{\sigma(1)}=h_n$, the first agent in the chain must be $\sigma(1)$. For each $r=2,\dots,q$, the house assigned to $\sigma(r)$ in
$\mathbf{c}^{\min}$ is $h_{r-1}$, which is the house assigned
to $\sigma(r-1)$ in $\widehat{\A}^q$. Since a chain update
moves each house only from an agent to the next agent in the
chain, the house $h_{r-1}$ cannot reach $\sigma(r)$ unless
$\sigma(r-1)$ appears earlier than $\sigma(r)$ in the chain. Thus, the first $q$ agents in the chain must be $\sigma(1),\sigma(2),\dots,\sigma(q)$.
Since $\sigma(q)$ is unassigned in $\widehat{\A}^q$, it cannot appear
before the last position of a chain update. 
Therefore, the chain must terminate at $\sigma(q)$. 
Thus, reaching $\mathbf{c}^{\min}$ requires a chain of length $q$.

Second, suppose the adversary sets $u_i(h_n)=B+r(i) \cdot (n-1)$ for every $i \in N$.
This corresponds to $x_n=n-1$, so $h_n$ is the largest house.
The unique welfare-maximizing complete allocation is
$\mathbf{c}^{\max}$, where
\[
    c^{\max}_{\sigma(j)}=h_j
    \quad\text{for } j=1,\dots,n-1,
    \quad
    c^{\max}_{\sigma(n)}=h_n.
\]
Similarly, starting from $\widehat{\A}^q$, reaching $\mathbf{c}^{\max}$ requires the update to pass through the agents $\sigma(n),\sigma(n-1),\dots,\sigma(q)$. 
Indeed, since $c^{\max}_{\sigma(n)}=h_n$, the first agent in the chain must be $\sigma(n)$. If $q=n$, this already gives a chain of length $1=n-q+1$. If $q<n$, then $c^{\max}_{\sigma(n-1)}=h_{n-1}$, which is the house assigned to $\sigma(n)$ in $\widehat{\A}^q$. Therefore, since $h_{n-1}$ is initially assigned to $\sigma(n)$ and finally assigned to $\sigma(n-1)$, the chain must pass through $\sigma(n)$ before it can assign the correct house to $\sigma(n-1)$. Continuing in the same way, the chain must contain $\sigma(n),\sigma(n-1),\dots,\sigma(q)$ in this order.
Since $\sigma(q)$ is unassigned in $\widehat{\A}^q$, it cannot appear
before the last position of a chain update. 
Therefore, the chain must terminate at $\sigma(q)$. Thus, reaching $\mathbf{c}^{\max}$ requires a chain of length $n-q+1$.

The adversary sees $q$ before round $n$, so it chooses the smaller final house if $q\ge n-q+1$, and the larger final house otherwise. For the resulting instance, reaching the unique welfare-maximizing complete allocation requires a chain of length $\max\{q,n-q+1\}$.
Since $q\neq c$ and $c=\lceil n/2\rceil$, either $q\ge c+1$,
in which case $\max\{q,n-q+1\}\ge c+1=\lambda_n$,  or $q\le c-1$, in which case $n-q+1\ge n-c+2\ge c+1=\lambda_n$.
Therefore, $\max\{q,n-q+1\}\ge \lambda_n$.
Thus, if $k<\lambda_n$, $A$ cannot reach the unique welfare-maximizing complete allocation. As above, since $m=n$, every complete allocation uses the same support $H$. Hence, Proposition~\ref{prop:hs-characterization} implies that any envy-freeable final allocation must be the unique welfare-maximizing
complete allocation. Hence, the final allocation is not envy-freeable.

In both cases, the adaptive adversary constructs an instance on
which $A$ fails. This proves the proposition.

\end{proof}

\subsection{Proof of Proposition~\ref{prop:welfare-max-n-lower-bound}}

\welfaremaxnlowerbound*

\begin{proof}
Let $A$ be any deterministic valid $k$-flexible algorithm with
$k<n$. 
In the first $n$ rounds, the adversary releases houses $h_1,\dots,h_n$ with $u_i(h_t)=it$ for every $i \in N$ and $t\in[n]$.
For these $n$ houses, the unique welfare-maximizing complete
allocation assigns $h_i$ to agent $i$ for every $i\in N$. This
follows from the observation in the proof of Proposition~\ref{prop:stronger-odd-lower-bound}, applied
with $r(i)=i$ and $x_{h_t}=t$.

If $A$ does not hold this allocation after round $n$, then the adversary stops with $m=n$, and $A$'s final allocation is not welfare-maximizing. Hence, any algorithm that guarantees final welfare maximization must hold $a^{n}_i=h_i$ for every $i \in N$ after round $n$.

The adversary now releases one final house $h_{n+1}$ with $u_i(h_{n+1})=i(n+1)$ for every $i \in N$.

Among the $n+1$ houses, the unique welfare-maximizing complete allocation is $c_i=h_{i+1}$ for every $i \in N$.
To see this, fix any $n$-element set of assigned houses and write it as $h_{t_1},h_{t_2},\dots,h_{t_n}$ where $t_1<t_2<\cdots<t_n$.

For this fixed set of houses, the exchange argument matches
$h_{t_i}$ to agent $i$ for every $i\in N$. The resulting welfare is $\sum_{i=1}^n i t_i$.
Since the selected indices are distinct and satisfy $t_1<t_2<\cdots<t_n$,
there are exactly $n-i$ selected indices larger than $t_i$.
All of these indices belong to $\{t_i+1,\dots,n+1\}$, which contains $n+1-t_i$ elements. Therefore $n+1-t_i\ge n-i$, so $t_i\le i+1$ for every $i\in N$. Therefore
\[
\sum_{i=1}^n i t_i \le \sum_{i=1}^n i(i+1).
\]
Since every coefficient $i$ is positive, equality holds only if
$t_i=i+1$ for every $i\in N$. Hence, the unique welfare-maximizing
complete allocation discards $h_1$ and assigns $h_{i+1}$ to agent
$i$ for every $i\in N$.

At round $n+1$, we can only have a rejection or a single
chain update of length at most $k$. A rejection does not produce
$\mathbf{c}$. If $k=0$, then after round $n$ no agent is unassigned, so $h_{n+1}$ must be rejected. Now suppose $k\ge 1$. 
Starting from $\A^{n}$, any chain update that reaches $\mathbf{c}$ must have length $n$. Indeed, we must have the chain as follows. 
Since $c_n=h_{n+1}$, the first agent in the chain must be
$n$. Moreover, $c_{n-1}=h_n=a_n^{n}$.
Thus, the house $h_n$, which is initially assigned to agent
$n$, must be transferred to agent $n-1$. In a chain
update, a house can only move from an agent to the next agent
in the chain. Therefore, agent $n$ must appear before agent $n-1$ in
the chain. Since agent $n$ is already the first agent in
the chain, agent $n-1$ must be the second. Continuing inductively, for every
$r=2,\dots,n$,
\[
c_{n-r+1}=h_{n-r+2}
=
a^{n}_{n-r+2},
\]
so agent $n-r+2$ must appear before agent $n-r+1$ in the
chain.
Thus, the chain contains all agents $n,n-1,\dots,1$, and $h_1$ is discarded at the end. Hence, reaching $\mathbf{c}$ requires a chain of length $n$. Since $k<n$, $A$ cannot reach $\mathbf{c}$ when $h_{n+1}$ arrives. The adversary then terminates the sequence, so $A$'s final allocation is not welfare-maximizing.
\end{proof}

\section{Omitted Proofs from Section~\ref{sec:subsidy-limits}}

\subsection{Proof of Lemma~\ref{lem:absolute-subsidy-bound}}

\absolutesubsidybound*
\begin{proof}
    If $n=1$, then the single agent envies no one, so $\mathrm{sub}(\A)=0$. 
    Assume henceforth that $n\ge 2$.
    Since $\A$ is envy-freeable, Proposition~\ref{prop:hs-characterization} implies that $G_\A$ has no positive-weight directed cycle.
    
    For a directed walk $Q=(i_0,i_1,\dots,i_q)$ in $G_\A$, let
    \[
    w_\A(Q)\coloneq \sum_{r=0}^{q-1} w_\A(i_r,i_{r+1}),
    \]
    with $w_\A(Q)=0$ when $q=0$. For each agent $i\in N$, let $L_i\coloneq \max\{w_\A(Q): Q \text{ is a directed walk in } G_\A \text{ starting at } i\}$.
    This maximum is well defined. Indeed, consider any directed walk $Q$ starting at $i$. If some vertex appears twice in $Q$, then the segment between two consecutive occurrences of that vertex is a closed directed walk. Since $G_\A$ has no positive-weight directed cycle, every closed directed walk has weight at most $0$, as it can be decomposed into directed cycles. Deleting this segment therefore cannot decrease the weight of $Q$. Repeating this operation produces a simple directed path starting at $i$ with weight at least $w_\A(Q)$. Hence, the maximum over all directed walks starting at $i$ is attained by one of the finitely many simple directed paths starting at $i$.
    
    We first show that every simple directed path in $G_\A$ has weight at most $1$. A simple directed path of length zero has weight $0$.
    Now let $P$ be a nonempty simple directed path from agent $i$ to agent $j$, where $i\ne j$. Adding the edge $j\to i$ to $P$ gives a directed cycle. Since $G_\A$ has no positive-weight directed cycle, $w_\A(P)+w_\A(j,i)\le 0$.
    
    Moreover, since $\A$ is complete, $a_i,a_j\in H$. Thus, by the normalization $u_j(h)\in[0,1]$ for every $h\in H$, $w_\A(j,i)=u_j(a_i)-u_j(a_j)\ge -1$.
    Therefore, $w_\A(P)\le -w_\A(j,i)\le 1$.
    Hence, $0\le L_i\le 1$ for every $i\in N$.

    Now let $L_{\min}\coloneq \min_{r\in N} L_r$, and define $s_i\coloneq L_i-L_{\min}$ for every $i \in N$.
    Then $s_i\ge 0$ for every $i$, and at least one agent receives
    subsidy $0$. The envy-free inequality is trivial when $i=j$. Now fix distinct agents $i,j\in N$. The edge $i\to j$ followed by a maximum-weight walk starting at $j$ is a directed walk in $G_\A$ starting at $i$. Hence $L_i \ge w_\A(i,j)+L_j$,     and therefore
    \[
    s_i-s_j=L_i-L_j\ge w_\A(i,j)=u_i(a_j)-u_i(a_i).
    \]
    Equivalently,
    \[
    u_i(a_i)+s_i\ge u_i(a_j)+s_j.
    \]
    Thus, $\mathbf{s}$ is envy-eliminating for $\A$. Since $L_{\min}\ge 0$, we have $s_i=L_i-L_{\min}\le L_i\le 1$
    for every $i\in N$. Moreover, at least one agent receives subsidy $0$. Therefore
    \[
    \mathrm{sub}(\A) \le \sum_{i\in N} s_i \le n-1. \qedhere
    \]
\end{proof}

\subsection{Proof of Theorem~\ref{thm:efficiency-subsidy-incompatible}}

For later use, suppose that $S$ contains $n$ houses with common
values $y_1\le y_2\le\cdots\le y_n$.
Then
\begin{equation}
\mu(S)=\sum_{r=1}^{n}(y_n-y_r).
\label{eqn:mu_1}
\end{equation}
Indeed, paying $y_n-y_r$ to an agent who receives value $y_r$
makes all agents equally well off, and any smaller total payment
leaves envy toward an agent receiving value $y_n$.

\efficiencysubsidyincompatible*

\begin{proof}
Choose $\delta>0$ so small that $0<\delta<\frac12$, $(n-1)\delta < \rho (1+(n-1)\delta )$, and $\frac{(n-1)(1-\delta)}{\delta} > R$.
Such a choice is possible: as $\delta \rightarrow 0$, the left-hand side of the second inequality tends to $0$ while its right-hand side tends to $\rho>0$, and the left-hand side of the third inequality tends to $\infty$.

Consider an identical-utilities instance $\I$ with an arbitrary arrival order over $n+1$ houses whose common values are
\[
1,\underbrace{\delta,\dots,\delta}_{n-1\text{ houses}},0.
\]
Since $0<\delta<1/2$, all values lie in $[0,1]$.

Since the utilities are identical, writing $u$ for the common valuation
as in Section~\ref{sec:prelim}, every feasible complete allocation $\A$ satisfies
\[
W(\A)=\sum_{i\in N} u(a_i)
     =\sum_{h\in \mathrm{supp}(\A)} u(h).
\]
Thus, $W(\A)$ depends only on $\mathrm{supp}(\A)$. Since $m=n+1$,
each complete allocation selects exactly $n$ of the $n+1$ houses.
Therefore, the maximum utilitarian welfare is obtained by omitting the
value-$0$ house, i.e., by using the value-$1$ house and all $n-1$
value-$\delta$ houses. Hence $W^*(\I)=1+(n-1)\delta$.

Any feasible complete allocation whose support does not contain the
value-$1$ house must use all remaining $n$ houses, namely the $n-1$
value-$\delta$ houses and the value-$0$ house. Its welfare is therefore
$(n-1)\delta$, which is strictly smaller than $\rho W^*(\I)$ since $(n-1)\delta < \rho (1+(n-1)\delta )$ by the
choice of $\delta$. Thus, every feasible complete allocation $\A$ with
$W(\A)\ge \rho W^*(\I)$ must contain the value-$1$ house in its support.

We next show that every subsidy-optimal support omits the value-$1$
house. To this end, we compare the subsidy values of the possible
$n$-element supports. By \eqref{eqn:mu_1}, if an $n$-element support $S$ has common values $y_1\le \cdots \le y_n$, then $\mu(S)=\sum_{r=1}^n (y_n-y_r)$.
Moreover, since every feasible complete allocation has an
$n$-element support, the definitions of $\mu$ and $\mathrm{OPT}$
give
\[
\mathrm{OPT}(\I)=
\min\{\mu(S): S\subseteq H,\ |S|=n\}.
\]

Among the $n$-element supports containing the value-$1$ house,
the minimum subsidy is obtained by omitting the value-$0$ house.
Indeed, every such support has maximum value $1$, so by
\eqref{eqn:mu_1} its subsidy equals the sum of the gaps between
$1$ and the remaining $n-1$ values. This quantity is minimized
when those remaining values are as large as possible, namely when
all $n-1$ value-$\delta$ houses are included. The resulting
subsidy is
\[
\mu(S)
=
\sum_{r=1}^{n-1}(1-\delta)
=
(n-1)(1-\delta).
\]

Now consider a support that contains both the value-$1$ house and
the value-$0$ house. Since the support contains exactly $n$
houses out of the $n+1$ available houses, it must omit one of the
$n-1$ value-$\delta$ houses. Its values are therefore
\[
0,\underbrace{\delta,\dots,\delta}_{n-2\text{ houses}},1.
\]
Applying \eqref{eqn:mu_1} gives $\mu(S)
=
1+(n-2)(1-\delta)
=
(n-1)(1-\delta)+\delta$,
which is strictly larger than $(n-1)(1-\delta)$.

Therefore, every feasible complete allocation $\A$ whose support
contains the value-$1$ house satisfies
\[
\mathrm{sub}(\A)
\ge
\mu(\mathrm{supp}(\A))
\ge
(n-1)(1-\delta),
\]
where the first inequality follows from the definition of $\mu$.

It remains to consider supports that avoid the value-$1$ house.
There is only one such $n$-element support: it consists of the
$n-1$ value-$\delta$ houses together with the value-$0$ house.
Its values are
\[
0,\underbrace{\delta,\dots,\delta}_{n-1\text{ houses}},
\]
so \eqref{eqn:mu_1} gives us $\mu(S) = \delta$.

Since $\delta<1/2$ and $n\ge2$, $\delta<(n-1)(1-\delta)$.
Hence, this support uniquely minimizes subsidy, and therefore $\mathrm{OPT}(\I)=\delta>0$.

Combining the welfare argument with the subsidy comparison, every
feasible complete allocation $\A$ with
$W(\A)\ge \rho W^*(\I)$ satisfies
\[
\mathrm{sub}(\A)
\ge
(n-1)(1-\delta)
>
R\delta
=
R\cdot \mathrm{OPT}(\I).
\]
This proves the theorem. 
\end{proof}

\subsection{Proof of Theorem~\ref{thm:no-online-competitive-ratio}}

\noonlinecompetitiveratio*

\begin{proof}
    
We use the identical-utilities formula in \eqref{eqn:mu_1}. Moreover, since $n=2$, for every
complete allocation $\A$ with $\mathrm{supp}(\A)=\{h,g\}$, we have $\mathrm{sub}(\A)=|u(h)-u(g)|$, i.e., for two agents with a common utility function $u$, the subsidy of a complete
allocation is exactly the absolute difference between the two common
values in its support.
Indeed, assume without loss of generality that $u(h)\le u(g)$.
Let $i$ be the agent with $a_i=h$, and let $j$ be the agent
with $a_j=g$. For any envy-eliminating subsidy vector $s$,
Definition~\ref{def:envy-freeness-with-subsidy} applied to agent $i$'s envy toward agent $j$ gives us $u(h)+s_i \ge u(g)+s_j$,
and hence $s_i-s_j \ge u(g)-u(h)$.
Since $s_j\ge 0$, this implies $s_i+s_j \ge u(g)-u(h)$.
Thus, every envy-eliminating subsidy vector has total subsidy at
least $u(g)-u(h)$. Conversely, the subsidy vector with
$s_i=u(g)-u(h)$ and $s_j=0$ is envy-eliminating. Hence, the
minimum total subsidy is $u(g)-u(h)=|u(h)-u(g)|$.

\paragraph{Deterministic Adaptive Lower Bound.} We first prove the deterministic adaptive lower bound. Fix a deterministic
valid online algorithm $A$, and fix $\varepsilon \in (0,1/4)$. The
adversary first releases three houses $h_1,h_2,h_3$ with common values $u(h_1)=0$, $u(h_2)=\frac12$, $u(h_3)=1$.

Since $A$ is valid and does not know whether the sequence stops after this prefix, we have that $|\mathrm{supp}(\A^3)|=\min\{3,2\}=2$.
Hence, exactly one of the three prefix houses is not in
$\mathrm{supp}(\A^3)$. This house is permanently discarded. Let
$d\in\{0,\frac12,1\}$ be its common value.

The adversary then releases a fourth house $h_4$ with common value
\begin{equation}
    \label{eq:fourth house}
    u(h_4)=d' \coloneq 
\begin{cases}
\varepsilon, & d=0,\\[2mm]
\frac12+\varepsilon, & d=\frac12,\\[2mm]
1-\varepsilon, & d=1.
\end{cases}
\end{equation}
All four common values lie in $[0,1]$. The offline optimum may choose
any two of the four houses. The pair consisting of the prefix house of
value $d$ and the fourth house of value $d'$ has absolute difference
$\varepsilon$. Every other two-house support has absolute difference at
least $\frac12-\varepsilon>\varepsilon$. Therefore, $\mathrm{OPT}(\I)=\varepsilon>0$.

At round $4$, even under the permissive update rule, the final support chosen by $A$ must be contained in $\mathrm{supp}(\A^3)\cup\{h_4\}$.
In particular, $A$ cannot use the discarded prefix house of value $d$.
Since $A$ is valid, its final allocation is complete, so the two common
values in its final support form a two-element subset of $\left(\left\{0,\frac12,1\right\}\setminus\{d\}\right)\cup\{d'\}$.
Every such two-element subset has absolute difference at least
$\frac12-\varepsilon$. Hence, letting $A(\I)$ be the final allocation produced by $A$ on instance $\I$, we have $\mathrm{sub}(A(\I))\ge \frac12-\varepsilon$,
and so
\[
\frac{\mathrm{sub}(A(\I))}{\mathrm{OPT}(\I)}
\ge
\frac{\frac12-\varepsilon}{\varepsilon}.
\]
As $\varepsilon\rightarrow 0$, this lower bound diverges. Thus, no
deterministic valid online algorithm has a bounded competitive ratio
against an adaptive adversary.

\paragraph{Randomized Non-Adaptive Lower Bound.} We next prove the randomized non-adaptive lower bound. Fix a randomized
valid online algorithm $A$, and again fix $\varepsilon\in(0,1/4)$.
Consider the common three-house prefix with values $0,\frac12,1$. Running $A$ on this instance induces a
distribution over the house discarded after the third round.
For each $z\in\{0,\frac12,1\}$, let $p_z$ be the probability, over
the internal randomness of $A$, that the prefix house of common value
$z$ is discarded, i.e., it is not in $\mathrm{supp}(\A^3)$ after processing this instance. By
randomized validity on the length-three instance, the allocation after
this prefix is complete with probability one. Therefore, exactly one of
the three prefix houses is discarded with probability one, and $\sum_{z\in\{0,\frac12,1\}} p_z = 1$.
Hence, there exists $d\in\{0,\frac12,1\}$ such that $p_d\ge 1/3$.

The probabilities $p_z$ are determined solely by the
algorithm and the common three-house prefix. In particular, they do
not depend on the value of any future fourth house. Thus, the choice of
$d$ does not depend on the realized random choices of $A$. A non-adaptive adversary can therefore inspect
this distribution and choose the value $d$ above before the
algorithm's random choices are realized. Consequently, a
non-adaptive adversary can fix in advance a four-house
instance obtained by appending a fourth house $h_4$ whose common value
$d'$ is defined by the same display as in our proof for the deterministic adaptive lower bound (see \eqref{eq:fourth house}).

For this fixed instance, the same gap calculation gives us $\mathrm{OPT}(\I)=\varepsilon>0$.
With probability at least $1/3$, the prefix house of common value $d$
has already been discarded before $h_4$ arrives. Conditional on this
event, even the permissive update rule forbids $A$ from using that
house at round $4$. Hence, every available complete final support has
absolute difference at least $\frac12-\varepsilon$. Since subsidies are
non-negative,
\[
\mathbb{E}[\mathrm{sub}(A(\I))]
\ge
p_d\left(\frac12-\varepsilon\right)
\ge
\frac13\left(\frac12-\varepsilon\right).
\]
Accordingly,
\[
\frac{\mathbb{E}[\mathrm{sub}(A(\I))]}{\mathrm{OPT}(\I)}
\ge
\frac{\frac13\left(\frac12-\varepsilon\right)}{\varepsilon},
\]
which diverges as $\varepsilon\rightarrow 0$. Hence, no randomized valid
online algorithm has a bounded competitive ratio against a non-adaptive
adversary.

The argument gave the algorithm the permissive update rule throughout.
Therefore, the lower bounds also hold under any more restrictive update
rule, including $2$-flexibility.

\end{proof}

\section{Omitted Proofs and Additional Results from Section~\ref{sec:positive-subsidy}}

\subsection{Proof of Theorem~\ref{thm:one-extra-house}}

\oneextrahouse*

\begin{proof}
We construct the following algorithm. First, fix once and for all a deterministic total order on feasible allocations of every finite prefix $H^t$ (e.g., a lexicographic order using the arrival indices of houses). The algorithm uses this order only to break ties among allocations that are feasible with respect to the houses currently known to the algorithm.
For the first $n$ arrivals, the algorithm runs Algorithm~\ref{alg:active-max-weight-chain} with this
fixed deterministic rule used to break any remaining ties. If $h_{n+1}$
arrives, it performs the update described below, again breaking any
remaining ties by the same fixed deterministic order; after setting $\A^{n+1}$, it sets $D^{n+1}\coloneq H^{n+1}\setminus \mathrm{supp}(\A^{n+1})$.
For every later arrival $h_t$ with $t>n+1$, the algorithm sets
$\A^t=\A^{t-1}$ and $D^t\coloneq H^t\setminus \mathrm{supp}(\A^t)$.
Thus, all later arrivals are rejected, the algorithm is deterministic,
respects the discarded-house convention from Section~\ref{sec:recourse}, and is defined
on every input sequence.

We begin by recording the properties of the first $n$
rounds inherited from Algorithm~\ref{alg:active-max-weight-chain}.
These properties will be used both to prove subsidy
optimality and to verify the recourse bound. By Theorem~\ref{thm:active-chain}, applied to the length-$n$ prefix, each of the
first $n$ updates is a chain update of length at most $n$, and the
round-$n$ allocation $\A^n$ maximizes $\sum_{i\in N} u_i(b_i)$ over all complete allocations $\mathbf{b}$ with support $H^n$. Equivalently,
$E(\A^n)$ is a maximum-weight perfect matching between $N$ and $H^n$, where edge $(i,h)$ has weight $u_i(h)$.

Moreover, during the first $n$ rounds, Algorithm~\ref{alg:active-max-weight-chain} assigns exactly $t$ houses from the $t$-element set $H^t$ at each round $t \le n$. Hence
\[
\mathrm{supp}(\A^t)=H^t
\quad\text{and}\quad
D^t= \varnothing
\quad\text{for every } t \le n .
\]
In particular, no house has been discarded before round $n+1$.

We next establish subsidy optimality for the cases
$m=n$ and $m=n+1$. We then show that the transition at
round $n+1$ is either a rejection or a chain update of
length at most $n$, and finally verify that the algorithm
is valid.

If the sequence stops at round $n$, then every feasible complete allocation has support $H^n$. Hence, the minimum possible subsidy among feasible complete allocations is $\mu(H^n)$.
Since $E(\A^n)$ is a maximum-weight perfect matching on this support,
Lemma~\ref{lem:fixed-support} gives us $\mathrm{sub}(\A^n)=\mu(H^n)=\OPT(\I)$.

It remains to consider the case $m=n+1$. Indeed, suppose that $h_{n+1}$ arrives. The algorithm
evaluates every possible $n$-element support obtainable by
discarding one house from $H^{n+1}$, and chooses one with
minimum subsidy value. For each $g\in H^{n+1}$, let $S_g \coloneq  H^{n+1}\setminus\{g\}$, and define $\mu^* \coloneq  \min_{g\in H^{n+1}} \mu(S_g)$.
Each set $S_g$ is a candidate support for a complete final
allocation when $m=n+1$, and the quantity $\mu(S_g)$ is the
minimum subsidy achievable on that support.
Moreover, $\mu(S_g)$ depends only on the utilities of houses in $S_g$, all of
which are known at round $n+1$, and by Lemma~\ref{lem:fixed-support} each $S_g$ admits a maximum-weight perfect matching attaining $\mu(S_g)$.

Among all feasible complete allocations $\mathbf{b}$ such that $\mathrm{supp}(\mathbf{b}) = S_g$
for some $g \in H^{n+1}$ with $\mu(S_g) = \mu^*$, and such that $E(\mathbf{b})$ is a maximum-weight perfect matching on its support, the algorithm chooses one minimizing $|E(\mathbf{b})\triangle E(\A^n)|$, breaking any remaining ties by the fixed deterministic order. 
It then sets
$\A^{n+1}\coloneq \mathbf{b}$. This choice uses only utilities of houses in $H^{n+1}$,
which are known at round $n+1$.

For an instance with $m = n+1$, every feasible complete allocation omits exactly one house from $H^{n+1}$, and hence has support $S_g$ for some $g \in H^{n+1}$. For every such allocation $\mathbf{c}$, $\mathrm{sub}(\mathbf{c}) \ge \mu(\mathrm{supp}(\mathbf{c}))$.

Conversely, for each $g \in H^{n+1}$, Lemma~\ref{lem:fixed-support} gives a feasible complete allocation $\mathbf{c}_g$ with support $S_g$ and $\mathrm{sub}(\mathbf{c}_g)=\mu(S_g)$.
Therefore, $\mathrm{OPT}(\I) = \min_{g\in H^{n+1}} \mu(S_g) = \mu^*$.
By the choice of $\A^{n+1}$, its support is some $S_g$ with $\mu(S_g)=\mu^*$,
and $E(\A^{n+1})$ is a maximum-weight perfect matching on this support. Thus, Lemma~\ref{lem:fixed-support} gives us $\mathrm{sub}(\A^{n+1})=\mu(S_g)=\mu^*=\OPT(\I)$.

It remains to show that the transition from $\A^n$ to $\A^{n+1}$ is either
a rejection of $h_{n+1}$ or a chain update of length at most $n$. Let $M \coloneq  E(\A^n)$ and $M' \coloneq  E(\A^{n+1})$.

First, suppose that $h_{n+1}\notin \mathrm{supp}(\A^{n+1})$. Then the
selected support is $S_{h_{n+1}}=H^n$, so $\mu(H^n)=\mu^*$. Since
$E(\A^n)$ is a maximum-weight perfect matching on $H^n$, the allocation $\A^n$ is one of the allocations considered by the round-$n+1$ selection. Its symmetric-difference distance from itself is $0$, i.e., $|E(\A^n)\triangle E(\A^n)| = 0$. Hence, the tie-breaking rule selects $\A^{n+1}=\A^n$, so the update is a rejection of $h_{n+1}$.

Now suppose that $h_{n+1}\in \mathrm{supp}(\A^{n+1})$. Then there is
exactly one old house $g\in H^n$ such that $g\notin \mathrm{supp}(\A^{n+1})$.
Since $h_{n+1}$ is incident to an edge of $M'$ and to no edge of $M$,
while $g$ is incident to an edge of $M$ and to no edge of $M'$,
the symmetric difference $M\triangle M'$ is nonempty.

Since $M\triangle M'$ is the symmetric difference of two matchings, every
connected component of $M\triangle M'$ is an alternating path or an alternating
cycle. Moreover, both $M$ and $M'$ match every agent exactly once, while
their house supports differ exactly by replacing the old house $g$ with
$h_{n+1}$. Thus, every agent has degree zero or two in $M\triangle M'$, and
the only house vertices of degree one are $h_{n+1}$ and $g$.
Hence, $h_{n+1}$ and $g$ are the
endpoints of the same alternating path, and every connected component of
$M\triangle M'$ that does not contain $h_{n+1}$ is an alternating cycle.

We claim that no such cycle exists. Suppose, for contradiction, that $C$
is an alternating cycle in $M\triangle M'$ that avoids $h_{n+1}$. Since
$C$ also avoids $g$, replacing the $M$-edges of $C$ by the
$M'$-edges of $C$ preserves the support $H^n$, and replacing the
$M'$-edges of $C$ by the $M$-edges of $C$ preserves the support
$\mathrm{supp}(\A^{n+1})$. We now compare the welfare contribution of the two matchings on the component $C$. Specifically, let
\[
\Delta(C)
\coloneq 
\sum_{(i,h)\in C\cap M'} u_i(h)
-
\sum_{(i,h)\in C\cap M} u_i(h).
\]
Suppose, for contradiction, that $\Delta(C)\ne 0$. We distinguish between two cases:
\begin{description}
    \item[Case 1:] If $\Delta(C)>0$, then replacing the $M$-edges of $C$ by the
    $M'$-edges of $C$ gives a higher-weight perfect matching on $H^n$,
    contradicting the choice of $\A^n$.

    \item[Case 2:]  If $\Delta(C)<0$, then replacing the
    $M'$-edges of $C$ by the $M$-edges of $C$ gives a higher-weight
    perfect matching on $\mathrm{supp}(\A^{n+1})$, contradicting the
    choice of $\A^{n+1}$ as a maximum-weight perfect matching on its support.
\end{description}

Therefore, $\Delta(C)=0$. We now show that the existence of such a component $C$ contradicts the tie-breaking rule.
Let $\widehat M \coloneq  (M' \setminus (C \cap M')) \cup (C \cap M)$.
Then $\widehat M$ is a perfect matching with support
$\mathrm{supp}(\A^{n+1})$ and with the same weight as $M'$. Hence, it is
still a maximum-weight perfect matching on the same support as $\A^{n+1}$, whose $\mu$-value is
$\mu^*$. Let $\widehat{\A}$ be the complete allocation with
$E(\widehat{\A})=\widehat M$. Then $\widehat{\A}$ is one of the allocations
considered by the round-$(n+1)$ selection. Moreover, since $C$ is a nonempty
connected component of $M\triangle M'$,
\begin{align*}
    |E(\widehat{\A})\triangle E(\A^n)|&=|\widehat M\triangle M|
   < |M'\triangle M| = |E(\A^{n+1})\triangle E(\A^n)|,
\end{align*}
contradicting the tie-breaking rule.
Therefore, no alternating cycle of
$M\triangle M'$ avoids $h_{n+1}$.

It follows that the only nonempty component of $M\triangle M'$ is the alternating path whose endpoints are $h_{n+1}$ and $g$. Starting from $h_{n+1}$, the first edge of this path belongs to $M'$, and the edges then alternate between $M$ and $M'$. Let $i_1,\dots,i_\ell$ be the agents encountered in this order. These agents are distinct and $\ell\le n$. The path gives us
\[
a^{n+1}_{i_1}=h_{n+1},\quad
a^{n+1}_{i_r}=a^n_{i_{r-1}}\quad\text{for every } r=2,\dots,\ell,
\]
and
\[
a^{n+1}_j=a^n_j
\quad\text{for every } j\notin\{i_1,\dots,i_\ell\}.
\]
Moreover, since $\A^n$ is complete, $a^n_{i_r}\neq \varnothing$ for every $r<\ell$. The final old house on the path is $a^n_{i_\ell}=g$, and since $g\notin \mathrm{supp}(\A^{n+1})$, it belongs to $D^{n+1}=H^{n+1}\setminus \mathrm{supp}(\A^{n+1})$.
Thus, the update is a chain update of length $\ell\le n$.

In summary, our constructed algorithm satisfies the following properties:
\begin{enumerate}
    \item \textbf{$n$-Flexibility:} The first $n$ updates are $n$-flexible by Theorem~\ref{thm:active-chain}, the
    update at round $n+1$ is either a rejection or a chain update
    of length at most $n$, and all later updates are rejections. Hence, the algorithm is $n$-flexible.

    \item \textbf{Validity:} Moreover, for every $t<n$ the first phase assigns exactly $t$
    houses, and from round $n$ onward the allocation assigns exactly
    $n$ houses. Hence $|\mathrm{supp}(\A^t)|=\min\{t,n\}$ for every round $t$, and in particular the final allocation is complete
    for every input sequence with $m\ge n$. Thus, the algorithm is valid.

    \item \textbf{Subsidy-Optimality:} The subsidy-optimality for $m\in\{n,n+1\}$ was shown above.
\end{enumerate}
\end{proof}

\subsection{Additive Guarantees for Identical Utilities}
\label{app:identical}

Throughout this section we assume $u(h)\in[0,1]$. This
normalization is essential for additive guarantees, since scaling all
utilities by $U>0$ scales both $\mathrm{sub}(\A)$ and
$\OPT(\I)$ by $U$. Consequently, every additive bound proved below
scales linearly with $U$.

For $\beta\ge0$, an online algorithm has additive loss at most
$\beta$ if $\mathrm{sub}(\A(\I))
\le
\mathrm{OPT}(\I)+\beta$ for every instance $I$ in the class.

\begin{theorem}
\label{thm:identical-two}
Assume $n=2$ and identical utilities in $[0,1]$. Against an
adaptive adversary, the best worst-case additive loss obtainable by a
deterministic valid online algorithm is $1/2$. Moreover, this loss
is attained by a deterministic valid $1$-flexible algorithm.
\end{theorem}
\begin{proof}
For two agents with common utility function $u$, any
complete allocation $\A$ with $\mathrm{supp}(\A)=\{h,g\}$ satisfies
\begin{equation} \label{eqn:additive_1}
    \mathrm{sub}(\A)=|u(h)-u(g)|.
\end{equation}
Indeed, assume without loss of generality that $u(h)\le u(g)$.
Let $i$ be the agent with $a_i=h$, and let $j$ be the agent
with $a_j=g$. For any envy-eliminating subsidy vector $s$,
Definition~\ref{def:envy-freeness-with-subsidy} applied to agent $i$'s envy toward agent $j$ gives $u(h)+s_i \ge u(g)+s_j$,  and hence $s_i-s_j \ge u(g)-u(h)$.
Since $s_j\ge 0$, this implies $s_i+s_j \ge u(g)-u(h)$.
Thus, every envy-eliminating subsidy vector has total subsidy at
least $u(g)-u(h)$. Conversely, the subsidy vector with
$s_i=u(g)-u(h)$ and $s_j=0$ is envy-eliminating. Hence, the
minimum total subsidy is $u(g)-u(h)=|u(h)-u(g)|$.

\paragraph{Lower Bound.} We first prove the lower bound. Fix a deterministic valid online
algorithm $A$ and $\varepsilon\in(0,1/4)$. The adversary first releases $h_1,h_2,h_3$ with common values $u(h_1)=0$, $u(h_2)=\frac12$, and $u(h_3)=1$.

Since the adversary could terminate after round $3$, $A$'s
validity implies that the round-$3$ allocation must satisfy $|\mathrm{supp}(\A^3)|=\min\{3,2\}=2$.
Hence, exactly one of $h_1,h_2,h_3$ belongs to
$D^3=H^3\setminus \mathrm{supp}(\A^3)$. Let
$d\in\{0,\frac12,1\}$ be its common value.

The adaptive adversary then releases a fourth house $h_4$ with common value
\[
u(h_4)=d' \coloneq 
\begin{cases}
\varepsilon, & d = 0,\\
\frac{1}{2}+\varepsilon, & d = \frac{1}{2},\\
1-\varepsilon, & d = 1,
\end{cases}
\]
and terminates the sequence. Let $\I$ be the resulting four-house instance.
Since $\varepsilon < 1/4$, all values lie in $[0,1]$.

The offline optimum can select the two houses with common values
$d$ and $d'$, whose gap is $\varepsilon$. Moreover, among the
two prefix values different from $d$, the value closest to $d'$
is at distance $1/2-\varepsilon$: this distance is
$\frac12-\varepsilon$ when $d=0$, $1-(\frac12+\varepsilon)$
when $d=\frac12$, and $(1-\varepsilon)-\frac12$ when $d=1$.
All other gaps are larger. Since $\varepsilon<1/4$, we have
$1/2-\varepsilon>\varepsilon$. Therefore, by
\eqref{eqn:additive_1}, $\OPT(\I)=\varepsilon$.

By Section~\ref{sec:recourse}, $D^3=H^3\setminus \mathrm{supp}(\A^3)$ and
discarded houses cannot be recovered, i.e., $D^3\subseteq D^4$.
Therefore, the prefix house with common value $d$ cannot belong to
$\mathrm{supp}(\A^4)$. Since the only new house after round 3 is
$h_4$, we have $\mathrm{supp}(\A^4) \subseteq \mathrm{supp}(\A^3) \cup \{h_4\}$.
Since $A$ is valid and $n=2$, $\A^4$ is complete, so the common
values of the houses in $\mathrm{supp}(\A^4)$ form a
two-element subset of $(\{0,\frac12,1\}\setminus\{d\})\cup\{d'\}$.
Every such two-element subset has gap at least $1/2-\varepsilon$. 
Thus, by \eqref{eqn:additive_1}, $\mathrm{sub}(A(\I))=\mathrm{sub}(\A^4)\ge \frac12-\varepsilon$.
Given any $\beta\in[0,1/2)$, choose $0<\varepsilon<
\min \{\frac14,\frac{1/2-\beta}{2} \}$.
For the constructed instance,
\[
\mathrm{sub}(A(\I))-\OPT(\I)
\ge \left(\frac{1}{2}-\varepsilon\right)-\varepsilon
= \frac{1}{2}-2\varepsilon
> \beta.
\]
Thus, no deterministic valid online algorithm can guarantee additive
loss strictly below $1/2$ against an adaptive adversary.

\paragraph{Upper Bound.} For the upper bound, consider the following deterministic algorithm. Fix a
deterministic tie-breaking rule over two-element subsets of arrived houses.
At round $1$, assign $h_1$ to agent $1$. At round $2$, assign $h_2$ to
agent $2$ and leave agent $1$ assigned to $h_1$. For every round $t \ge 3$,
choose a two-element set $S_t \subseteq \mathrm{supp}(\A^{t-1}) \cup \{h_t\}$
minimizing $|u(h)-u(g)|$ over all two-element subsets $\{h,g\} \subseteq \mathrm{supp}(\A^{t-1})\cup\{h_t\}$, breaking ties by the fixed rule. If $h_t \notin S_t$, reject $h_t$. If $h_t\in S_t$, then, since $h_t\notin \mathrm{supp}(\A^{t-1})$
and both $S_t$ and $\mathrm{supp}(\A^{t-1})$ have size $2$,
the set $\mathrm{supp}(\A^{t-1})\setminus S_t$ consists of a
unique house; call it $g$. Assign $h_t$ to the agent who held
$g$ under $\A^{t-1}$, and leave the other agent's assignment
unchanged.

The algorithm is valid and respects the discarded-house convention
from Section~\ref{sec:recourse}. After round $1$ it assigns one house, and after
every round $t\ge 2$ it assigns two houses. Hence, $|\mathrm{supp}(\A^t)|=\min\{t,2\}$ for every round $t$. Moreover, for every $t\ge 3$, the new
allocation uses only houses in $\mathrm{supp}(\A^{t-1})\cup\{h_t\}$.
Therefore, with $D^t=H^t\setminus \mathrm{supp}(\A^t)$ as in
Section~\ref{sec:recourse}, we have $D^{t-1}\subseteq D^t$ for every $t\ge 3$.
Thus, no discarded house is ever recovered. Since $m\ge n=2$, the
final allocation $\A^m$ is complete.

The algorithm is also $1$-flexible. The updates in rounds $1$ and
$2$ are chain updates of length $1$, each terminating at a
previously unassigned agent. For $t\ge 3$, if $h_t\notin S_t$, the
algorithm rejects $h_t$. If $h_t\in S_t$, exactly one agent receives
$h_t$, the other agent keeps her previous house, and the replaced old
house is discarded. This is a chain update of length $1$ in the sense
of Definition~\ref{def:chain-update}.

It remains to bound the subsidy. If $m=2$, every complete allocation
has support $H^2$, so by \eqref{eqn:additive_1} the algorithm is subsidy-optimal. Now suppose $m\ge 3$. Among the three values
$u(h_1),u(h_2),u(h_3)\in[0,1]$, two have distance at most $1/2$.
Since at round $3$ the algorithm chooses a minimum-gap pair from
$\{h_1,h_2,h_3\}$, the two houses in $\mathrm{supp}(\A^3)$
have gap at most $1/2$.

For every later round $t\ge 4$, the previous support
$\mathrm{supp}(\A^{t-1})$ is one of the candidate two-element
subsets of $\mathrm{supp}(\A^{t-1})\cup\{h_t\}$. Therefore, the
gap of the selected pair cannot increase. Hence, the final selected
pair has gap at most $1/2$, and by \eqref{eqn:additive_1}, $\mathrm{sub}(\A^m)\le \frac12$.

By \eqref{eqn:additive_1}, every complete allocation in this identical-utilities setting
has finite non-negative subsidy, and therefore $\OPT(\I)\ge 0$.
Hence, $\mathrm{sub}(\A^m) \le \OPT(\I) + \frac{1}{2}$.
This proves an additive-loss upper bound of $1/2$. Together with
the lower bound above, the optimal worst-case additive loss is
exactly $1/2$.
\end{proof}
The construction in Theorem~\ref{thm:no-online-competitive-ratio} also gives an
expected additive-loss lower bound of $1/6-o(1)$ for randomized
algorithms against a non-adaptive adversary. For $n\ge3$, the best
deterministic additive loss remains open.

\subsection{A $1$-Flexible Algorithm for Every Number of Agents}

\begin{theorem}
\label{thm:identical-general}
Assume that $u_i(h)=u(h)\in[0,1]$ for every agent $i$ and every
house $h$. There is a deterministic valid $1$-flexible online
algorithm whose final allocation $\A^m$ satisfies $\mathrm{sub}(\A^m)=\mathrm{OPT}(\I)$ if $m=n$, and $\mathrm{sub}(\A^m)\le \frac{(n-1)^2}{n}$ if $m \ge n+1$.
Consequently, its additive loss is at most $(n-1)^2/n$.
\end{theorem}
\begin{proof}
For an $n$-element set $S\subseteq H$, let the common values of the houses in $S$ be $y_1 \le y_2 \le \cdots \le y_n$.
By \eqref{eqn:mu_1}, for every complete allocation $\mathbf{b}$ with $\mathrm{supp}(\mathbf{b})=S$,
\[
\mathrm{sub}(\mathbf{b})=\mu(S)=\sum_{r=1}^n (y_n-y_r).
\]
Thus, under identical utilities, every complete allocation whose support is
$S$ has the same minimum subsidy, namely $\mu(S)$. Equivalently, the
minimum subsidy depends only on the set of assigned houses and not on
which agent receives which house. In particular, whenever two complete
allocations have the same support, they have the same subsidy.

If $n=1$, use the algorithm that assigns $h_1$ to the unique agent
and rejects every later arrival. This algorithm is valid and $1$-flexible,
and every complete allocation has subsidy $0$. Hence, $\mathrm{sub}(\A^{m})=\mathrm{OPT}(\I)=0$, so both claimed bounds hold. Thus, assume $n\ge 2$.

The algorithm is as follows. For each $t\le n$, assign $h_t$ to
the lowest-index agent who is unassigned in $\A^{t-1}$. Therefore, after
round $t\le n$, the support is $H^t$, and no house has been discarded.

For each $t>n$, consider the $n+1$ houses in $\mathrm{supp}(\A^{t-1})\cup\{h_t\}$.
Using \eqref{eqn:mu_1}, choose an $n$-element subset $S_t$ of this set minimizing $\mu(S_t)$, breaking ties by a fixed deterministic order over subsets of arrived houses. This choice uses only the common values of houses that have already arrived.

If $h_t\notin S_t$, then necessarily $S_t=\mathrm{supp}(\A^{t-1})$. Set $\A^t=\A^{t-1}$ and discard $h_t$. If $h_t\in S_t$, then, since both $S_t$ and
$\mathrm{supp}(\A^{t-1})$ have size $n$ and
$S_t\subseteq \mathrm{supp}(\A^{t-1})\cup\{h_t\}$,
there is a unique old house $g\in \mathrm{supp}(\A^{t-1})\setminus S_t$.
Let $i$ be the agent with $a_i^{t-1}=g$. Set $a_i^t=h_t$ and $a_j^t=a_j^{t-1}$ for every $j \neq i$, and discard $g$. Since only the house assigned to agent $i$ changes, feasibility is preserved.
In either case, $\mathrm{supp}(\A^t)=S_t$.

The fixed tie-breaking rule makes the algorithm deterministic. Feasibility is preserved because the first $n$ houses are assigned to distinct previously unassigned agents, and after round $n$ a non-rejection update only replaces one assigned house by the newly arrived house $h_t$. Moreover, $|\mathrm{supp}(\A^t)|=\min\{t,n\}$
for every $t$. No house is discarded before round $n+1$. For $t>n$, every future candidate set is formed only from the current support $\mathrm{supp}(\A^{t-1})$ and the new arrival $h_t$, so any rejected or replaced house is never used again. Thus, the algorithm is consistent with the discarded-house rule, and since $m\ge n$, the final allocation $\A^m$ is complete. Hence, the algorithm is valid.

The algorithm is $1$-flexible. For $t\le n$, the update is a chain update of length $1$ ending at a previously unassigned agent. For $t>n$, if $h_t\notin S_t$, the algorithm rejects $h_t$. If $h_t\in S_t$, then with $i_1=i$ we have
\[
a_{i_1}^t=h_t
\quad\text{and}\quad
a_j^t=a_j^{t-1}\text{ for every }j\ne i_1,
\]
while the old house $a_{i_1}^{t-1}=g\ne\varnothing$ is discarded. This is a chain update of length $1$.

It remains to prove the subsidy bound. If $m=n$, then every complete
allocation has support $H$. Since, under identical utilities, the
subsidy depends only on the support, the algorithm's final allocation
has subsidy $\mu(H)=\mathrm{OPT}(\I)$. Hence, the additive loss is
$0$. We may therefore assume that $m\ge n+1$.

We first prove the following selection lemma.
\begin{lemma}\label{lem:selection-bound}
    For any $n+1$ values $0\le x_1\le x_2\le \cdots \le x_{n+1}\le 1$, there is a choice of $n$ of these values whose identical-utilities
    subsidy is at most $\frac{(n-1)^2}{n}$.
\end{lemma}
\begin{proof}
Let $d_r\coloneq x_{r+1}-x_r$ for $r\in[n]$.
Consider the two choices consisting of the first $n$ values and the
last $n$ values. Their subsidies are
\[
L\coloneq\sum_{r=1}^{n-1} r d_r
\quad\text{and}\quad
R\coloneq\sum_{r=2}^{n} (r-1)d_r,
\]
respectively. Hence, the minimum subsidy over all choices of $n$
values is at most $\min\{L,R\}$. Since $\frac1nL+\frac{n-1}{n}R$
is a convex combination of $L$ and $R$, it is at least
$\min\{L,R\}$. Therefore,
\[
\min\{L,R\}\le \frac{1}{n}L+\frac{n-1}{n}R,
\]
and it is sufficient to bound the weighted average on the right-hand side. The coefficient of $d_1$ is $1/n$, which is at most $(n-1)^2/n$ because $n\ge 2$. For $2\le r\le n-1$, the coefficient of $d_r$ is
\[
\frac{r}{n}+\frac{(n-1)(r-1)}{n}
= r-1+\frac{1}{n}
\le n-2+\frac{1}{n}
= \frac{(n-1)^2}{n}.
\]
The coefficient of $d_n$ is $\frac{(n-1)^2}{n}$.
Therefore, every coefficient is at most $(n-1)^2/n$, and so
\begin{equation*}
    \min\{L,R\}
\le
\frac{(n-1)^2}{n}\sum_{r=1}^n d_r =
\frac{(n-1)^2}{n}(x_{n+1}-x_1)  \le
\frac{(n-1)^2}{n}.
\end{equation*}
This proves the lemma.
\end{proof}

We now apply Lemma~\ref{lem:selection-bound} to the first round at which
the algorithm has a choice of which house to discard, namely round $n+1$. At round $n+1$, we have
$\mathrm{supp}(\A^n)=H^n$ and no house has been discarded.
Hence, the candidate supports considered by the algorithm are precisely
the $n$-element subsets of $H^{n+1}$. By construction, the algorithm
selects a subset minimizing $\mu$ among these candidates. Applying
Lemma~\ref{lem:selection-bound} to the common values of the $n+1$
houses in $H^{n+1}$ gives us
\[
\mu(\mathrm{supp}(\A^{n+1}))
\le
\frac{(n-1)^2}{n}.
\]
For every later round $t>n+1$, the previous support
$\mathrm{supp}(\A^{t-1})$ is one admissible $n$-element candidate
support among the subsets of $\mathrm{supp}(\A^{t-1})\cup\{h_t\}$.
Since the algorithm chooses an $n$-element subset with minimum $\mu$ among
these candidates, we have $\mu(\mathrm{supp}(\A^{t}))
\le
\mu(\mathrm{supp}(\A^{t-1}))$.
Thus, the support subsidy never increases after round $n+1$, and
therefore, whenever $m\ge n+1$,
\[
\mathrm{sub}(\A^{m})
=
\mu(\mathrm{supp}(\A^{m}))
\le
\frac{(n-1)^2}{n}.
\]
By the support-only equality at the beginning of the proof, we get $\mathrm{sub}(\A^m)=\mu(\mathrm{supp}(\A^m))$.
Thus, whenever $m\ge n+1$,
\[
\mathrm{sub}(\A^m)\le \frac{(n-1)^2}{n}.
\]
Moreover, in identical-utilities instances every complete allocation has finite subsidy by \eqref{eqn:mu_1}, so $\OPT(\I)\ge 0$. Hence
\[
\mathrm{sub}(\A^m)\le \OPT(\I)+\frac{(n-1)^2}{n}.
\]
Together with the exact equality for $m=n$, the theorem follows.
\end{proof}

\subsection{An $n+1$-House Example}
\label{app:tight-one-step}
\begin{proposition}
\label{prop:selection-bound-lower}
For every $n\ge 2$, there are $n+1$ houses with common utility in
$[0,1]$ such that every $n$-element support $S$ satisfies $\mu(S)\ge \frac{(n-1)^2}{n}$.
Therefore, the constant in Lemma~\ref{lem:selection-bound} cannot be reduced.
\end{proposition}

\begin{proof}
Let the common values be
\begin{align*}
    & u(h_1)=\cdots=u(h_{n-1})=0,\\
    & u(h_n)=\frac{n-1}{n},\quad
u(h_{n+1})=1 .
\end{align*}
Fix any $n$-element support $S$. Since there are $n+1$ houses in total,
it omits exactly one house, so we distinguish the following cases.

If $S$ omits $h_{n+1}$, then its values are
$0,\dots,0,(n-1)/n$, with $n-1$ zeros. Hence
\[
\mu(S)=(n-1)\frac{n-1}{n}=\frac{(n-1)^2}{n}.
\]

If $S$ omits one of $h_1,\dots,h_{n-1}$, then its values are
$0,\dots,0,(n-1)/n,1$, with $n-2$ zeros. Hence
\[
\mu(S)=(n-2)+\frac1n=\frac{(n-1)^2}{n}.
\]

Finally, if $S$ omits $h_n$, then its values are
$0,\dots,0,1$, with $n-1$ zeros. Hence
\[
\mu(S)=n-1\ge \frac{(n-1)^2}{n}.
\]
Thus, every $n$-element support has subsidy at least
$\frac{(n-1)^2}{n}$, and this bound is attained when $S$ omits
$h_{n+1}$ or omits any one of $h_1,\dots,h_{n-1}$.
\end{proof}

\section{Omitted Proofs and Additional Results from Section~\ref{sec:learning-augmented}}

\subsection{Proof of Theorem~\ref{thm:score-prediction}}

\scoreprediction*

\begin{proof}
The algorithm chooses $\A^t$ as follows. Among all feasible partial
allocations $\mathbf{b}$ with $\mathrm{supp}(\mathbf{b})=B_t$, it chooses one maximizing utilitarian social welfare $\sum_{i\in N} u_i(b_i)$, while breaking ties by minimizing $|E(\mathbf{b})\triangle E(\A^{t-1})|$,
and then by the fixed deterministic order over feasible allocations.
Afterwards, it defines the set of houses discarded after round $t$ by: $D^t \coloneq H^t\setminus B_t = H^t\setminus \mathrm{supp}(\A^t)$.
The choice is well defined. Since $|B_t|=\min\{t,n\}\le n$, the houses
in $B_t$ can be assigned injectively to agents, so there exists at least
one feasible partial allocation whose support is exactly
$B_t$.

\paragraph{Welfare Maximization.} All utilities used in the maximization are known by
round $t$, and the fixed tie-breaking rule makes the selected allocation
unique. Thus, after
every round $t$, $\A^t$ maximizes utilitarian welfare among all feasible
allocations whose assigned houses are exactly $B_t$.

\paragraph{Validity.} Since $|B_t|=\min\{t,n\}$, we have $|\mathrm{supp}(\A^t)|=\min\{t,n\}$ for every $t$. Hence, the final allocation is complete whenever $m\ge n$, so the algorithm is valid.

\paragraph{Discarded-House Condition.} We next check the discarded-house condition from
Section~\ref{sec:recourse}. That is, fixing a round $t$, we next prove that $D^t\subseteq D^{t'}$ for every $t'>t$. If $t\le n$, then
$|B_t|=\min\{t,n\}=t$, and hence $B_t=H^t$. Therefore $D^t=H^t\setminus B_t=\varnothing$, so the condition holds trivially.

Now suppose that $t>n$, and let $h\in D^t=H^t\setminus B_t$. In particular, $|B_t|=n$.
Since $h\notin B_t$, the set $B_t$ consists of $n$ houses that are ranked
higher than $h$ in the fixed score order. The relative order of already
arrived houses never changes, so these same $n$ houses remain ranked
higher than $h$ at every later round $t'>t$. Therefore, $h\notin B_{t'}$
for every $t'>t$, and hence $D^t\subseteq D^{t'}$ for every $t'>t$.

\paragraph{$n$-Flexibility.} We next prove that the algorithm is $n$-flexible. Fix a round $t$, and let $M\coloneq E(\A^{t-1})$ and $M'\coloneq E(\A^t)$.
View $M$ and $M'$ as matchings in the bipartite graph with agent side $N$
and house side $H^t$. Since $\A^t$ avoids $D^{t-1}$, every house incident to
an edge of $M'$ is either $h_t$ or belongs to $\mathrm{supp}(\A^{t-1})$. Since $M$
and $M'$ are matchings, every connected component of $M\triangle M'$ is an
alternating path or an alternating cycle. Moreover, at most one connected
component of $M\triangle M'$ contains $h_t$.

Now, if $h_t\notin B_t$, then $B_t=B_{t-1}$. Since $\A^{t-1}$ is already
welfare-maximizing among allocations whose assigned-house set is $B_t$,
and its distance from itself is zero, the tie-breaking rule gives
$\A^t=\A^{t-1}$. Thus, the update rejects $h_t$.

Assume now that $h_t\in B_t$. Then $B_t\subseteq B_{t-1}\cup\{h_t\}$,
and if $t>n$, exactly one house in $B_{t-1}$ is not in $B_t$.

We claim that every nonempty connected component of $M\triangle M'$ contains
$h_t$. Suppose not, and let $C$ be a nonempty component not containing $h_t$.
We first show that $C$ is balanced, i.e., $|C\cap M|=|C\cap M'|$.
This property will allow us to switch along $C$ without changing the number
of assigned houses.

We distinguish between two cases.
\begin{description}
    \item[Case 1:] If $t\le n$, then $M$ matches every house in $H^{t-1}$, while $M'$ matches every house
    in $H^t$. Since $C$ avoids $h_t$, every house vertex in $C$ has degree two in
    $M\triangle M'$. Thus, $C$ is either an alternating cycle or an alternating
    path whose two endpoints are agents. In either case, $|C\cap M|=|C\cap M'|$.

    \item[Case 2:] If $t>n$, then both $M$ and $M'$ match every agent. Furthermore, because
    $h_t\in B_t$, then $M'$ uses $h_t$ and has the same
    cardinality as $M$. Since every other house used by $M'$ belongs to
    $\mathrm{supp}(\A^{t-1})$, exactly one house $g\in \mathrm{supp}(\A^{t-1})\setminus \mathrm{supp}(\A^t)$     is not used by $M'$. The
    only house vertices of degree one in $M\triangle M'$ are $h_t$ and
    that unique house $g$, while every agent has degree zero or two. Therefore, $h_t$ and $g$ are the endpoints of the unique alternating path
    component, and every component avoiding $h_t$ is an alternating cycle. As before, we have
    $|C\cap M|=|C\cap M'|.$
\end{description}
In either case, $|C\cap M|=|C\cap M'|$.

Thus, every house vertex in $C$ is matched by both $M$ and
$M'$. Consequently, replacing the $M$-edges of $C$ by the $M'$-edges of
$C$, or conversely replacing the $M'$-edges of $C$ by the $M$-edges
of $C$, preserves the set of assigned houses.

We now compare the welfare contribution of the two matchings on the component $C$. Specifically, let
\[
\Delta(C)\coloneq
\sum_{(i,h)\in C\cap M'} u_i(h)
-
\sum_{(i,h)\in C\cap M} u_i(h).
\]
Since $C$ is a connected component of $M\triangle M'$,
switching along $C$ preserves the matching property: vertices
outside $C$ are unchanged, and the alternating structure of $C$
ensures that no vertex of $C$ is incident to more than one edge
after the switch. Since $C$ is balanced, this switch also preserves
the number of assigned houses.

Suppose, for contradiction, that $\Delta(C)\neq 0$. We distinguish
between two cases.

\begin{description}
    \item[Case 1:]
    Suppose $\Delta(C)>0$. Replacing the $M$-edges of $C$ by the
    $M'$-edges of $C$ in $M$ gives a feasible allocation. Since every
    house vertex in $C$ is matched by both $M$ and $M'$, this switch
    preserves the assigned-house set, which therefore remains exactly
    $B_{t-1}$. Moreover, the switched allocation has welfare larger than
    that of $\A^{t-1}$ by $\Delta(C)>0$. This contradicts the fact that
    $\A^{t-1}$ is a maximum-welfare allocation among all feasible
    allocations whose assigned-house set is $B_{t-1}$.
    
    \item[Case 2:]
    Suppose $\Delta(C)<0$. Replacing the $M'$-edges of $C$ by the
    $M$-edges of $C$ in $M'$ gives a feasible allocation. Again, because
    every house vertex in $C$ is matched by both $M$ and $M'$, the
    assigned-house set remains exactly $B_t$. The switched allocation has
    welfare larger than that of $\A^t$ by $-\Delta(C)>0$. This
    contradicts the fact that $\A^t$ is a maximum-welfare allocation
    among all feasible allocations whose assigned-house set is $B_t$.
\end{description}

Therefore, $\Delta(C)=0$. We now show that the existence of such a component $C$ contradicts the tie-breaking rule. Let $\widehat M \coloneq
(M'\setminus (C\cap M'))\cup (C\cap M)$.
By the same switching argument, $\widehat M$ gives a feasible allocation $\widehat a$ whose assigned-house set is $B_t$. Since
$\Delta(C)=0$, this allocation has the same welfare as $\A^t$, so it is also
a maximum-welfare allocation among all allocations whose support is $B_t$. Moreover,
\[
\widehat M\triangle M
=
(M'\triangle M)\setminus ((C\cap M')\cup (C\cap M)),
\]
because the switch replaces the $M'$-edges of $C$ by the corresponding
$M$-edges. Since $C$ is nonempty, this removes at least one edge from the
symmetric difference, and therefore
\[
|\widehat M\triangle M|<|M'\triangle M|.
\]
This contradicts the tie-breaking rule used to choose $\A^t$. Hence, every
nonempty component of $M\triangle M'$ contains $h_t$.

Thus, $M\triangle M'$ has at most one nonempty component. $M\triangle M'=\varnothing$, then $E(\A^t)=E(\A^{t-1})$, and thus $\A^t=\A^{t-1}$. Since the newly arrived house $h_t$ is not in
$\mathrm{supp}(\A^{t-1})$ in this case, the algorithm rejects $h_t$. 

Otherwise, there is a unique nonempty component, and it contains $h_t$.
Since $h_t$ is incident to no edge of $M$, this component is an alternating
path with $h_t$ as an endpoint. The transition from $\A^{t-1}$ to $\A^t$ is
completely determined by this path: every assignment change occurs along the
path, while all agents outside the path keep their previous houses. Read this
path starting from $h_t$, and let $i_1,\dots,i_\ell$ be the agents
encountered along it. The agents $i_1,\dots,i_\ell$ are distinct, and
$\ell\le n$. The first edge of the path is $(i_1,h_t)\in M'$, so $a^t_{i_1}=h_t$.
For every $r=2,\dots,\ell$, the path enters $i_r$ through the house
previously held by $i_{r-1}$, so $a^t_{i_r}=a^{t-1}_{i_{r-1}}$.
Every agent outside $\{i_1,\dots,i_\ell\}$ is incident to no edge of
$M\triangle M'$, and therefore
\[
a^t_j=a^{t-1}_j
\quad\text{for every }j\notin\{i_1,\dots,i_\ell\}.
\]
Furthermore, for every $r<\ell$, $a^{t-1}_{i_r}\ne \varnothing$,
since the path continues through the house previously held by $i_r$.
If the path ends at the agent $i_\ell$, then $i_\ell$ is not incident to an
$M$-edge, so $a^{t-1}_{i_\ell}=\varnothing$. Otherwise, the path ends at the
house $a^{t-1}_{i_\ell}$; this house is not incident to any edge of $M'$;
otherwise, the path would continue. Since $\mathrm{supp}(\A^t)=B_t$, we have
$a^{t-1}_{i_\ell}\notin B_t$. Hence, $a^{t-1}_{i_\ell}\in D^t=H^t\setminus B_t$,
so it is discarded. By Definition~\ref{def:chain-update}, this is a chain
update of length $\ell\le n$. Hence the algorithm is $n$-flexible.

\paragraph{Completing the proof of Theorem~\ref{thm:score-prediction}.} At the final round, $|B_m|=n$. By construction, $\A^m$ is complete and
$E(\A^m)$ is a maximum-weight perfect matching between $N$ and $B_m$,
where edge $(i,h)$ has weight $u_i(h)$. Lemma~\ref{lem:fixed-support} therefore gives $\mathrm{sub}(\A^m)=\mu(B_m)$.
Therefore, by the definition of $\eta_{\mathrm{sc}}(p;\I)$,
\begin{equation}
    \label{eq:sub scores}
    \mathrm{sub}(\A^m)
    =
    \mu(B_m)
    =
    \mathrm{OPT}(\I)+\eta_{\mathrm{sc}}(p;\I).
\end{equation}

Finally, Lemma~\ref{lem:fixed-support} implies that the final allocation $\A^m$ is
envy-freeable. Moreover, $\A^m$ is complete because $|B_m|=n$. If, in addition, $u_i(h)\in[0,1]$ for every $i\in N$
and $h\in H$, then Lemma~\ref{lem:absolute-subsidy-bound} gives $\mathrm{sub}(\A^m)\le n-1$.
Combining this bound with the equality in \eqref{eq:sub scores} gives us
\[
\mathrm{sub}(\A^m)
\le
\min\{\mathrm{OPT}(\I)+\eta_{\mathrm{sc}}(p;\I),\,n-1\}. \qedhere
\]
\end{proof}

\subsection{Small Changes to Score Predictions}\label{app:score-perturbations}

Let $p^*$ be a correct score prediction, and let
$B_m^*$ be its final top $n$ houses. Define $\Gamma(p^*)
:=
\min_{h\in B_m^*, g\in H\setminus B_m^*}
(p^*(h)-p^*(g))$,
with the convention that
$\Gamma(p^*)=+\infty$ when $m=n$.
If $\Gamma(p^*)>0$ and $\max_{h\in H}
\lvert p(h)-p^*(h)\rvert
<
\frac{\Gamma(p^*)}{2}$,
then $p$ and $p^*$ have the same final top $n$ houses.
Consequently, the score-keeping algorithm returns subsidy
$\mathrm{OPT}(\I)$. When $\Gamma(p^*)=0$, this argument gives
no positive error bound.

\subsection{Cutoff Predictions for Common-Quality Utilities}
\label{app:cutoff}
Score predictions are general, but they require a full ranking of the houses.
For common-quality utilities, the offline support-selection problem is
one-dimensional: after sorting houses by quality, an optimal support can be
chosen as a consecutive block. This suggests replacing the full ranking by a
single cutoff prediction that identifies where this block should end.

Formally, an instance has \textit{common-quality utilities} if each house $h$ has a quality $q_h\ge 0$ and each agent $i$ has parameters $\alpha_i,\beta_i\ge 0$ such that $u_i(h)=\beta_i+\alpha_i q_h$.
Relabel the agents so that $\alpha_1\le \alpha_2\le \cdots \le \alpha_n$.
This class includes identical utilities and also allows agents to have different sensitivities to the same quality scale. In this subsection, the quality $q_h$ of an arriving house is observed when $h$ arrives, in addition to the utilities $u_i(h)$ for all $i\in N$.

The next lemma formalizes the one-dimensional structure: among common-quality
utilities, the minimum-subsidy support can always be chosen as a consecutive
block in the quality order.

\begin{restatable}{lemma}{commonqualitysupport}
\label{lem:common-quality-support}
Consider a common-quality instance and an $n$-house set $S$ with qualities $q_1\le q_2\le \cdots \le q_n$.
Then, $\mu(S)=\sum_{r=1}^{n-1} r\alpha_r(q_{r+1}-q_r)$.
Moreover, after sorting all houses in $H$ by nondecreasing quality,
breaking ties in any fixed way, there is an $n$-element set
minimizing $\mu(S)$ over all $n$-element subsets $S\subseteq H$
that consists of $n$ consecutive houses in this order.
\end{restatable}

\begin{proof}
If $n=1$, then every complete allocation gives the only
agent one house. Hence, $\mu(S)=0$, which is the empty sum, and the final statement is
immediate. Thus, assume $n\ge 2$.

Denote $S=\{g_1,\dots,g_n\}$ so that $q_{g_1}=q_1\le q_{g_2}=q_2\le \cdots \le q_{g_n}=q_n$.

By Lemma~\ref{lem:fixed-support}, $\mu(S)$ is obtained by any maximum-welfare
complete allocation using exactly the houses in $S$, together
with a minimum envy-eliminating subsidy vector for that allocation.
For any complete allocation using exactly the houses in $S$, the
terms $\beta_i$ contribute $\sum_{i\in N}\beta_i$ to the total
welfare. Therefore, welfare maximization is equivalent to maximizing $\sum_{i\in N}\alpha_i q_{a_i}$.

Among all complete allocations using exactly the houses in $S$, choose an allocation $\A$
maximizing this sum. If $i < j$ and $q_{a_i} > q_{a_j}$, then swapping
the two assigned houses changes the sum by
\[
\alpha_i q_{a_j}+\alpha_j q_{a_i}-\alpha_i q_{a_i}-\alpha_j q_{a_j}
= (\alpha_j-\alpha_i)(q_{a_i}-q_{a_j}) \ge 0.
\]
Therefore, by repeatedly swapping such pairs, we obtain another
maximum-welfare allocation in which the assigned qualities are
nondecreasing with the agent index. If equal-quality houses are not in
the order $g_1,\dots,g_n$, swap only those equal-quality houses; this
does not change any agent's utility. Hence, we may choose a
maximum-welfare allocation $\A$ such that $a_r = g_r$ for every $r\in N$.
It remains to compute the minimum total subsidy for this allocation. Let $\Delta^r\coloneq q_{r+1}-q_r$ for every $r\in[n-1]$.

Define $s_n\coloneq0$, and $s_i\coloneq\sum_{r=i}^{n-1}\alpha_r\Delta^r$ for every $i < n$.
These subsidies are non-negative. We next show that they eliminate envy.
For $i=j$, there is nothing to prove. If $i<j$, then
\begin{equation*}
    s_i-s_j =
\sum_{r=i}^{j-1}\alpha_r\Delta^r
\ge \alpha_i\sum_{r=i}^{j-1}\Delta^r =
\alpha_i(q_j-q_i)
=
u_i(g_j)-u_i(g_i).
\end{equation*}
If $i>j$, then
\begin{equation*}
    s_i-s_j=
-\sum_{r=j}^{i-1}\alpha_r\Delta^r
\ge
-\alpha_i\sum_{r=j}^{i-1}\Delta^r =
\alpha_i(q_j-q_i)
=
u_i(g_j)-u_i(g_i),
\end{equation*}
where the inequality uses $\alpha_r\le \alpha_i$ for all $r<i$.
Hence, $s_i-s_j\ge u_i(a_j)-u_i(a_i)$ for all $i,j\in N$, so
$(\A,\mathbf{s})$ is envy-free by Definition~\ref{def:envy-freeness-with-subsidy}.

Now, let $s'$ be any envy-eliminating subsidy vector for the same
allocation $\A$. For every $r\in[n-1]$, agent $r$ must not envy
agent $r+1$, so
\[
s'_r-s'_{r+1}
\ge
u_r(g_{r+1})-u_r(g_r)
=
\alpha_r(q_{r+1}-q_r)
=
\alpha_r\Delta^r .
\]
Since $s'_n\ge 0$, for every $i<n$,
\[
s'_i
=
s'_n+\sum_{r=i}^{n-1}(s'_r-s'_{r+1})
\ge
\sum_{r=i}^{n-1}\alpha_r\Delta^r .
\]
Therefore
\[
\sum_{i=1}^n s'_i
\ge
\sum_{i=1}^{n-1}\sum_{r=i}^{n-1}\alpha_r\Delta^r
=
\sum_{r=1}^{n-1} r\alpha_r\Delta^r .
\]
For the subsidy vector constructed above,
\[
\sum_{i=1}^n s_i
= \sum_{i=1}^{n-1}\sum_{r=i}^{n-1}\alpha_r\Delta^r
= \sum_{r=1}^{n-1} r\alpha_r\Delta^r .
\]
Thus, this is the minimum total subsidy for the maximum-welfare
allocation $\A$. Since $\A$ uses exactly the houses in $S$, Lemma~\ref{lem:fixed-support} gives us $\mu(S) =
\sum_{r=1}^{n-1} r\alpha_r(q_{r+1}-q_r)$.

It remains to prove the consecutive-house statement. Sort all houses in $H$ by nondecreasing quality and let $\bar h_1,\dots,\bar h_m$
be this order, where $Q_\ell \coloneq q_{\bar h_\ell}$ denotes the quality of $\bar h_\ell$. That is, their qualities satisfy $Q_1\le Q_2\le \cdots \le Q_m$,
breaking ties in any fixed way. For $r\in[n-1]$, let $c_r\coloneq r\alpha_r$.

Since $\alpha_1\le\cdots\le \alpha_n$ and all $\alpha_i\ge0$, the
sequence $c_r=r\alpha_r$ is nondecreasing: indeed,
\[
c_{r+1}-c_r=(r+1)\alpha_{r+1}-r\alpha_r
\ge (r+1)\alpha_r-r\alpha_r=\alpha_r\ge0.
\]
Hence, $c_1\le c_2\le\cdots\le c_{n-1}$.

For any chosen qualities $x_1\le\cdots\le x_n$, the formula already proved
gives
\begin{equation}
    \label{eq:expression}
    \sum_{r=1}^{n-1}c_r(x_{r+1}-x_r) =
    c_{n-1}x_n+\sum_{r=2}^{n-1}(c_{r-1}-c_r)x_r-c_1x_1 .
\end{equation}
The coefficient of $x_n$ is $c_{n-1}\ge0$, while the coefficient of $x_1$ is
$-c_1\le0$ and, for each $r=2,\dots,n-1$, the coefficient of $x_r$ is
$c_{r-1}-c_r\le0$. Namely, the coefficients of
$x_1,\dots,x_{n-1}$ are nonpositive. Therefore, with $x_n$ fixed, the expression in \eqref{eq:expression} is
nonincreasing in each of $x_1,\dots,x_{n-1}$.

Fix $j\in\{n,\dots,m\}$, and consider any $n$-house set whose
largest index in the above order is $j$. Let its indices be $k_1<k_2<\cdots<k_n=j$.
Then $k_r\le j-n+r$ for every $r\in[n]$, and hence
\begin{equation}
    \label{eq:q monotone}
    Q_{k_r}\le Q_{j-n+r}
\quad \text{for every } r\in[n].
\end{equation}
Since the expression in \eqref{eq:expression} is nonincreasing in $x_1,\dots,x_{n-1}$ with $x_n$
fixed and due to \eqref{eq:q monotone}, replacing $Q_{k_1},\dots,Q_{k_{n-1}}$ by $Q_{j-n+1},\dots,Q_{j-1}$
cannot increase the value of $\mu$. Thus, among all $n$-house sets whose
largest index is $j$, the consecutive set $\{\bar h_{j-n+1},\dots,\bar h_j\}$
has minimum $\mu$.

Now, take any $n$-house set minimizing $\mu$, and let $j$ be the
largest index of a house in that set. The consecutive set $\{\bar h_{j-n+1},\dots,\bar h_j\}$ has $\mu$-value no larger, so it also minimizes $\mu$. This proves
the final statement.
\end{proof}

Lemma~\ref{lem:common-quality-support} gives an offline description
of an optimal $n$-house set in this class: sort the houses by quality and
check the consecutive $n$-house sets. In particular, a consecutive optimal support is
determined by its largest quality. This motivates a single cutoff prediction,
which aims to identify the largest quality appearing in an optimal
minimum-subsidy support.

To formalize cutoff predictions, we first fix a deterministic tie-breaking rule $\prec$, independent of the round. We use $\prec$ to break ties between two houses that have either the same quality or the
same score. For a predicted cutoff $\tau \ge 0$, let $C(\tau)$ be the support predicted by $\tau$, i.e., the set of the $n$ highest-quality houses among $\{h\in H:q_h\le \tau\}$, whenever this set
contains at least $n$ houses, with ties broken according to $\prec$. If this
set contains fewer than $n$ houses, then $C(\tau)$ is not defined.

A cutoff $\tau^*$ is correct if $C(\tau^*)$ exists,
$\mu(C(\tau^*))=\mathrm{OPT}(\I)$, and the largest quality among the houses in
$C(\tau^*)$ is $\tau^*$. Given a predicted cutoff $\tau$, define the induced score of each house $h$ as
\[
p_\tau(h)=
\begin{cases}
q_h, & q_h\le \tau,\\
-1, & q_h>\tau.
\end{cases}
\]
The dummy score $-1$ ensures that every house above the cutoff ranks below
every house at or below the cutoff. Thus, the top $n$ houses under
the score order induced by $p_\tau$ are exactly the houses in $C(\tau)$. 

Now, consider the cutoff algorithm, which is the score-keeping algorithm from the score-prediction subsection run with these
scores. The next theorem shows that, if the cutoff correctly identifies the
largest quality appearing in an optimal support, then the cutoff
algorithm recovers the offline optimum. Moreover, overestimating a
correct cutoff increases the subsidy by at most a quantity proportional
to the prediction error.

\begin{theorem}
\label{thm:predicted-cutoff}
Assume $n\ge2$, common-quality utilities, and
$u_i(h)\in[0,1]$ for every agent $i$ and house $h$.
For every cutoff $\tau$, the cutoff algorithm is deterministic,
valid, and $n$-flexible, and its final allocation $\A^m$ satisfies $\mathrm{sub}(\A^m)\le n-1$.
If $\tau$ is correct, then $\mathrm{sub}(\A^m)=\mathrm{OPT}(\I)$.
Moreover, if $\tau^*$ is correct and
$\tau\ge\tau^*$, then
\[
\mathrm{sub}(\A^m)
\le
\min\left\{
\mathrm{OPT}(\I)
 +(n-1)\alpha_{n-1}(\tau-\tau^*),
\, n-1
\right\}.
\]
\end{theorem}

\begin{proof}
The cutoff algorithm is the score-keeping algorithm from Theorem~\ref{thm:score-prediction}
run with the scores
\[
p_\tau(h)=
\begin{cases}
q_h, & q_h\le \tau,\\
-1, & q_h>\tau.
\end{cases}
\]
Therefore, by Theorem~\ref{thm:score-prediction}, it is deterministic, valid, and
$n$-flexible. Since $u_i(h) \in [0,1]$ for every $i \in N$ and $h \in H$, Theorem~\ref{thm:score-prediction} also gives $\mathrm{sub}(\A^m) \le n-1$.

Suppose first that $\tau$ is correct. Then $C(\tau)$ exists, $\mu(C(\tau))=\OPT(\I)$,
and the largest quality among the houses in $C(\tau)$ is $\tau$.
Since $q_h\ge 0$ for every house $h$, every house with $q_h\le \tau$
has score $p_\tau(h)=q_h\ge 0$, while every house with $q_h>\tau$
has score $-1$. Among the houses with $q_h \le \tau$, the score order is exactly the quality order, with ties broken by the same rule $\prec$ used in the definition of $C(\tau)$. Since $C(\tau)$ exists, at least $n$ houses have quality at most $\tau$, and all such houses have score at least $0$, while every house with quality greater than $\tau$ has score $-1$. Therefore, the top $n$ houses under the score order induced by
$p_\tau$ are exactly the houses in $C(\tau)$. Hence, $B_m=C(\tau)$.
By Theorem~\ref{thm:score-prediction},
\[
\mathrm{sub}(\A^m)=\mu(B_m)=\mu(C(\tau))=\OPT(\I).
\]

Now suppose that $\tau^*$ is correct and $\tau\ge \tau^*$. Let $C^*\coloneq C(\tau^*)$ and $B\coloneq B_m$.
Since the largest quality in $C^*$ is $\tau^*$, we have $C^*\subseteq \{h\in H:q_h\le \tau\}$.
Hence, $C(\tau)$ exists. By the same score-order argument as above,
$B=C(\tau)$.

Write the qualities of the houses in $C^*$ and $B$ as $q_1^*\le \cdots \le q_n^*$ and $q_1\le \cdots \le q_n$,
respectively. Since $B=C(\tau)$ consists of the $n$ highest-quality houses among
$\{h \in H : q_h \le \tau\}$, and $C^*$ is an $n$-element subset of this set, we have $q_r \ge q^*_r$ for every $r \in [n]$.
Indeed, if $q_r < q^*_r$ for some $r$, then the set
$\{h \in H : q_h \le \tau\}$ contains at least $n-r+1$ houses with quality at least $q^*_r$, namely the houses of $C^*$ with qualities
$q^*_r,\dots,q^*_n$. But $B$ contains at most $n-r$ houses with quality at least $q^*_r$, since $q_1,\dots,q_r < q^*_r$. This contradicts the choice of $B=C(\tau)$ as the $n$ highest-quality houses in that set. Also, $q^*_n=\tau^*$ and $q_n \le \tau$.

For $r\in [n-1]$, let $c_r\coloneq r\alpha_r$. Since
$\alpha_1\le \cdots \le \alpha_n$ and the $\alpha_i$'s are
non-negative, we have $c_1\le c_2\le \cdots \le c_{n-1}$.
By Lemma~\ref{lem:common-quality-support}, for any sorted qualities $z_1\le \cdots \le z_n$,
\[
\sum_{r=1}^{n-1} c_r(z_{r+1}-z_r)
=
c_{n-1}z_n+\sum_{r=2}^{n-1}(c_{r-1}-c_r)z_r-c_1z_1 .
\]
Therefore,
\begin{align*}
\mu(B)-\mu(C^*)
&=
c_{n-1}(q_n-q_n^*)
+\sum_{r=2}^{n-1}(c_{r-1}-c_r)(q_r-q_r^*)
-c_1(q_1-q_1^*)\\
&\le c_{n-1}(q_n-q_n^*) \le (n-1)\alpha_{n-1}(\tau-\tau^*).
\end{align*}
The first inequality uses $q_r\ge q_r^*$ for every $r\in[n]$, together with $-c_1\le0$ and $c_{r-1}-c_r\le0$ for all $r=2,\dots,n-1$,
from which we infer that all coefficients of $q_1-q_1^*,\,q_2-q_2^*,\,\dots,\,q_{n-1}-q_{n-1}^*$
are nonpositive.

Since $\tau^*$ is correct, $\mu(C^*)=\OPT(\I)$. By
Theorem~\ref{thm:score-prediction},
\[
\mathrm{sub}(\A^m)=\mu(B)
\le \mathrm{OPT}(\I)+(n-1)\alpha_{n-1}(\tau-\tau^*).
\]
Combining this with the normalized bound
$\mathrm{sub}(\A^m)\le n-1$ gives
\[
\mathrm{sub}(\A^m)
\le
\min\{\mathrm{OPT}(\I)+(n-1)\alpha_{n-1}(\tau-\tau^*),\,n-1\}.
\]
\end{proof}

If an upper bound on the cutoff error is available, the one-sided condition $\tau\ge \tau^*$ can be enforced by shifting the prediction upward. Given a predicted cutoff $\widehat{\tau}$ and a number $\varepsilon\ge 0$, run the cutoff algorithm with $\tau\coloneq\widehat{\tau}+\varepsilon$.
If $\tau^*$ is correct and $|\widehat{\tau}-\tau^*|\le \varepsilon$, then $\tau^*\le \tau\le \tau^*+2\varepsilon$.
Thus, Theorem~\ref{thm:predicted-cutoff} gives
\begin{equation}    
    \label{eq:two-sided}
    \mathrm{sub}(\A^m)
    \le
    \min\{\mathrm{OPT}(\I)+2(n-1)\alpha_{n-1}\varepsilon,\, n-1\}.
\end{equation}
In particular, when $\varepsilon=0$, the shifted rule is exact. We next give a matching example showing that the
linear dependence on the cutoff error, $(n-1)\alpha_{n-1}\varepsilon$, is tight.

\begin{proposition}
\label{prop:cutoff-error-tight}
Fix $n\ge 2$, parameters $0\le \alpha_1\le\cdots\le \alpha_n$,
and $\varepsilon>0$ with $\alpha_n\varepsilon\le 1$. For the cutoff
algorithm in Theorem~\ref{thm:predicted-cutoff}, there is a
normalized common-quality instance with correct cutoff $\tau^*=0$
such that, when the prediction is $\tau=\varepsilon$, $\mathrm{sub}(\A^m)-\OPT(\I)=(n-1)\alpha_{n-1}\varepsilon$.
\end{proposition}

\begin{proof}
Set $\beta_i=0$ for every $i$. Let there be $n$ houses of quality
$0$ and one house of quality $\varepsilon$. The instance is normalized
because $\alpha_i q_h \le \alpha_n\varepsilon \le 1$ for every $i$
and $h$.

The cutoff $\tau^*=0$ is correct: selecting the $n$ zero-quality
houses gives subsidy $0$, so $\mathrm{OPT}(\I)=0$. Run the cutoff
algorithm with prediction $\tau=\varepsilon$. Then all $n+1$ houses have quality at most $\tau$, so the cutoff algorithm
selects the $n$ highest-quality houses among them, namely the
$\varepsilon$-quality house and $n-1$ zero-quality houses. By
Lemma~\ref{lem:common-quality-support}, the resulting subsidy is $(n-1)\alpha_{n-1}\varepsilon$.
Thus, the additive loss is exactly
$(n-1)\alpha_{n-1}\varepsilon$.
\end{proof}

The upward shift is necessary for a two-sided error guarantee. Appendix \ref{sec:underestimate-costly} shows that, without it, even a very small underestimate of a correct cutoff can lead to additive loss arbitrarily close to $n-1$.

\subsection{Proof of Theorem~\ref{thm:predicted-support-recourse-tight}}

\predictedsupportrecoursetight*

\begin{proof}
Fix $n \ge 2$ and $k<n$. Suppose, for contradiction, that there is a
deterministic valid $k$-flexible online algorithm $A$, which is given the
score $p(h_t)$ when $h_t$ arrives, such that $\mathrm{sub}(\A^m)=\mu(B_m)$ for every input sequence with $m=n+1$ and every score prediction.

Choose $\delta>0$ such that $\delta \sum_{r=1}^{n-1} r^2 < 1$.
Define $x_1\coloneq0$ and $x_t\coloneq1+(t-2)\delta$ for $t=2,\dots,n+1$.
The adversary first releases houses $h_1,\dots,h_n$. For every $i\in N$
and every $t\in[n]$, the adversary sets $u_i(h_t)=i x_t$.
Equivalently, this is a common-quality instance with
$q_{h_t}=x_t$, $\alpha_i=i$, and $\beta_i=0$. The agents are
already ordered so that $\alpha_1\le \cdots \le \alpha_n$, as required
in Lemma~\ref{lem:common-quality-support}. Give $h_1$ score $0$, and give each of
$h_2,\dots,h_n$ score $1$.

Since $A$ is valid and the algorithm cannot know whether the sequence
stops after round $n$, Section~\ref{sec:recourse} gives
\[
\left|\mathrm{supp}(\A^n)\right|= \left|\{a_i^n : i\in N,\ a_i^n\neq \varnothing\}\right|=n.
\]
Only the houses $h_1,\dots,h_n$ have arrived, and the allocation is
feasible, so all of them are assigned after round $n$. For the $n$-element set $H^n$, the unique maximum-weight perfect
matching between $N$ and $H^n$ assigns $h_i$ to agent $i$ for
every $i\in N$. To see this, suppose $i<j$, agent $i$ receives
$h_b$, agent $j$ receives $h_a$, and $a<b$. Since $x_a<x_b$,
swapping these two houses changes the welfare by
\[
(i x_a+jx_b)-(i x_b+jx_a)=(j-i)(x_b-x_a)>0.
\]
Hence, any maximum-weight perfect matching has no such pair. Since
$x_1,\dots,x_n$ are strictly increasing, the unique maximum-weight
perfect matching is $\{(i,h_i):i\in N\}$.

We first show that, after round $n$, algorithm $A$ must have $a_i^n=h_i$ for every $i \in N$.
Suppose not. After observing the round-$n$ allocation,
the adversary releases one more house $h_{n+1}$ with score
$-1$. Choose $M$ such that $M>x_n+\mu(H^n)$, and set $u_i(h_{n+1})=iM$ for every $i \in N$.
For this continuation, the final top $n$ houses are exactly $B_{n+1}=H^n$.

If $h_{n+1}$ is assigned at the end, let $S$ be the set of houses assigned
in $\A^{n+1}$. Since $\A^{n+1}$ is complete, $|S|=n$, and
$h_{n+1}\in S$. Write the qualities of the houses in $S$ as $q_1\le \cdots \le q_{n-1}\le q_n$.
Then $q_n=M$ and $q_{n-1}\le x_n$. By Lemma~\ref{lem:common-quality-support}, applied with
$\alpha_r=r$, we have
\begin{equation*}
\mu(S) =\sum_{r=1}^{n-1} r^2(q_{r+1}-q_r) \ge (n-1)^2(M-q_{n-1}) \ge M-x_n > \mu(H^n)=\mu(B_{n+1}).
\end{equation*}
Since $\A^{n+1}$ is complete and $\mathrm{supp}(\A^{n+1})=S$, the definition
of $\mu$ gives us
\[
\mathrm{sub}(\A^{n+1})\ge \mu(S)>\mu(H^n)=\mu(B_{n+1}),
\]
contradicting the assumed guarantee.

If $h_{n+1}\notin \mathrm{supp}(\A^{n+1})$, then $\A^{n+1}=\A^n$.
For $k=0$, this follows from Definition~\ref{def:k-flexibility}: after round $n$ no
agent is unassigned, so a non-rejection $0$-flexible update is
impossible. For $k\ge 1$, every non-rejection chain update assigns
$h_{n+1}$ to the first agent in the chain by Definition~\ref{def:chain-update}. Thus, if
$h_{n+1}$ is not assigned, the update must be a rejection.

By the supposition that $a_i^n\ne h_i$ for some $i$, we have $E(\A^n)\neq \{(i,h_i):i\in N\}$.
Since $\{(i,h_i):i\in N\}$ is the unique maximum-weight perfect
matching between $N$ and $H^n$, $E(\A^n)$ is not a
maximum-weight perfect matching between $N$ and $H^n$. By
Lemma~\ref{lem:fixed-support}, $\A^n$ is not envy-freeable. Hence $\mathrm{sub}(\A^{n+1})=\mathrm{sub}(\A^n)=\infty$,
contradicting $\mathrm{sub}(\A^{n+1})=\mu(B_{n+1})$.

Therefore, after the first $n$ rounds, $a_i^n=h_i$ for every $i\in N$.

Alternatively, after the same first $n$ rounds, the adaptive
adversary may release a different final house $h_{n+1}$ with score $2$, and set $u_i(h_{n+1})=i x_{n+1}$ for every $i \in N$.

The utilities and scores observed during the first $n$ rounds are
the same as before. Since $A$ is deterministic, its round-$n$
allocation is still $a_i^n=h_i$ for every $i\in N$.
For this continuation, the final top $n$ houses are $B_{n+1}=\{h_2,h_3,\dots,h_{n+1}\}$.
Their qualities are $1,\,1+\delta,\,\dots,\,1+(n-1)\delta$.
By Lemma~\ref{lem:common-quality-support},
\[
\mu(B_{n+1})
=
\delta \sum_{r=1}^{n-1} r^2
<1.
\]

Now let $S\subseteq H^{n+1}$ be any $n$-element set with
$S\neq B_{n+1}$. Since $B_{n+1}=H^{n+1}\setminus\{h_1\}$, the set
$S$ contains $h_1$. If $q_1\le \cdots \le q_n$
are the qualities of the houses in $S$, then $q_1=0$ and $q_2\ge 1$.
By Lemma~\ref{lem:common-quality-support}, applied with $\alpha_r=r$, we have
\[
\mu(S)=\sum_{r=1}^{n-1} r^2(q_{r+1}-q_r)
      \ge q_2-q_1
      \ge 1.
\]
Hence, if a complete allocation $a$ satisfies
$\mathrm{supp}(a)\ne B_{n+1}$, then, by the definition of $\mu$,
\[
\mathrm{sub}(a)\ge \mu(\mathrm{supp}(a))\ge 1>\mu(B_{n+1}).
\]
Therefore, any complete allocation $a$ with $\mathrm{sub}(a)=\mu(B_{n+1})$
must satisfy $\mathrm{supp}(a)=B_{n+1}$. Since this value is finite, $a$ is
envy-freeable. By Lemma~\ref{lem:fixed-support}, $E(a)$ must be a maximum-weight
perfect matching between $N$ and $B_{n+1}$. Applying the same
exchange argument to the strictly increasing qualities
$x_2<\cdots<x_{n+1}$, this matching is unique: $E(\A)=\{(i,h_{i+1}):i\in N\}$.
Equivalently, $a_i=h_{i+1}$ for every $i\in N$.

Starting from $a_i^n=h_i$ for every $i\in N$, the allocation
$a_i^{n+1}=h_{i+1}$ for every $i\in N$ cannot be reached by
rejecting $h_{n+1}$. If $k=0$, Definition~\ref{def:k-flexibility} does not allow a
non-rejection update, because all agents are already assigned after
round $n$. Suppose $k\ge 1$. In any chain update reaching this
allocation, Definition~\ref{def:chain-update} forces agent $n$ to be first, because only
the first agent in the chain can receive $h_{n+1}$. It then forces
agent $n-1$ to be second, because $h_n$ was held by agent $n$
before the update. Continuing in the same way, the chain must be $n,n-1,\dots,1$,
and $h_1$ is discarded at the end. Therefore, the required
non-rejection update has length $n$.

Since $k<n$, algorithm $A$ cannot perform this update at
round $n+1$. Therefore, its final allocation cannot satisfy $\mathrm{sub}(\A^{n+1})=\mu(B_{n+1})$,
contradicting the assumed guarantee. Thus, no deterministic valid
$k$-flexible online algorithm with $k<n$ can ensure
$\mathrm{sub}(\A^m)=\mu(B_m)$ for every score prediction,
even when $m=n+1$.
\end{proof}

\subsection{Proof of Theorem~\ref{thm:wrong-prediction-loss}}

\wrongpredictionloss*

\begin{proof}
We first prove the cutoff-prediction case. Use identical utilities
$u_i(h)=q_h$ for every $i\in N$ and every house $h$; equivalently,
in the common-quality notation, $\alpha_i=1$ and $\beta_i=0$ for
every $i$. Thus, qualities and common values coincide. The
score-prediction case is handled at the end. The proof constructs a common prefix after which the adaptive adversary
chooses between two continuations: one in which the prediction is
correct and forces a specific support, and one in which the same
prediction becomes incorrect.

Fix $\gamma>0$. Choose $\eta>0$ such that
\[
\eta <
\min\left\{
\frac{1}{n-1},
\frac{2}{(n-1)(n+2)},
\frac{2\gamma}{(n-1)(3n-2)}
\right\}.
\]
Then
\begin{equation*}
(n - 1)(1 - (n - 1)\eta) -
\frac{\eta n(n - 1)}{2} = n-1-\frac{(n-1)(3n-2)}{2}\eta > n-1-\gamma, \quad 1-(n-1)\eta>0, \quad \text{and}
\end{equation*}
\[
\frac{\eta n(n-1)}{2}<1-(n-1)\eta.
\] 
Then choose $0<\varepsilon<
\min \{
\frac{\eta n}{2},
\frac{1-(n-1)\eta}{n}
\}$,
Thus
\[
(n - 1)\varepsilon <
\frac{\eta n(n - 1)}{2}
\quad\text{and}\quad
n\varepsilon < 1 - (n - 1)\eta.
\]
Give the algorithm the cutoff prediction $\tau=\varepsilon$.

The adaptive adversary first releases a common prefix consisting of
$n-1$ houses of quality $0$, followed by $n$ houses with qualities $1,\,1-\eta,\,1-2\eta,\,\dots,\,1-(n-1)\eta$.

Let $t_0\coloneq2n-1$. By validity and the convention in Section~\ref{sec:recourse}, $|\mathrm{supp}(\A^{t_0})|=n$.

After observing the allocation produced by $A$ on the common prefix,
the adaptive adversary may continue by releasing one additional house
of quality $\varepsilon$. Let $\I^+$ denote the resulting instance. Let $S_L$ be the $n$-house
set consisting of the $n-1$ zero-quality houses and the
$\varepsilon$-quality house. Let $S_H$ be the $n$-house set consisting
of the $n$ high-quality houses. By the identical-utilities formula $\mu(S)=\sum_{r=1}^n (y_n-y_r)$,
where $y_1\le \cdots \le y_n$ are the common values in $S$, we have $\mu(S_L)=(n-1)\varepsilon$
and $\mu(S_H)=\eta\sum_{r=0}^{n-1}r
=\frac{\eta n(n-1)}{2}$.
By the choice of $\varepsilon$, we have $\mu(S_L) < \mu(S_H)$.

Now, consider any $n$-house set $S$ different from both $S_L$
and $S_H$. Since $S\neq S_L$ and $S$ has size $n$, it cannot
consist only of the $n$ houses of quality at most $\varepsilon$. Hence
it contains at least one high-quality house. Similarly, since
$S\neq S_H$, it contains at least one house of quality at most
$\varepsilon$. Therefore, $S$ contains one house of quality at most
$\varepsilon$ and one house of quality at least $1-(n-1)\eta$. If the common values in $S$ are
$y_1 \le \cdots \le y_n$, then by \eqref{eqn:mu_1},
\[
\mu(S)=\sum_{r=1}^n (y_n-y_r) \ge y_n-y_1
        \ge 1-(n-1)\eta-\varepsilon .
\]
The inequality $n\varepsilon<1-(n-1)\eta$ gives us
\[
1-(n-1)\eta-\varepsilon > (n-1)\varepsilon=\mu(S_L).
\]
Hence, every $n$-house set other than $S_L$ has larger
$\mu$-value than $S_L$, so $S_L$ is the unique
$n$-house set minimizing $\mu$ in the continued instance.

Moreover, $C(\tau)=S_L$, and the largest quality in $C(\tau)$ is
exactly $\tau=\varepsilon$. Thus, $\tau$ is correct for the continued
instance. By exactness, if $\A^+$ is the final allocation returned by
$A$ on the continued instance $\I^+$, then $\mathrm{sub}(\A^+) = \mathrm{OPT}(\I^+)$.
Since $A$ is valid, $\A^+$ is complete. Hence
$\mathrm{supp}(\A^+)$ is an $n$-element set of houses.
By the definition of $\mu$, $\mu(\mathrm{supp}(\A^+)) \le \mathrm{sub}(\A^+)$.
Also, since $\mathrm{OPT}(\I^+)$ is the minimum subsidy
over all complete allocations, $\mathrm{OPT}(\I^+) \le \mu(\mathrm{supp}(\A^+))$.
Therefore
\[
\mathrm{OPT}(\I^+) \le
\mu(\mathrm{supp}(\A^+)) \le
\mathrm{sub}(\A^+) =
\mathrm{OPT}(\I^+),
\]
so $\mu(\mathrm{supp}(\A^+))=\mathrm{OPT}(\I^+)$.
Since $S_L$ is the unique $n$-house set minimizing $\mu$,
we get $\mathrm{supp}(\A^+) = S_L$.

Thus, all $n-1$ zero-quality houses are assigned at the end
of the continued instance. By Section~\ref{sec:recourse}, $D^t = H^t \setminus \mathrm{supp}(\A^t)$ and  $D^t \subseteq D^{t'}$ for all $t' \ge t$.
Therefore, in the run on $\I^+$, if one of the zero-quality houses
were not in $\mathrm{supp}(\A^{t_0})$ after the common prefix,
then it would belong to $D^{t_0}$ and could not be assigned later.
Since all zero-quality houses belong to $\mathrm{supp}(\A^+)$,
all $n-1$ zero-quality houses must already belong to
$\mathrm{supp}(\A^{t_0})$.

Alternatively, the adaptive adversary may terminate the sequence after
the common prefix. Let $\I$ denote this stopped instance, with
the same cutoff prediction $\tau=\varepsilon$. The algorithm has seen
the same prefix and the same prediction, so by determinism it has the
same allocation after round $t_0$. Therefore, all $n-1$ zero-quality
houses belong to its final assigned set. Since the allocation is complete, exactly one high-quality
house is also assigned, and that house has quality at least
$1-(n-1)\eta$. The assigned common values therefore consist of $n-1$ zeros
and one value $y\ge 1-(n-1)\eta$. Let $S=\mathrm{supp}(A(\I,\tau))$.
By the definition of $\mu$, $\mathrm{sub}(A(\I,\tau))\ge \mu(S)$.
Using~\eqref{eqn:mu_1}, we get that
\[
\mu(S)=(n-1)y\ge (n-1)(1-(n-1)\eta).
\]
Consequently,
\[
\mathrm{sub}(A(\I,\tau))\ge (n-1)(1-(n-1)\eta).
\]

For the stopped instance $\I$, the $n$ high-quality houses
form the set $S_H$, and $\mu(S_H)=\frac{\eta n(n-1)}{2}$
Any other $n$-house set contains at least one zero-quality
house. Since there are only $n-1$ zero-quality houses, it also
contains at least one high-quality house. Hence, by~\eqref{eqn:mu_1}, its
$\mu$-value is at least $1-(n-1)\eta$, which is larger than $\eta n(n-1)/2$ by the choice of $\eta$.
Thus, $S_H$ is the unique $n$-house set minimizing $\mu$ in
the stopped instance, and
\[
\mathrm{OPT}(\I)=\mu(S_H)=\frac{\eta n(n-1)}{2}.
\]
Also, since $n\varepsilon < 1-(n-1)\eta$, we have $\varepsilon < 1-(n-1)\eta$.
Thus, every high-quality house has quality greater than $\tau=\varepsilon$.
The stopped instance therefore contains only the $n-1$ zero-quality
houses with quality at most $\tau$. Thus
$C(\tau)$ is not defined for the stopped instance, so the cutoff
prediction $\tau$ is incorrect. Combining the bounds gives
\begin{equation*}
    \mathrm{sub}(A(\I,\tau))-\mathrm{OPT}(\I) \ge
(n-1)(1-(n-1)\eta)-\frac{\eta n(n-1)}{2}\ge n-1-\gamma.
\end{equation*}

For score predictions, use the same arrival sequence. Assign
score $1$ to each zero-quality house and score $0$ to each
high-quality house. In the continued instance, assign score $1$
to the final house of quality $\varepsilon$. Then the final top
$n$ houses in the continued instance are exactly $S_L$:
there are exactly $n$ houses with score $1$. Since $S_L$ is
the unique $n$-house set minimizing $\mu$, the score prediction
is correct on the continued instance. By exactness and the same
argument used above, $\mathrm{supp}(\A^+) = S_L$.
The discarded-house rule then implies that all $n-1$
zero-quality houses belong to $\mathrm{supp}(\A^{t_0})$
after the common prefix.

Now stop the instance after the common prefix, with the same scores
for all houses that have arrived. Since score predictions are observed
when houses arrive, the algorithm receives exactly the same information
through round $t_0$ in the stopped and continued runs. By determinism,
it therefore has the same allocation after round $t_0$. Hence, its
final assigned houses contain all $n-1$ zero-quality houses and
one high-quality house. In the stopped instance, the final top
$n$ houses under the score prediction consist of the $n-1$
zero-quality houses and one high-quality house, whereas the
unique $n$-house set minimizing $\mu$ is $S_H$. Thus, the final top $n$ houses under the score prediction do not form
the unique $n$-house set minimizing $\mu$, so the score prediction is
incorrect. The stopped score-prediction instance has the same house values as
the stopped cutoff-prediction instance. The algorithm's final assigned
set again contains all $n-1$ zero-quality houses and one high-quality
house, while
\[
\OPT(\I)=\mu(S_H)=\frac{\eta n(n-1)}{2}.
\]
Thus, the same lower bound gives
\begin{equation*}
    \mathrm{sub}(A(\I,p))-\OPT(\I) \ge
(n-1)(1-(n-1)\eta)-\frac{\eta n(n-1)}{2} \ge n-1-\gamma.
\end{equation*}
This proves the theorem.

\end{proof}

\subsection{A Small Underestimate Can Be Costly}
\label{sec:underestimate-costly}

\begin{proposition}\label{prop:cutoff-underestimate}
Fix $n\ge 2$. For every $\bar{\varepsilon}>0$ and every $\gamma>0$, there is a normalized identical-utilities instance with a correct cutoff $\tau^*$ and a prediction $\tau<\tau^*$ such that $0<\tau^*-\tau\le \bar{\varepsilon}$ and the cutoff algorithm returns an allocation $\A^m$ with $\mathrm{sub}(\A^m)-\mathrm{OPT}(\I)\ge n-1-\gamma$.
\end{proposition}

\begin{proof}
Choose $0<\varepsilon<
\min\{
\bar\varepsilon,\,
\frac12,\,
\frac{\gamma}{2(n-1)}
\}$.
Then choose $\delta>0$ such that
\[
n\delta<\varepsilon
\quad\text{and}\quad
\frac{\delta n(n-1)}{2}
<
\min\left\{
\varepsilon,\,
1-\varepsilon,\,
\frac{\gamma}{2}
\right\}.
\]
Consider identical utilities with the following house qualities:
\[
\underbrace{0,\dots,0}_{n-1\text{ houses}},\quad
1-\varepsilon,\quad
1-(n-1)\delta,\;1-(n-2)\delta,\;\dots,\;1-\delta,\;1.
\]
Let $\tau^*=1$ and $\tau=1-\varepsilon$.

Let $S_H$ denote the $n$ highest-quality houses: $1-(n-1)\delta,\;1-(n-2)\delta,\;\dots,\;1-\delta,\;1$.
By Lemma~\ref{lem:common-quality-support},
\[
\mu(S_H)=\frac{\delta n(n-1)}{2}.
\]

Now let $S$ be any other $n$-house set. If $S$ contains a zero-quality house, then it also contains at least one
positive-quality house, whose quality is at least $1-\varepsilon$. Hence,
by \eqref{eqn:mu_1}, its subsidy is at least $1-\varepsilon>\mu(S_H)$.

Otherwise, $S$ contains no zero-quality house. Since
$S\neq S_H$, the set $S$ contains the house of quality
$1-\varepsilon$ and omits one house from $S_H$.

First suppose that the omitted house has quality $q<1$.
Then the quality-$1$ house remains in $S$, so the maximum
quality is unchanged. Using $\mu(T)=n\max_{h\in T}u(h)-\sum_{h\in T}u(h)$
for identical utilities, we obtain
\[
\mu(S)-\mu(S_H)
=q-(1-\varepsilon)
\ge
\varepsilon-(n-1)\delta
>0.
\]

Now suppose that the omitted house has quality $1$.
Then the maximum quality in $S$ is $1-\delta$. A direct
calculation gives
\[
\mu(S)-\mu(S_H)
=
\varepsilon-n\delta
>0.
\]

Therefore, every $n$-house set $S\neq S_H$ has
$\mu(S)>\mu(S_H)$. Hence $S_H$ is the unique
minimum-subsidy $n$-house set. In particular,
$\tau^*=1$ is correct and $\OPT(\I)=\frac{\delta n(n-1)}{2}$.

With prediction $\tau=1-\varepsilon$, the cutoff algorithm keeps the house of quality $1-\varepsilon$ and the $n-1$ zero-quality houses. Its subsidy is $(n-1)(1-\varepsilon)$.
Therefore
\begin{equation*}
\mathrm{sub}(\A^m)-\mathrm{OPT}(\I) =
(n-1)(1-\varepsilon)
-\frac{\delta n(n-1)}{2} >
n-1-\frac{\gamma}{2}-\frac{\gamma}{2}
=
n-1-\gamma. \qedhere
\end{equation*}
\end{proof}

\end{document}